\documentclass[12pt]{article}
\usepackage{amsmath,amsthm,amssymb,amsfonts,fullpage,url,soul,natbib,thmtools,listings}
\usepackage[dvipsnames]{xcolor}
\usepackage[normalem]{ulem}
\usepackage{booktabs,colortbl,cancel,multirow,titling}
\usepackage{algorithm,bbm,graphicx}
\usepackage{hyperref}
\usepackage[]{algpseudocode}
\usepackage[titletoc,toc,title]{appendix}
\usepackage[T1]{fontenc}
\usepackage{pifont}

\makeatletter
\newcommand{\leqnos}{\tagsleft@true\let\veqno\@@leqno}
\newcommand{\reqnos}{\tagsleft@false\let\veqno\@@eqno}
\reqnos
\makeatother

\usepackage{arydshln}

\def\lQ{\scalebox{-1}[1]{''}}

\newcommand{\mc}[2]{\multicolumn{#1}{c}{#2}}
\definecolor{Gray}{gray}{0.85}
\definecolor{LightCyan}{rgb}{0.88,1,1}
\newcolumntype{a}{>{\columncolor{Gray}}c}
\newcolumntype{b}{>{\columncolor{white}}c}

\def\max{\text{max}}
\def\min{\text{min}}
\def\min{\text{min}}
\def\var{\text{var}}

\def\diag\text{diag}
\def\thetab{\theta_{b}}

\def\thetabf{\theta_{b+1}}
\def\gd{\textsc{gd}}
\def\nr{\textsc{nr}}
\def\dmk{\textsc{dmk}}
\def\ks{\textsc{ks}}
\def\rnr{r\textsc{nr}}
\def\rqn{rq\textsc{n}}

\def\rgd{r\textsc{gd}}
\def\sgd{s\textsc{gd}}
\def\snr{s\textsc{nr}}

\def\FL{\textsc{re}}

\def\bar{\overline}
\def\pk{P_k}
\def\pb{P_b}

 \newtheorem{theorem}{Theorem}
 
 \newtheorem{proposition}{Proposition}
 \newtheorem{lemma}{Lemma}
 \newtheorem{corollary}{Corollary}
 \newtheorem{assumption}{Assumption}

 \newcommand{\neutralize}[1]{\expandafter\let\csname c@#1\endcsname\count@}

\begin{document}
\title{{\bf Estimation and Inference by Stochastic Optimization}}
\author{Jean-Jacques Forneron\thanks{Department of Economics, Boston University, 270 Bay State Rd, MA 02215 Email: jjmf@bu.edu
\newline  I warmly thank Serena Ng for her many insights and feedback. I would also like to thank Tim Christensen, Iv\`an Fern\`andez-Val, S\'ilvia Gon\c{c}alves, Jessie Li, Elie Tamer, participants of the  Optimization-Conscious Econometrics Conference held at the University of Chicago, the World Congress of the Econometric Society, the 2021 New York Camp Econometrics, and the 2021 ASSA meetings, and seminar participants at UC Santa-Cruz, BU/BC, Georgetown, Tilburg University, Yale, UPenn, and HKUST for many helpful comments.  This paper and \citet{jjng-aerpp:21} supersede the manuscript ``Inference by Stochastic Optimization: A Free-Lunch Bootstrap'' \citep{jjng-rnr}.}
}
\date{\today}
\maketitle
\begin{abstract} 
In non-linear estimations, it is common to assess sampling uncertainty by bootstrap inference. For complex models, this can be computationally intensive.  
This paper combines optimization with resampling: turning stochastic optimization into a fast resampling device. Two methods are introduced: a resampled Newton-Raphson (\rnr) and a resampled quasi-Newton (\rqn) algorithm. Both produce draws that can be used to compute consistent estimates, confidence intervals, and standard errors \textit{in a single run}. The draws are generated by a gradient and Hessian (or an approximation) computed from batches of data that are resampled at each iteration. The proposed methods transition quickly from optimization to resampling when the objective is smooth and strictly convex. Simulated and empirical applications illustrate the properties of the methods on large scale and computationally intensive problems. Comparisons with frequentist and Bayesian methods highlight the features of the algorithms.


\end{abstract}


\setcounter{page}{0}
\thispagestyle{empty}
\baselineskip=18.pt
\newpage

\section{Introduction}


Many empirical economic analyses involve the process of estimating parameters $\theta$ from a sample of $n$ observations and assessing sampling uncertainty. The typical routine is to first produce a consistent M-estimator $\hat\theta_n$ of the true $\theta^\dagger$ by minimizing the sample objective $Q_n$. Then, standard errors and confidence intervals are computed to perform inferences. Because the sandwich estimator is harder to implement for more complicated models, it is very common to use a bootstrap at that stage. For more computationally intensive nonlinear estimations, it is common to report standard errors based on few bootstrap draws, e.g. 20, 25, or 50.\footnote{71 of all papers published in the American Economic Review between 2012 and 2016 relied on bootstrap inference, 49 of them report bootstrap standard errors. More than one in five of those 49 papers use $100$ bootstrap replications or fewer to compute standard errors. Survey conducted by the author, details available upon request.} The main bottleneck is that each replication requires another estimation on resampled data, which is nontrivial for these complex models. There is a long-standing interest in finding shortcuts to relieve this computational burden without sacrificing too much accuracy. Examples include \citet{Davidson1999}, \citet{Andrews2002}, \citet{Kline2012}, and more recently \citet{Honore2017}. These methods compute standard errors taking a converged estimate $\hat\theta_n$ as given. As such, estimation always precedes inference. 

This paper proposes to combine resampling with optimization to produce estimates, confidence intervals, and standard errors in a single step. This is very close in spirit to Markov-chain Monte Carlo (MCMC) methods that produce draws from a Bayesian posterior to be used for both estimation and inference. Unlike MCMC, the information matrix equality is not required for valid inference. The main loop of the Algorithm is a stochastic optimization routine which evaluates the objective on batches of resampled data with size $m\leq n$. Starting from an initial guess $\theta_0$, one updates $\theta_b$ to $\theta_{b+1}$, $b\geq 0$, using the resampled gradient and inverse Hessian (or an approximation) as a conditioning matrix, with a fixed learning rate $\gamma$. The Algorithm can also be used with reweighted samples of data. For a suitable choice of $\gamma$ and a strictly convex objective $Q_n$, it is shown that the average over $B$ draws $\theta_b$ is equal to $\hat\theta_n$ up to order $1/\min(m,B)$. Even with $m$ and $B \ll n$, the estimates are consistent and asymptotically normal as long as $m$ and $B \gg \sqrt{n}$. The results also show that the distribution, conditional on the sample, of $\sqrt{m}(\theta_b-\hat\theta_n)$ is first-order equivalent to that of $\sqrt{n}(\hat\theta_n-\theta^\dagger)$, up to scale, and standard errors are consistent. Inference results are valid for a resampled Newton-Raphson (\rnr) and a resampled quasi-Newton (\rqn) algorithm based on a new quasi-Newton update. For other conditioning matrices, the draws can produce consistent estimates but they need not be valid for inference.

The main theoretical insight of the paper is that \rnr\, and \rqn\, draws are well approximated by a simple AR(1) process when the resampling size $m$ is sufficiently large. The process is centered around the full sample estimate $\hat\theta_n$, and its innovations have variance proportional to the sandwich formula. Hence, after burning-in the first couple of iterations, the distribution of the draws can be used as a bootstrap distribution -- up to a simple adjustment which accounts for the Markov-chain properties of the algorithms. The main results assume the sample and resampled objective functions are smooth and strictly convex. Additional results and illustrations are given in the Supplement for some non-strictly convex settings. In particular, the algorithms can perform well in situations where their non-resampled counterparts are typically challenged and return inconsistent estimates.  This is another advantage of combining optimization with resampling. 


After introducing the algorithms, a small empirical example illustrates the properties of the methods. Then one simulated and two empirical examples illustrate different ways in which they can be applied in empirical work. First, the Monte-Carlo example considers the estimation of a dynamic discrete choice model with heterogeneous agents. Estimates are significantly biased because of non-linearity, resulting in poor inferences due to size distortions. A bootstrap bias correction would be computationally intensive. \rqn, combined with a simple split-panel bias-correction technique, produces accurate estimates and valid inferences in reasonable time. Next, using data from \citet{helpman2008},  \rnr\, and \rqn\, are benchmarked against MCMC and stochastic gradient-descent (\sgd) on a large-scale probit. \rnr\, and \rqn\, converge faster and have better mixing so that fewer draws are needed. Allowing for dependence to compute cluster robust standard errors is straightforward, which is not the case with MCMC. Finally, a more complex model is considered. A replication of \citet{donaldson2018} shows how a computationally intensive grid search used for both estimation and bootstrapping on a 100-core cluster environment can be replaced with \rnr, running in less than six hours on a desktop computer. The time saved allows to explore other model specifications. The results also highlight identification issues which are easily diagnosed with \rnr. Additional empirical and simulated examples are given in a companion paper \citep{jjng-aerpp:21}. Finally, while the results in this paper are targeted at a class of smooth and convex M-estimation problems, the challenges of non-smooth and non-convex estimation constitute a serious impediment to empirical research. Another companion paper provides algorithms and theoretical results for finite-sample generalized method of moments estimation without assuming smoothness or convexity \citep{Forneron2022}.

\paragraph{Outline of the paper.}  Section \ref{sec:classical} begins with a review of classical and stochastic optimization. The proposed algorithms are presented in Section \ref{sec:algorithm}, related methods are discussed. A simple analytical and empirical example highlight the theoretical aspects of the paper. Section \ref{sec:rnr1} provides the theoretical results.  Section \ref{sec:examples} covers the simulated and empirical examples. Appendices \ref{apx:inter}-\ref{apx:add_ex} provide proofs, sample R codes, and primitive conditions for some of the assumptions. Issues related to non-convexity are covered in Appendix \ref{sec:discussion}. 

\section{Setting and Overview of Gradient-Based Optimization} \label{sec:classical}
Consider minimizing a twice-differentiable and convex sample objective function $Q_n(\theta)$ with respect to parameters $\theta \in \Theta \subset \mathbb{R}^{d_\theta}$:
\[ \hat\theta_n = \text{argmin}_{\theta \in \Theta} Q_n(\theta),\quad Q_n(\theta) = \frac{1}{n} \sum_{i=1}^n q(z_i;\theta), \]
using a sample of $n$ iid observations $z_i = (y_i,x_i)$. Examples include non-linear least-squares (NLS) for which $q(z_i;\theta) = [y_i - f(x_i;\theta)]^2$ where $f$ is the regression function, known up to the parameters of interest $\theta$. For (pseudo) maximum likelihood estimation (MLE), $q(z_i;\theta) = -\ell(y_i;\theta|x_i)$ is the negative of the log-likelihood. The true value, denoted by $\theta^\dagger$, minimizes the limiting objective $Q = \text{plim}_{n\to \infty} Q_n$. 
Under the conditions in Theorem 2.1 of \citet{newey-mcfadden-handbook},  $\hat\theta_n$ is consistent for $\theta^\dagger$. If, in addition, the assumptions in Theorem 3.1 of \citet{newey-mcfadden-handbook} hold, then $\hat\theta_n$ is also asymptotically normal:
 \[\sqrt{n}V_n^{-1/2}(\hat\theta_n-\theta^\dagger)\overset{d}{\to}  N(0,I_d), \]
where $V_n =[H_n(\hat\theta_n)]^{-1} \Sigma_n [H_n(\hat\theta_n)]^{-1}$ is a consistent estimator of the sandwich variance, $[H_n(\hat\theta_n)]^{-1} \overset{p}{\to} [H(\theta^\dagger)]^{-1}$ approximates the bread, and $\Sigma_n \overset{p}{\to} \lim_{n\to\infty}\text{var}[\sqrt{n}G_n(\theta^\dagger)]$ approximates the meat. Practitionners are often interested in conducting inference for a function of the parameters $h(\theta) \in \mathbb{R}$. Examples include inference on a single coefficient or a counterfactual based on the structural estimates $\hat\theta_n$. When $h$ is continuously differentiable, the delta-method gives the standard error formula $\text{se}[h(\hat\theta_n)] = \sqrt{ \nabla h(\hat\theta_n) V_n \nabla h(\hat\theta_n)^\prime /n}$ and the $95\%$ level confidence interval $h(\hat\theta_n) \pm 1.96 \times \text{se}[h(\hat\theta_n)]$. It is also common to compute confidence intervals from bootstrap draws using quantiles or the standard deviation of $h(\theta_n^{(b)})$, where $\hat\theta_n^{(b)}$ minimizes the resampled objectives $Q_n^{(b)}$. 

For most problems, the estimator $\hat\theta_n$ does not have closed-form and is usually obtained by numerical optimization. The methods below rely on the sample gradient and Hessian of $Q_n$, defined respectively by:
\begin{eqnarray*} 
G_n(\theta)=\nabla Q_n(\theta) = \frac{1}{n} \sum_{i=1}^n \nabla q(z_i;\theta),\quad 
H_n(\theta)=\nabla^2 Q_n(\theta) = \frac{1}{n} \sum_{i=1}^n \nabla^2 q(z_i;\theta).
\end{eqnarray*}  

\subsection{Classical Optimizers} \label{sec:classic}
When $Q_n$ is quadratic, as in OLS or IV regressions, the solution has closed form and $\hat\theta_n = \theta_0-[H_n(\theta_0)]^{-1} G_n(\theta_0)$, for any initial guess $\theta_0$. For general non-quadratic objectives, the solution does not have closed-form. A typical strategy is to consider a sequence of approximations:
\begin{align}
      Q_n(\theta) \simeq Q_n(\theta_k) + G_n(\theta_k)(\theta-\theta_k) + \frac{1}{2\gamma_k} (\theta-\theta_k)^\prime P_k^{-1} (\theta-\theta_k),\quad k=0,1,2,\dots \label{eq:approx_quad}
\end{align}
where $P_k$ is a symmetric positive definite conditioning matrix, $\gamma_k \in (0,1]$ is called the learning rate, it penalizes for the quality of the approximation (\ref{eq:approx_quad}). For $\gamma_k=1$ and $P_k^{-1} = H_n(\theta_k)$, the approximation is exact for $Q_n$ quadratic. Otherwise, it coincides with a second-order Taylor expansion around $\theta=\theta_k$, which is not exact. From a guess $\theta_k$, minimizing (\ref{eq:approx_quad}) leads to:
\begin{eqnarray}
   \theta_{k+1} &\equiv &\theta_k-\gamma_k  P_k G_n(\theta_k),\quad k=0,1,2,\dots \label{eq:gradient-rule}
 \end{eqnarray}
 Different choices of conditioning matrix $P_k$ lead to different algorithms. Newton-Raphson (\nr) iterations use $P_k = [H_n(\theta_k)]^{-1}$ and $\gamma_k=1$; a damped update uses $\gamma_k<1$. For strongly convex problems, \nr\, enjoys a fast quadratic convergence property, but it comes at the cost of computing both the gradient and the Hessian at each iteration $k$.\footnote{In statistical computing, the convergence of $\theta_k$ to $\hat\theta_n$ is said to be linear if $\|\theta_{k+1}-\hat\theta_n\|/\|\theta_k-\hat\theta_n\|^q<r$ for some $r\in (0,1)$ if $q=1$ and quadratic if $q=2$. Convergence is superlinear if $\lim_{k\rightarrow\infty} \|\theta_{k+1}-\hat\theta_n\|/ \|\theta_k-\hat\theta_n\|=0$. See \citet{boyd-vandenberghe:04} Section 9.3.1 for linear convergence of gradient methods, and \citet[Theorem 3.5]{nocedal-wright:06} for quadratic convergence of Newton's method when $\gamma=1$ or $\gamma_k \to 1$ at an appropriate rate.   `Damped Newton' updating with  $\gamma_k \in (0,1)$ has a linear convergence rate, see \citet{boyd-vandenberghe:04} Section 9.5.3 and \citet{nesterov2018} Section 1.2.4.}  
To avoid the cost of computing the $(d_\theta+1)d_\theta/2$ second-order derivatives in $H_n$. Other updates rely on different choices of $P_k$. These include quasi-Newton (q\textsc{n}) and gradient-descent (\gd). 

The most popular quasi-Newton update is \textsc{bfgs} for which $P_k$ approximates $[H_n(\theta_k)]^{-1}$ iteratively using only gradients. Although not as fast as \nr, convergence for quasi-Newton with an appropriate choice of $\gamma_k$ can be superlinear.\footnote{See \citet{dennis1977} and \citet[Ch6]{nocedal-wright:06} for an overview of quasi-Newton updates and their properties.}  Quasi-Newton methods bypass Hessian computations, but analytical convergence results are more difficult to obtain.  
The gradient-descent algorithm (\textsc{gd}) sets $P_k=I_d$, it only requires computing a single gradient $G_n(\theta_k)$.  Convergence, however, requires choosing $\gamma_k$ small enough that all eigenvalues of $I_d-\gamma_k H_n(\tilde \theta_k)$ lie in the restricted region $(-1,1)$, for some intermediate value $\tilde\theta_k$. This results in very slow convergence when the ratio of the largest to smallest eigenvalue of $H_n$ is large.

The necessary conditions for a local minimum are $\|G_n(\hat\theta_n)\|=0$ and $H_n(\hat\theta_n)$ positive semi-definite. The sufficient conditions are $\|G_n(\hat\theta_n)\|=0$ and $H_n(\hat\theta_n)$ positive definite.
Note that for non-convex objectives, any $\theta_k$ with $G_n(\theta_k)=0$ is a fixed-point of the update rule (\ref{eq:gradient-rule}). This includes local and global minima and maxima as well as locally suboptimal solutions, called saddle points, for which $H_n$ has both positive and negative eigenvalues. 

\subsection{Stochastic Optimizers} \label{sec:stochastic}
In their seminal paper on stochastic approximation, \citet{Robbins1951} consider the situation where only noisy gradients are available $G_n(\theta_{k})+e_k$ with $e_k$ iid, mean-zero. Using the notation above, they propose the gradient-descent update :
$ \theta_{k+1}=\theta_{k}- \gamma_{k} [G(\theta_{k})+e_k]$, and show under regularity conditions that $\theta_k \overset{a.s.}{\to} \hat\theta_n$ when $\gamma_k > 0$ satisfy
\begin{equation}
\label{eq:monro-robbins-conditions}
 \text{(i)}\quad \sum_{k=1}^\infty \gamma_k= +\infty, \quad\quad   \text{(ii)} \quad\sum_{k=1}^\infty \gamma_k^2<+\infty.
\end{equation}
The first condition ensures that all possible solutions will be reached with high probability regardless of the $\theta_0$ while the second ensures convergence to the true value. \citet{Kiefer1952Stochastic} extend the results to non-linear least squares with gradients computed by finite-differences. Feasible choices of learning rate schedules include $\gamma_k = \gamma_0 k^{-\delta}$, with $\delta \in (1/2,1]$ and $\gamma_0>0$ as choice parameters. Depending on $\delta$, convergence as measured by $\mathbb{E}(\|\theta_k-\hat\theta_n\|^2)$ can occur at a $1/k$ rate or slower. \citet{ruppert1988efficient} and \citet{polyak1992acceleration} show that the averaged values $\overline{\theta}_k = \frac{1}{k}\sum_{i=1}^k \theta_i$, converge at the fastest $1/k$ rate for all choices of $\delta \in (1/2,1]$. This is commonly refered to as Polyak-Ruppert averaging. Building on these results, modern stochastic gradient-descent updates are given by
\begin{eqnarray*}
       \theta_{k+1}&=&\theta_{k}- \gamma_k G_m(\theta_{k})
\end{eqnarray*}
where $G_m(\theta_k)=\frac{1}{m}\sum_{i=1}^m \nabla q(z_i^{(k)},\theta_k)$ is an unbiased estimate of $G_n(\theta_k)$ based on $m<n$ randomly chosen observations $z_1^{(k)},\dots,z_m^{(k)}$ drawn at each $k$. Setting $e_k = G_m(\theta_k) - G_n(\theta_k)$ results in the same setting as stochastic approximation above. Though $m=1$ is computationally inexpensive and is a popular choice, the stochastic gradients $G_m(\theta_k)$ are very noisy and a small $\gamma_k$ satisfying (\ref{eq:monro-robbins-conditions}) is needed to compensate. This results in slow convergence making it well suited for problems where the gradients $G_m$ are very cheap to computed and running very many iterations is feasible. Unfortunately, this is typically not the case for more complicated economic models. A number of methods can accelerate convergence,\footnote{Improvements to \sgd\, include  momentum \citep{polyak:64} and accelerated gradient-descent \citep{nesterov:83}.  As for \gd, using $P_k=I_d$ in \sgd\, can lead to very slow convergence. Popular extensions of stochastic gradient descent include \textsc{adagrad} \citep{duchi-acm:11}, and \textsc{adam} \citep{kingma2014}.} but they tend to converge slower than \nr, or \textsc{bfgs}. More closely related to this paper, some study stochastic \textsc{bfgs}, but compute the full sample $G_n$ many times which is more demanding.\footnote{See e.g. Algorithm 1, step 4, in \citet{pmlr-v51-moritz16}. Note that stochastic Newton-Raphson (\snr), with a small $m =1$, is not used in practice because $H_m^{(k)}$ is often not invertible when $m<d_\theta$.} 

Beyond reducing the computional cost, stochastic optimization can improve on classical methods in non-convex settings. For instance, suboptimal solutions aren't fixed points for \sgd\, \citep{ge2015escaping,jin2017escape}. Though \sgd\, requires $\gamma_k \to 0$, it is common to use $\gamma_k = \gamma$ fixed in practice. \citet{dieuleveut-aos:20} consider $Q_n$ quadratic with $\gamma$ and $m$ fixed and find that it leads to estimation bias which needs to be corrected. Few results are available for inference, notably \citet{chen-aos:20} consider the computation of standard errors for \sgd. This paper considers $\gamma$ fixed with $Q_n$ non-quadratic but requires $m \to \infty$. For inference, $P_k$ is either a Hessian or the new quasi-Newton update.  The connection with bootstrap inference and the AR(1) representation are new.

\section{Algorithms and Intuition} \label{sec:algorithm}
Because the iterations are used both for estimation and bootstrap inference, they will be indexed by $b=0,1,2,\dots$, rather than $k=0,1,2,\dots$ used for estimation only.
The main algorithm below relies on an resampled objective $Q_m^{(b)}$, and its derivatives $G_m^{(b)},H_m^{(b)}$, computed using either $m \leq n$ resampled observations, or using all $m=n$ observations but with random reweighting. Unless otherwise stated, $Q_m^{(b)}$ will refer to either a resampled or reweighted objective. Using these quantities, Algorithm \ref{algo:REI} describes a way to perform both estimation and inference for a function $h$ of $\theta$.

\begin{algorithm}[h] 
\caption{Resampled Estimation and Inference} \label{algo:REI} 
      \begin{algorithmic}
        \State 1) \textbf{Inputs} (a) an initial guess $\theta_{0}$, (b) a bootstrap sample size $B$ and a burn-in period \textsc{burn}, (c) a batch size $m\leq n$ such that $m/n\rightarrow c \in [0,1]$, and $\sqrt{n}/m\rightarrow 0$,   (d) a fixed learning rate $\gamma \in (0,1]$, (e) a resampling or re-weighting scheme. 
      \State 2) \textbf{Burn-in and Resample} 
      \For{$b=0,\dots,\textsc{burn}+B-1$} 
            \State Resample, or re-weight, the data, \Comment{see (\ref{eq:resample}), (\ref{eq:reweight})}
            \State Update the conditioning matrix $P_b$, \Comment{see Algorithms \ref{algo:rnr} and \ref{algo:rqn}}
            \State Compute $\theta_{b+1} = \thetab - \gamma P_b G_m^{(b+1)}(\thetab).$
      \EndFor
      \State 3) \textbf{Discard} the first \textsc{burn} draws, re-index $\theta_{\text{burn}+b}$ to $\thetab$ for $b=1,\ldots B$.
      \State 4) \textbf{Estimation}
      \begin{itemize} \setlength\itemsep{0em}
            \item Estimates: $\overline{\theta}_{\textsc{re}}=\frac{1}{B} \sum_{b = 1}^B \thetab$
      \end{itemize}
      \State 5) \textbf{Inference} (\rnr\, and \rqn) 
      \begin{itemize} \setlength\itemsep{0em}
            \item Standard errors: $\textsc{se}[h(\hat\theta_n)] = \sqrt{ \frac{m}{n\phi(\gamma)}\frac{1}{B} \sum_{b=1}^B [h(\theta_b)-h(\overline{\theta}_{\textsc{re}})]^2 }$,
            \item Confidence interval ($1-\alpha$ level): $\Big[c_{h,b}(\alpha/2),c_{h,b}(1-\alpha/2)\Big]$,
      \end{itemize}     
      where $\phi(\gamma)=\frac{\gamma^2}{1-(1-\gamma)^2}$, $c_{h,b}(\alpha)$ is the $\alpha$-th quantile of $h(\tilde\theta_b)-\overline{h}_B$, $\alpha \in (0,1)$, and $\tilde\theta_b = \overline{\theta}_{\textsc{re}} +  \sqrt{\frac{m}{n\phi(\gamma)}}(\theta_b-\overline{\theta}_{\textsc{re}})$.
      \end{algorithmic}
\end{algorithm}

Similar to MCMC, after discarding an initial burn-in period, estimates are computed by averaging over draws. Also, standard errors and confidence intervals are computed from the standard deviation and quantiles of the draws but here with an adjustment $\sqrt{\frac{m}{n\phi(\gamma)}}$ which accounts for $m$ potentially smaller than $n$, and $\phi(\gamma)$ adjusts for the Markov-chain properties of $\theta_b$ as explained below. Compared to \sgd, Algorithm \ref{algo:REI} requires $m \gg 1$ but allows for $m/n \to 0$ with the restriction that $\sqrt{n}/m \to 0$. This is the cost of using a fixed learning rate $\gamma$ rather than a decaying sequence $\gamma_k$ satisfying (\ref{eq:monro-robbins-conditions}), used in nearly all of the literature.

The learning rate $\gamma$ controls the optimization behaviour of the Algorithm, a larger value of $\gamma$ is associated with faster convergence. For quadratic $Q_n$, any $\gamma \in (0,1]$ is feasible, otherwise a smaller $\gamma < 1$ is needed, $\gamma \in [0.1,0.3]$ performs well in the Examples. The parameter $\gamma$ also affects the distributional properties of $\theta_b$ which can be more concentrated around $\hat\theta_n$ than conventional bootstrap draws. For instance, with $m=n$ and $\gamma=0.1$ we have $1/\phi(\gamma) = 19$: the variance of $\theta_b$ is $19$ times smaller than the variance of $\hat\theta_n$.\footnote{$\phi$ is decreasing in $\gamma$: the smaller $\gamma$, the more localized the Markov-chain will be. For comparison $\phi(1)=1$, $\phi(0.4)=4$, $\phi(0.2)=9$, and $\phi(0.01)=199$.} A number of different implementations are possible, a resampled Gradient-Descent (\rgd) would use the same $P_b = I_d$ for each iteration $b$. Valid inference, however, requires implementations for which $P_b$ approximates the inverse Hessian sufficiently well. This is the case for the resampled Newton-Raphson and resampled quasi-Newton algorithms. The first is conceptually simple and provides good intuition for the results. For most applications, the second is more computationally efficient and will be preferred.

Depending on the choice of resampling or reweighting scheme, $G_m^{(b)},H_m^{(b)}$ will be computed differently. First, the resampled objective, gradient, and Hessian are:
\begin{align} 
      &Q_m^{(b)}(\theta) = \frac{1}{m} \sum_{i=1}^m q(z_i^{(b)},\theta), \notag\\ &G_m^{(b)}(\theta) = \frac{1}{m} \sum_{i=1}^m \nabla q(z_i^{(b)},\theta),\quad H_m^{(b)}(\theta) = \frac{1}{m} \sum_{i=1}^m \nabla^2 q(z_i^{(b)},\theta), \label{eq:resample}
\end{align}
where $z_1^{(b)},\dots,z_m^{(b)}$ are drawn independently with replacement from the empirical data $(z_1,\dots,z_n)$ at each iteration $b$. Unlike \sgd, the resampling should match the dependence structure of the data. For instance, clustered data should be resampled at the cluster rather than individual level. Second, the re-weighted objective, gradient, and Hessian are:
\begin{align}
      &Q_n^{(b)}(\theta) = \frac{1}{n} \sum_{i=1}^n w_i^{(b)} q(z_i,\theta),\notag\\ &G^{(b)}_n(\theta) = \frac{1}{n} \sum_{i=1}^n w_i^{(b)} \nabla q(z_i,\theta),\quad H^{(b)}_m(\theta) = \frac{1}{n} \sum_{i=1}^n w_i^{(b)} \nabla^2 q(z_i,\theta), \label{eq:reweight}
\end{align}
where $w_i^{(b)}$ are iid with mean and variance equal to $1$. Examples include Gaussian weights $w_i^{(b)} \sim \mathcal{N}(1,1)$, exponential $w_i^{(b)} \sim \exp(1)$, or Poisson $w_i^{(b)} \sim \text{Pois}(1)$. This is known as a multiplier, or exchangeably weighted, bootstrap.\footnote{See \citet[Ch3.6.2]{VanderVaart1996} and \citet[Ch10]{kosorok2007}. The Poisson bootstrap approximates, with computational advantages, the standard $n$ out of $n$ bootstrap \citep{chamandy2012}.} Similar to resampling, the weights should vary only between clusters in clustered data, and more generally match the dependence of the data. Depending on the setting, reweighting can be preferred to resampling. For instance, if $\theta$ includes fixed effects, resampling can produce samples with no observation that is informative about one of the fixed effects.\footnote{This is also an issue for censored models where resampling can result in datasets with no censored observations. This also affects \sgd; with $m=1$ a  gradient is only informative about at one fixed effect at a time. This makes convergence very slow, see Appendix \ref{apx:add_ex2} for an illustration.} 

\subsection{Resampled Newton-Raphson (\rnr)}
The resampled Newton-Raphson update described in Algorithm \ref{algo:rnr} below can be used within Algorithm \ref{algo:REI} to produce consistent estimates and valid inferences.  
\begin{algorithm}[h] 
      \caption{Resampled Newton-Raphson (\rnr)} \label{algo:rnr} 
            \begin{algorithmic}
              \State Compute $P_b = [H_m^{(b+1)}(\theta_b)]^{-1}$
            \end{algorithmic}
\end{algorithm}
The implementation is very simple, as it only requires computing a Hessian. In many settings, however, repeatedly computing this matrix of second-order derivatives can be very demanding so that a quasi-Newton update is often preferred. Nevertheless, the \rnr\, implementation gives a good first intuition about why and how Algorithm \ref{algo:REI} produces estimates and valid inferences. This is illustrated with a simple analytical and a small empirical example below.


\paragraph{A pen and pencil example.} The following illustrates the main optimization and bootstrap properties of \rnr\, using a tractable OLS regression of $y_i$ on $x_i$: $Q_n(\theta) = \frac{1}{2n} \sum_{i=1}^n (y_i - x_i^\prime \theta)^2$. Using $m$ out of $n$ resampling, we have $G_m^{(b+1)}(\theta) = -\frac{1}{m} \sum_{i=1}^m x_i^{(b+1)} (y_i^{(b+1)} - x_i^{(b+1)\prime} \theta)$, and $H_m^{(b)}(\theta) = \frac{1}{m} \sum_{i=1}^m x_i^{(b+1)} x_i^{(b+1)\prime}$. In more compact notation: $G_m^{(b+1)}(\theta) = - \frac{1}{m} X_m^{(b+1)\prime} ( y_m^{(b+1)} - X_m^{(b+1)}\theta)$, $H_m^{(b+1)}(\theta) = \frac{1}{m} X_m^{(b+1)\prime} X_m^{(b+1)}$. For $b \geq 0$, one \rnr\, update takes the form:\footnote{Derivations for \rnr\, and other methods are given in Appendix \ref{apx:OLS}.}
\[ \theta_{b+1} = \theta_b + \gamma \left( X_m^{(b+1)\prime} X_m^{(b+1)} \right)^{-1} X_m^{(b+1)\prime} ( y_m^{(b+1)} - X_m^{(b+1)}\theta_b). \]
Notice that $\left( X_m^{(b+1)\prime} X_m^{(b+1)} \right)^{-1} X_m^{(b+1)\prime} y_m^{(b+1)} = \hat\theta_m^{(b+1)}$ is the bootstrap OLS estimate. Substract $\hat\theta_n$ on both sides and re-arrange terms to find that:
\[ \theta_{b+1} -\hat\theta_n = (1-\gamma)(\theta_b-\hat\theta_n) + \gamma ( \hat\theta_m^{(b+1)}-\hat\theta_n) \]
is an AR(1) process with conditionally iid innovations based on the centered bootstrap estimates: $\hat\theta_m^{(b)}-\hat\theta_n$. Notice that neither $\hat\theta_n$ nor $\hat\theta_m^{(b)}$ are computed explicitly in the \rnr\, update. Using a backwards recursion, we have:
\[ \theta_{b+1} -\hat\theta_n = \underbrace{(1-\gamma)^b(\theta_0-\hat\theta_n) \vphantom{\gamma \sum_{j=0}^{b} (1-\gamma)^j ( \hat\theta_m^{(b-j)}-\hat\theta_n)}}_{\text{Starting value bias}} + \underbrace{\gamma \sum_{j=0}^{b} (1-\gamma)^j ( \hat\theta_m^{(b+1-j)}-\hat\theta_n)}_{\text{Resampling noise}}  \]
Suppose, for simplicity, that under the bootstrap expectation $\mathbb{E}^\star(\hat\theta_m^{(b+1)})=\hat\theta_n$. Then the starting value bias due to $\theta_0 \neq \hat\theta_n$ is 
\[ \mathbb{E}^\star(\theta_{b+1}-\hat\theta_n) = (1-\gamma)^b(\theta_0-\hat\theta_n),\] which converges exponentially fast to zero. Very quickly, the bias becomes negligible and the average draw consistently estimates $\hat\theta_n$. A seamless transition occurs from optimization to resampling. Using the formula for the variance of an AR(1) process, we have 
\[ \text{var}^\star ( \theta_{b+1} ) = \gamma^2\frac{1- (1-\gamma)^{b+1}}{1-(1-\gamma)^2} \text{var}^\star(\hat\theta_m^{(b)}) \simeq \phi(\gamma) \text{var}^\star ( \hat\theta_{m}^{(b)} ),\] which also converges exponentially fast. The factor $\phi(\gamma) = \gamma^2/[1-(1-\gamma)^2]$ is used in Algorithm \ref{algo:REI}. Assuming $\text{var}^\star ( \hat\theta_{m}^{(b)} )$ consistently estimates standard errors, then so does \rnr\, with the suggested adjustment. The AR(1) representation only holds if $P_b$ approximates the inverse Hessian sufficiently well, this is not the case for \gd\, and several quasi-Newton updates.

\paragraph{Illustration with real data.} The main results extend the properties described for OLS above to more general non-quadratic strictly convex objectives. The main insight is that even though the AR(1) representation above does not hold exactly, there exists a sequence $\theta_b^\star$, called coupling, which follows an AR(1) process and is uniformly close to the Markov-chain $\theta_b$ up to an error of order $1/m$. A good approximation requires $m \to \infty$ but still allows for $m/n \to 0$.
\begin{figure}[ht] \caption{Probit Example:  draws for $\theta_{\text{educ}}$} \label{fig:mroz}
      \centering  
      \includegraphics[scale=0.75]{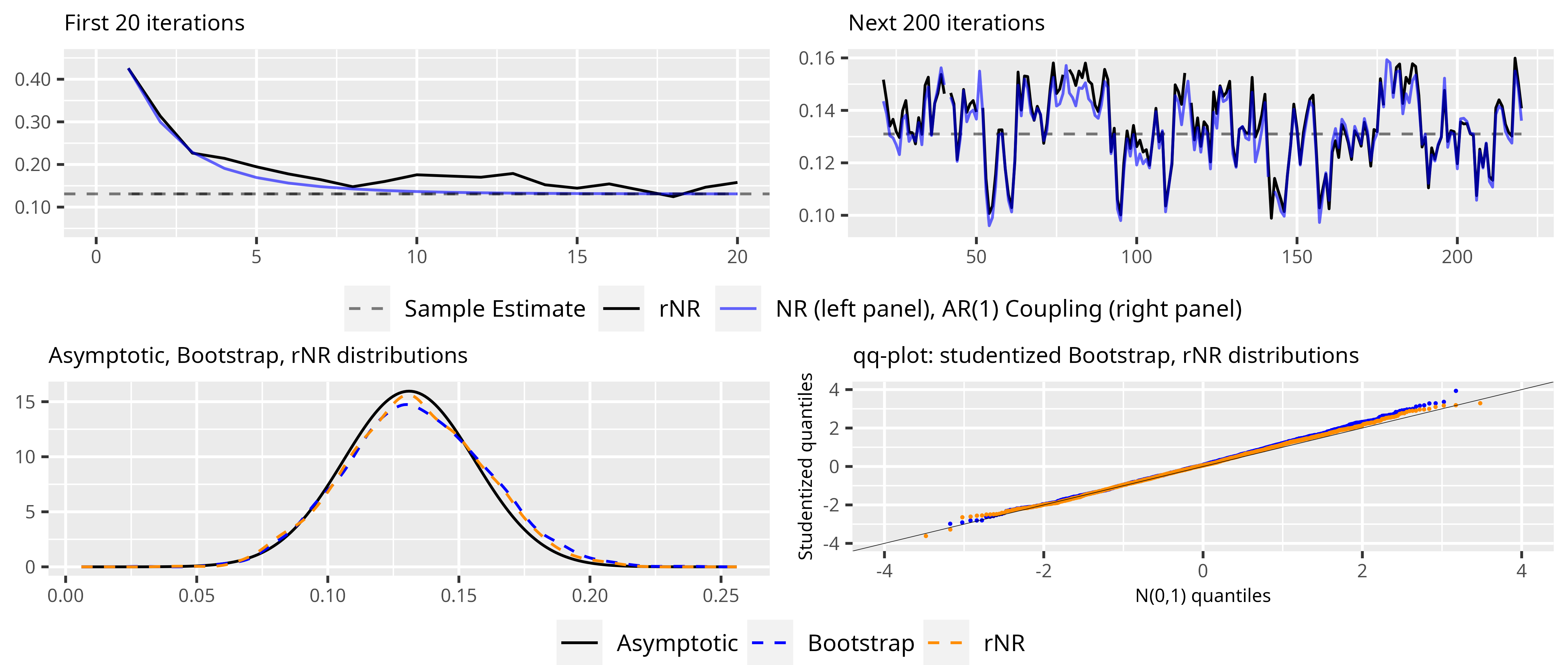}
\end{figure}
To illustrate this feature, Figure \ref{fig:mroz} shows the draws for the coefficient on education for a probit model fitted on data from \citet{mroz:87}. This example uses $\gamma=0.3$ and resamples $m=n/2$ out of $n=753$ observations with replacement. The top left panel shows the first $20$ iterations, which corresponds to the initial convergence phase. The resampled \rnr\, (black) and the deterministic \nr\, (blue) have very similar paths, moving quickly to the true value (dashed black). Then, the top right panel shows the next $200$ iterations -- the resampling phase. Even with $m<n$, the \rnr\, draws (black) are very close to the AR(1) process $\theta_b^\star$ (blue) used to prove the results. The bottom panel illustrates the bootstrap properties of the draws. The left panel compares a $\mathcal{N}(\hat\theta_n,V_n/n)$ with the standard $n$ out of $n$ bootstrap, and \rnr\, draws adjusted for $m$ and $\phi(\gamma)$. The distributions are very close. The right panel further confirms this with a more detailed q-q plot view. Replication code for this example is provided in Appendix \ref{apx:Rcode} with a detailed comparison of \rnr, \rqn, for different choices of $m$, with standard errors computed using the sandwich formula, the standard bootstrap, and methods by \citet{Davidson1999}, \citet{Kline2012}. Additional comparisons with these methods can be found in \citet{jjng-aerpp:21}.

\subsection{Resampled quasi-Newton (\rqn)}

In practice, quasi-Newton methods are often preferred to Newton-Raphson because the latter requires computing the full Hessian matrix many times, which can be costly.\footnote{In some settings, it is possible to compute the Hessian `for free' once the gradient is computed; see Example 3, Section \ref{sec:examples} and Appendix \ref{apx:Ex3}.} The main idea behind the quasi-Newton update in Algorithm \ref{algo:rqn} below is that the $d_\theta(d_\theta+1)/2$ unique elements in the Hessian can be recovered using only $d_\theta$ scalar derivatives. Consider the Hessian-vector product $y_b = H_m^{(b+1)}(\theta_b)s_b$, for a direction $s_b \neq 0$. Notice that $y_b = \lim_{\varepsilon \to 0} [G_m^{(b+1)}(\theta_b + \varepsilon s_b)-G_m^{(b+1)}(\theta_b - \varepsilon s_b)]/2\varepsilon$ is a scalar derivative which can be computed without the full matrix $H_m^{(b+1)}(\theta_b)$. The full Hessian matrix can be recovered from Hessian-vector products computed in $d_\theta$ linearly independent directions $s_b,s_{b-1},\dots$ by running a simple linear regression of $y_b$ on $s_b$. Using an overdetermined system with $L > d_\theta$ directions ensures the linear regression is well conditionned. In practice, the different $y_b$ are computed at different values of $\theta_b$ so that the Hessian recovered by least-squares need not be symmetric nor positive definite. A simple transformation is used to enforce these two features, while retaining its approximation properties. The resulting quasi-Newton update is new. The while loop ensures the regression is well-conditionned. It never ran in the examples of Section \ref{sec:examples} but should included as a safeguard. A more detailed discussion of difference with existing updates and the choice of tuning parameters is deferred to Section \ref{sec:corro}.

\begin{algorithm}[H] 
      \caption{Resampled quasi-Newton (\rqn)} \label{algo:rqn} 
            \begin{algorithmic}
            \State 1) \textbf{Inputs} (a) number of secants $L \geq d_\theta$, (b) cutoffs $\lambda_S>0$, $\underline{\lambda}>0$
            \State 2) \textbf{Least-Squares Approximation}
              \If{$b=0$}  \Comment{Initialization}
                  \State Set an initial guess $\hat H_0$, e.g. $H_m^{(0)}(\theta_0)$
                  \State Draw $s_j$, normalize $s_j = s_j/\|s_j\|_2$, $j = 0,\dots,-L+1$
                  \State Compute $y_j = \hat H_0 s_j$, $j=0\dots,-L+1$
              \Else     
                  \State Compute $s_b = \theta_b-\theta_{b-1}$, normalize $s_b = s_b/\|s_b\|_2$  \Comment{Update direction}
                  \State Compute $y_b = H_m^{(b+1)}(\theta_b)s_b$  \Comment{Hessian-vector product}
              \EndIf
              \State Combine $S_b = (s_b,\dots,s_{b-L+1})^\prime$, $Y_b = (y_b,\dots,y_{b-L+1})^\prime$
              \While{ $\lambda_{\min}(S_b^\prime S_b) < \underline{\lambda}_S$ } 
                  \State Discard $s_{b-L+1}$, re-index $s_j$ to $s_{j-1}$, $j = b,\dots,b-L+2$
                  \State Draw $s_b$, normalize $s_b = s_b/\|s_b\|_2$
                  \State Compute $y_b = H_m^{(b+1)}(\theta_b)s_b$
                  \State Combine $S_b = (s_b,\dots,s_{b-L+1})^\prime$, $Y_b = (y_b,\dots,y_{b-L+1})^\prime$
            \EndWhile
            \State Compute $\hat H_b = Y_b^\prime S_b (S_b^\prime S_b)^{-1}$ \Comment{Least-Squares Approximation}
            \State 3) \textbf{Conditioning Matrix} 
            \State Compute $P_b = ( \hat H_b^\prime \hat H_b + \tau_b I_d )^{-1/2}$, where $\tau_b = \underline{\lambda}^2$ if $\lambda_{\min}(\hat H_b^\prime \hat H_b) \leq \underline{\lambda}^2$, $\tau_b = 0$ otherwise
            \end{algorithmic}
\end{algorithm}

\subsection{Comparison with other methods} 
To generate $B$ draws, the standard bootstrap with resampling or reweighting requires solving the minimization problem $\hat\theta_m^{(b)} = \text{argmin}_{\theta \in \Theta} Q_m^{(b)}(\theta)$ $B$ times which can be computationally intensive. Faster alternatives have been proposed in the literature. \citet[\dmk]{Davidson1999} consider a $k$-step \nr\, update: $\hat\theta_{n,k+1}^{(b)} = \hat\theta_{n,k}^{(b)} - [H_n^{(b)}(\hat\theta_{n,k}^{(b)})]^{-1}G_n^{(b)}(\hat\theta_{n,k}^{(b)})$ starting at the same full sample estimate $\hat\theta_{n,0}^{(b)} = \hat\theta_n$ for all $b$. \citet{Andrews2002} shows that $k=1$ is enough for asymptotically valid inference, and $k>1$ provides asymptotic refinements. \dmk\, also propose a $k$-step quasi-Newton update which requires $k>1$ steps.\footnote{Following earlier work by \citet{robinson1988}, the results are based on quadratic and super-linear convergence properties of \nr\, and \textsc{bfgs}, respectively, when $\gamma=1$ which holds only if $\|\hat\theta_{n}^{(b)}-\hat\theta_n\|$ is sufficiently small and if (and only if) the quasi-Newton matrix $P_k$ satisfies certain properties \citep[Th3.1]{dennis1977}.} In comparison, using the same $m=n$, a $k=1$-step implementation of \dmk\, has the same cost as \rnr, but is higher than \rnr\, for $k>1$ steps. The $k$-step quasi-Newton update requires $k>1$ making it $k$ times more costly than \rqn\, for $m=n$. Using $m<n$ further reduces the cost of both estimation and inference for \rnr , \rqn. The original idea of \dmk\, has been extended to fast subsampling \citep{hong2006}, two-step estimation \citep{Armstrong2014}, and $\ell_1$-penalized estimation \citep{li2021}.

Building on \citet{hu1995} and \citet{hu2000}, \citet[\ks]{Kline2012} propose a wild score bootstrap which can be written as a one-step \nr\, update: $\hat\theta_{n}^{(b)} = \hat\theta_{n} - [H_n(\hat\theta_{n})]^{-1}G_n^{(b)}(\hat\theta_{n})$ where $G_n^{(b)}$ is reweighted using de-meaned weights. If both the estimation and using $m=n$ are not too demanding, this approach is computationally attractive. \citet[Sec6.2]{fastwild:19} point out some caveats with this approach. By only evaluating $G_m^{(b)}$ at $\hat\theta_n$, it does not capture sampling variation in $H_n(\hat\theta_{n})$ or non-linearities in the objective. The latter also applies to \dmk\, with $k=1$. By design, \rnr\, and \rqn\, reflect variation in $H_m^{(b)},G_m^{(b)}$ around $\hat\theta_n$. Notice that if the optimizer returns a suboptimal solution $\tilde\theta_n$, then \dmk\, with $k>1$ may not be centered around $\tilde\theta_n$ which is a red flag. By de-meaning the weights, \ks\, produces draws centered around $\tilde\theta_n$ so that the issue goes unnoticed.\footnote{Appendix \ref{apx:add_ex}, Table \ref{tab:DDC_het_50}, illustrates failed optimizations for a dynamic discrete choice model, in comparison \rqn\, produces accurate estimates. It would be problematic to apply \dmk\, or \ks\, to the inconsistent estimates.}  \citet{Honore2017} propose to approximate the sandwich $V_n$ using $ d_\theta(d_\theta+1)/2 \times B$ scalar minimizations of $Q_n^{(b)}$. When $d_\theta$ increases, this quickly becomes taxing.\footnote{In Section \ref{sec:examples}, Example 2, $55945 \times B$ scalar optimizations would be needed which is quite large.} Using a preliminary estimate $\hat\theta_n$ to resample is the defining characteristic of these methods compared to \rnr\, and \rqn\, which never compute $\hat\theta_n$ explicitly.


For likelihood problems, it is common to use Bayesian inference and deploy the MCMC toolkit. Unlike frequentist estimation, Bayesian analyses rely on sampling from a posterior distribution rather than optimization. Well known samplers include Gibbs and Metropolis-Hastings. Because they can be slow to converge, gradient-based samplers are increasingly popular with Metropolis-adjusted Langevin dynamics (MALA), Hamiltonian Monte-Carlo (HMC), or stochastic gradient Langevin dynamics (SGLD) algorithms. A comparison between \rnr, \rqn, and MALA is given in Section \ref{sec:examples}. Unlike Bayesian inference, this paper does not require the information matrix equality for valid inference. Convergence diagnostics used for MCMC algorithms such as trace plots or formal tests \citep[][Ch11.4]{gelman1992,Gelman1998,gelman2013} can also be used to monitor convergence of \rnr\, and \rqn. Lastly, while the detailed balance condition guarantees the posterior is a stationary solution of a Metropolis-Hastings algorithm, convergence rate results typically assume a (nearly) concave log-posterior -- similar to the convexity assumption in this paper.\footnote{Recall that the detailed balance condition ensures that the random walk Metropolis-Hastings algorithm is ergodic but deriving the rate of convergence is more difficult. See e.g. \citet{mengersen1996}, \citet{brooks1998}, \citet{belloni2009}.} \citet{jjng-16} show how optimization can be used to sample from a Bayesian posterior with an intractable likelihood. Here, optimization is used for resampling.

\section{Statistical Properties of \rnr\, and \rqn} \label{sec:rnr1}
This section has three parts. The first subsection presents the main assumptions on $Q_n$, $Q_m^{(b)}$, $P_b$, and derives convergence results for classical optimizers and Algorithm \ref{algo:REI} as well as the pivotal coupling result with the AR(1) process $\theta_b^\star$. These results are non-asymptotic. The second subsection focuses on large sample estimation and inference. The theorems can be applied to any implementation of Algorithm \ref{algo:REI} for which $P_b$ has the required properties. The third subsection specializes them to \rnr\, and \rqn\, described in Algorithms \ref{algo:rnr} and \ref{algo:rqn}. 


\subsection{Convergence of Classical and Resampled Optimizers}
The econometric theory for extremum estimation in e.g. \citet{newey-mcfadden-handbook} takes as a given that an optimizer finds the unique solution.  But the econometric assumptions are not enough to ensure a numerical optimizer will provide a globally convergent solution to the sample problem  $Q_n(\theta)$. The following is standard in gradient-based optimization theory: \vspace{-0.5cm}
\begin{assumption} \label{ass:A1} There exist constants $\underline{\lambda}_H,\overline{\lambda}_H,C_1$ such that for all $\theta \in \Theta$: \vspace{-\topsep}
\begin{itemize} \setlength\itemsep{0em}
      \item[i.] $0 < \underline{\lambda}_H \leq \lambda_{\min}(H_n(\theta)) \leq \lambda_{\max}(H_n(\theta)) \leq \overline{\lambda}_H < \infty$ with probability $1$,
      \item[ii.]  $\|H_n(\theta)-H_n(\hat\theta_n)\|_2 \leq C_1\|\theta-\hat\theta_n\|_2$ with probability $1$.
\end{itemize}  
The matrix $P_k$ is symmetric, such that for some $\underline{\lambda}_P,\overline{\lambda}_P$ and all $k \geq 1$: \vspace{-\topsep}
      \begin{itemize} \setlength\itemsep{0em}
            \item[iii.] $0 < \underline{\lambda}_P \leq \lambda_{\min}(P_k) \leq \lambda_{\max}(P_k) \leq \overline{\lambda}_P < \infty$.
      \end{itemize}
\end{assumption}
\noindent  Condition i. implies that $Q_n$ is strongly convex with probability $1$ while condition ii. requires Lipschitz continuity of the Hessian.\footnote{A function $Q_n$ is strongly convex on $\Theta$ if for all $\theta\in\Theta$,  there exists some $\underline{\lambda}>0$ such that $\nabla^2 Q_n(\theta) \ge \underline{\lambda} I_d$. For bounded $\Theta$, there also exists $\overline{\lambda}$ such that  $\nabla^2 Q_n(\theta) \le \overline{\lambda} I_d$. Then $\overline{\lambda}/\underline{\lambda}$ is an upper bound on the condition number of $\nabla^2 Q_n(\theta)$.} Strong convexity of the sample objective function can be less restrictive than it appears as some non-convex problems can be convexified by penalization.\footnote{In some cases, the strict convexity assumption restricts the parameter space to be bounded in the absence of penalization. See e.g. the discussion of Assumption (H3) in \citet{Bach2011}.}  Assumption \ref{ass:A1} implies the so-called  Polyak-\L{}ojasiewicz inequalities:
\begin{align} 
\label{eq:cvx1} \langle \theta - \hat \theta_n, G_n(\theta) \rangle &= (\theta - \hat \theta_n)^\prime H_n(\tilde \theta_n) (\theta - \hat \theta_n) \geq \underline{\lambda}_{H}  \|\theta-\hat\theta_n\|_2^2, \\
\label{eq:cvx2}
      \| G_n(\theta) \|_2^2 &=  (\theta - \hat \theta_n)^\prime H_n(\tilde \theta_n)^2 (\theta - \hat \theta_n) \leq \overline{\lambda}_H^2  \|\theta-\hat\theta_n\|^2_2,
\end{align}
where $\tilde{\theta}_n$ is an intermediate value, between $\theta$ and $\hat{\theta}_n$. 
The first inequality follows from the positive definiteness of $H_n(\hat\theta_n)$.  To see the usefulness of Assumption \ref{ass:A1}, note that for any symmetric positive definite conditioning matrix $P_k$:
\begin{align*} 
      \|\theta_{k+1}-\hat\theta_n\|_2^2 &= \|\theta_{k}-\hat\theta_n - \gamma \pk G_n(\theta_k)\|_2^2 \nonumber\\
      &= \|\theta_{k}-\hat\theta_n\|_2^2 - 2 \gamma  \langle \theta_{k}-\hat\theta_n, \pk G_n(\theta_k) \rangle + \gamma^2 \| \pk G_n(\theta_k) \|_2^2 \nonumber\\
      &\le \underbrace{\left( 1 - 2\gamma \underline{\lambda}_P \underline{\lambda}_H  +\gamma^2 [\overline{\lambda}_P  \overline{\lambda}_H]^2 \right)}_{=A(\gamma)}
      \|\theta_{k}-\hat\theta_n\|_2^2,
\end{align*}
where the last inequality follows from Assumption \ref{ass:A1} and $0<\underline{\lambda}_P \leq \overline{\lambda}_P < \infty$ bound the eigenvalues of $P_k$ (see Assumption \ref{ass:A1} iii). 
Now a contraction occurs if the choice of $\gamma$ is such that $A(\gamma)\in [0,1)$. To see that this is feasible, note that $A(0)=1$ and $\partial _\gamma A(0)<0$. Thus by continuity and local monotonicity of $A(\cdot)$, there exists a $\gamma^\star \in(0,1]$  such that $A(\gamma)\in[0,1)$ for all $\gamma\in(0,\gamma^\star]$ as desired. For such choice of $\gamma$, define $\overline{\gamma}(\underline \lambda_P,\underline \lambda_H,\bar\lambda_P,\bar\lambda_H)  \in (0,1]$  independent of $k$ be such that:
$A(\gamma) = (1-\overline{\gamma})^2  \in [0,1)$. It follows that:
\begin{align*}
      \|\theta_{k+1}-\hat\theta_n\|_2 &\leq \sqrt{A(\gamma)}\|\theta_{k}-\hat\theta_n\|_2\\& \leq (1-\overline{\gamma})   \|\theta_{k}-\hat\theta_n\|_2 \nonumber \\
      &\leq (1-\overline{\gamma})^{k}   \|\theta_{0} - \hat\theta_n\|_2\rightarrow 0,\quad \text{as } k\rightarrow \infty.\qed
\end{align*}
Global convergence follows from iterating on the contraction at each $k$. The derivations above rely on the positive definiteness of both $H_n$ and $P_k$. A non-definite $P_k$ can be problematic for convergence. While Assumption \ref{ass:A1} i-ii require a smooth and strictly convex objective function, Assumption \ref{ass:A1} iii restricts the choice of conditioning matrix.

The finite upper bound in Assumption \ref{ass:A1} iii is particularly important to ensure the optimizer is well behaved. It is automatically satisfied for \gd\, since $P_b = I_d$ but requires the lower bound in Assumption \ref{ass:A1} i for \nr. For quasi-Newton, it needs to be enforced as in Algorithm \ref{algo:rqn}.  Note that for \nr, with an adjustement similar to the $\tau_b$ in Algorithm \ref{algo:rqn}, it possible to get convergence even when $H_n$ is singular away from $\hat\theta_n$.\footnote{This is illustrated in Appendix \ref{sec:bad}. Note that modifications to the resampled hessian is also present in \citet{gonccalveswhite2005} when deriving consistency of bootstrap standard errors in OLS regressions.} 

\begin{lemma}
\label{lem:cv_non_stochastic}
Under Assumption \ref{ass:A1} i-iii there exists $\gamma \in (0,1]$ such that $A(\gamma)\in [0,1)$. Let $\overline{\gamma}$ be such that $A(\gamma) = (1-\overline{\gamma})^2$. Then
$ \|\theta_k - \hat\theta_n\|_2 \leq (1-\overline{\gamma})^k\|\theta_0-\hat\theta_n\|_2 \to 0, \text{ as } k\to\infty. $
\end{lemma} 
The first Lemma derives the deterministic global convergence of a classical gradient-based optimizer for an appropriate choice of $\gamma$ at a dimension-free linear rate of $(1-\overline{\gamma})$. Convergence of Newton-Raphson with $P_k=H_k^{-1}$ is implied by the Lemma under Assumption \ref{ass:A1} only. This first result is useful for deriving the stochastic convergence of the resampled Algorithm \ref{algo:REI}. In particular, notice that the resampled update is:
\[ \theta_{b+1}  =
 \underbrace{ \vphantom{\theta_b - \gamma P_b G_n(\theta_b)} \vphantom{\gamma P_b [ G_n(\theta_b) - G_m^{(b+1)}(\theta_b)  ]}\theta_b - \gamma P_b G_m^{(b+1)}(\theta_b)}_{\text{Resampled update}} =
 \underbrace{ \vphantom{\gamma P_b [ G_n(\theta_b) - G_m^{(b+1)}(\theta_b)  ]}\vphantom{\theta_b - \gamma P_b G_m^{(b+1)}(\theta_b)} \theta_b - \gamma P_b G_n(\theta_b)}_{\text{Deterministic update}} +
 \underbrace{ \vphantom{\theta_b - \gamma P_b G_n(\theta_b)} \vphantom{\theta_b - \gamma P_b G_m^{(b+1)}(\theta_b)} \gamma P_b [ G_n(\theta_b) - G_m^{(b+1)}(\theta_b)  ]}_{\text{Resampling noise}}, \]
which is the sum of the deterministic update studied in Lemma \ref{lem:cv_non_stochastic} with resampling noise. The convergence has a deterministic and a stochastic component. The deterministic part is covered by Assumption \ref{ass:A1} and the condition on $\gamma$ in Lemma \ref{lem:cv_non_stochastic}. The following assumption will be used to study the stochastic part of the update.  

\begin{assumption}
\label{ass:A2}
Let $\mathbb{E}^\star$ denote the expectation under resampling or reweighting, i.e. conditional on $(z_1,\dots,z_n)$. Suppose there exists $\underline{\lambda}_H,\overline{\lambda}_H,C_2,C_3$ such that for $p \geq 2$: \vspace{-\topsep}
\begin{itemize} \setlength\itemsep{0em}
      \item[i.] $[\mathbb{E}^\star(\sup_{\theta \in \Theta}\|G_m^{(b)}(\theta)-G_n(\theta)\|^p)]^{1/p} \leq C_2 m^{-1/2}$ and $\mathbb{E}^\star[G_m^{(b)}(\theta)]=G_n(\theta)$,
      \item[ii.] $[\mathbb{E}^\star(\sup_{\theta \in \Theta}\|H_m^{(b)}(\theta)-H_n(\theta)\|^p)]^{1/p} \leq C_3 m^{-1/2}$,
      \item[iii.] $0 < \underline{\lambda}_H \leq \lambda_{\min}(H_m^{(b)}(\theta)) \leq \lambda_{\max}(H_m^{(b)}(\theta)) \leq \overline{\lambda}_H < \infty$, with probability $1$.
\end{itemize}
The conditioning matrix $P_b$ is such that for all $b \geq 1$: \vspace{-\topsep}
      \begin{itemize} \setlength\itemsep{0em}
            \item[iv.] $0 < \underline{\lambda}_P \leq \lambda_{\min}(P_b) \leq \lambda_{\max}(P_b) \leq \overline{\lambda}_P < \infty$, with probability $1$.
      \end{itemize}
\end{assumption}
These high-level conditions will be used to prove convergence of $\thetab$ in $L_p$-norm, $p \geq 2$. Appendix \ref{apx:primA2} provides sufficient conditions for i-ii. using either $m$ out of $n$ resampling with replacement or reweighting with Gaussian multiplier weights.\footnote{Dependence of $C_2,C_3$ on $n$ is omitted to simplify notation. The conditions under Gaussian reweighting are the same for $p=2$ and $p>2$ but not with resampling. The results can be extended to other weighting distributions with the same conditions as resampling by adapting the proof of Theorem 3.6.13 in \citet{VanderVaart1996} similar to \citet[pp18-20]{cheng2015}.} Consistent estimation and valid inference using quantiles in Algorithm \ref{algo:REI} (steps 4 and 5), will require $p=2$ in Assumption \ref{ass:A2}. Consistency of standard errors further requires $p=4$. Consistency of moments for the standard bootstrap is also associated with particular moment conditions, see \citet{gonccalveswhite2005}, \citet{Kato:11}, \citet{cheng2015}. Since the estimator in Algorithm \ref{algo:REI} is computed by averaging draws, the derivations are fundamentally moment-based.

Assumption \ref{ass:A2} iv reiterates the importance of the bounds for $P_b$ for convergence but in the stochastic case. Under Assumptions \ref{ass:A1}, \ref{ass:A2}, and the same choice of $\gamma$ used in Lemma \ref{lem:cv_non_stochastic}:
\[ \|\theta_{b+1}-\hat\theta_n\|_2 \leq (1-\overline{\gamma})\|\theta_{b}-\hat\theta_n\|_2 + \gamma \overline{\lambda}_P \sup_{\theta \in \Theta} \| G_m^{(b+1)}(\theta)-G_n(\theta) \|_2,  \]
with probability 1. Iterating the recursion reveals the general behaviour of $\theta_b$:
\[ \|\theta_{b+1}-\hat\theta_n\|_2 \leq (1-\overline{\gamma})^b\|\theta_{0}-\hat\theta_n\|_2 + \gamma \overline{\lambda}_P \sum_{j=0}^{b-1} (1-\overline{\gamma})^j\sup_{\theta \in \Theta} \| G_m^{(b+1-j)}(\theta)-G_n(\theta) \|_2,  \]
with probability 1. The leading term converges deterministically and exponentially fast to zero. It corresponds to a starting values bias. When Assumption \ref{ass:A2} i holds for some $p \geq 2$ and $(1-\overline{\gamma})\in [0,1)$, the second term has bounded $L_p$-norm, of order $1/\sqrt{m}$. This characterizes the stochastic convergence of $\theta_b$ as show in the following Lemma.

\begin{lemma} \label{lem:cv_stochastic} 
 Suppose Assumptions \ref{ass:A1}-\ref{ass:A2} hold
and $\gamma\in(0,1]$ is such that $(1-\overline{\gamma})^2=A(\gamma)\in[0,1)$, as defined in Lemma \ref{lem:cv_non_stochastic}. Then there exists a constant $C_4$ which only depends on $\gamma,\overline{\lambda}_P,C_3$ such that:
      \begin{eqnarray*} 
      \left[ \mathbb{E}^\star \left( \|\theta_{b+1}-\hat\theta_n\|_2^p \right) \right]^{1/p}
      &\le &  (1-\overline{\gamma})^{b+1} \bigg[\mathbb{E}^\star(\|\theta_0-\hat\theta_n\|_2^p)\bigg]^{1/p}+\frac{C_4}{\overline{\gamma} \sqrt{m}}  .
      \end{eqnarray*}
\end{lemma} 
Lemma \ref{lem:cv_stochastic} shows that $\theta_b$ converges in $L_p$-norm to a $1/\sqrt{m}$-neighborhood of $\hat\theta_n$. Unlike classical optimizers or \sgd\, with $\gamma_k \to 0$ where $\theta_k \to \hat\theta_n$ either deterministically or in probability, here $\thetab \not\to \hat\theta_n$ as $b$ increases -- this is similar to MCMC which produces draws around $\hat\theta_n$. 
Except for quadratic objectives, the Markov-process followed by $\theta_b$ around $\hat\theta_n$ is intractable. Nevertheless, the following will show that
\begin{equation}
\thetabf^\star - \hat \theta_n = \Psi_n(\thetab^\star - \hat\theta_n) - \gamma \overline{P}_m G_m^{(b+1)}(\hat\theta_n), \quad \theta^\star_0 = \theta_0,
\label{eq:lin}
\end{equation}
gives a good approximation of $\theta_b$ and its distribution. The autoregressive coefficient $\Psi_n$ and $\overline{P}_m$ depend on the choice of update $P_b$. Table \ref{tab:thetastar} makes them explicit for \rgd, \rnr, and \rqn. 
\begin{table}[H] \caption{Coefficients in the linear coupling process $\theta_b^\star$ } \label{tab:thetastar} \centering
      \begin{tabular}{l|ccc} \hline \hline
            Algorithm & $P_b$ & $\overline{P}_m$ & $\Psi_n$ \\ \hline
            Resampled Gradient Descent (\rgd) & $I_d$ & $I_d$ & $I_d - \gamma H_n(\hat\theta_n)$ \\
            Resampled Newton-Raphson (\rnr) & $[H_m^{(b+1)}(\theta_b)]^{-1}$ & $[H_n(\hat\theta_n)]^{-1}$ & $(1-\gamma)I_d$ \\
            Resampled quasi-Newton (\rqn) & Algorithm (\ref{algo:rqn}) & $[H_n(\hat\theta_n)]^{-1}$ & $(1-\gamma)I_d$ \\
            \hline \hline
      \end{tabular}
\end{table}





 \begin{assumption} \label{ass:A3}
      Let $d_{0,n} = [ \mathbb{E}^\star(\|\theta_0-\hat\theta_n\|^p_2) ]^{1/p}$, suppose that there exists a matrix $\overline{P}_m$ symmetric positive definite and constants $C_5$, $\rho \in [0,1)$ such that for $\Psi_n = I -\gamma \overline{P}_m H_n(\hat\theta_n)$:\vspace{-\topsep}
      \begin{itemize} \setlength\itemsep{0em}
            \item[i.] $0 \leq \lambda_{\max}(\Psi_n) < 1$,
            \item[ii.] $\left[ \mathbb{E}^\star(\|I-P_b[\overline{P}_m]^{-1}\|^p_2) \right]^{1/p} \leq C_5 \left( \rho^b d_{0,n} + 1/\sqrt{m} \right)$.
      \end{itemize}
\end{assumption}

In the literature, $\thetab^\star$ is known as a coupling of $\thetab$. The two have different marginal distributions because one follows a linear and the other a  non-linear process. Nonetheless, they live on the same probability space because they are generated from the same resampled objective $Q_m^{(b)}$. 
Hence if we can show that $\|\overline{\theta}_{\rnr}-\overline{\theta}_{\rnr}^\star\|$ is negligible, we can work with the distribution of $\thetab^\star$ which is more tractable.  

\begin{proposition} \label{prop:coupling}
      Suppose that Assumptions \ref{ass:A1}, \ref{ass:A2}, and \ref{ass:A3} hold for $p\geq 2$ and $\gamma \in (0,1]$ is such that Lemma \ref{lem:cv_stochastic} is satisfied. Let $\overline{\rho} = \max[1-\overline{\gamma},\rho,\lambda_{\max}(\Psi_n)]$, and $d_{0,n}=[\mathbb{E}^\star(\|\theta_0-\hat\theta_n\|_2^p)]^{1/p}$. There exists a $C_6$, which depends on $C_1,\dots,C_5,\overline{\rho},\gamma,\overline{\gamma},\overline{\lambda}_H,\overline{\lambda}_P$ such that:
      \begin{align} \left[\mathbb{E}^\star \bigg(\| \theta_b-\theta^\star_{b}\|_2^{p/2}\bigg)\right]^{2/p}\le C_6 \bigg(\frac{1}{m}+\overline{\rho}^b[d_{0,n}+d_{0,n}^2] \bigg). \label{eq:avg_coupling}
      \end{align}
\end{proposition}
The constant $\overline{\gamma} \in (0,1]$ is  the global rate of convergence in Lemma \ref{lem:cv_stochastic}. Putting $p=2$, the lemma states that that the expected deviation $\mathbb{E}^\star(\|\thetab-\thetab^\star\|)$ is of order $1/m$ plus a term that vanishes exponentially fast in $b$.  This in turn implies that for $p=2$, upon  averaging,  $\mathbb{E}^\star(\|\overline{\theta}_{\textsc{re}}-\overline{\theta}_{\textsc{re}}^\star\|_2)$  is of order $1/m$ plus a $1/B$ term due to $d_{0,n}$. Using these non-asymptotic bounds, the large sample estimation and inference properties of Algorithm \ref{algo:REI} can be derived.

\subsection{Large Sample Estimation and Inference}

Proposition \ref{prop:coupling} is pivotal for what follows. Suppose for simplicity that there is no approximation error, i.e. $\theta_b=\theta_b^\star$, then:
\[ \mathbb{E}^\star(\theta_b) = \hat\theta_n + [\Psi_n]^b (\theta_0-\hat\theta_n), \]
converges to $\hat\theta_n$ exponentially fast under Assumption \ref{ass:A3} i. This implies that \rgd, \rnr, and \rqn\, deliver consistent estimates. This should be of interest for practitionners using \sgd\, with $\gamma_k = \gamma$ fixed and Polyak-Ruppert averaging, which is common in practice, since the theory gives restrictions on $m$ for the bias to be negligible. Furthermore, $\theta_b^\star$ being an ergodic first-order vector autoregressive process:
\[ m \times \text{var}^\star(\theta_b) \to \gamma^2 \Psi_n \text{var}^\star[\sqrt{m}G_m^{b}(\hat\theta_n)] \Psi_n^\prime , \]
converging exponentially fast. It consistently estimates the asymptotic variance, up to scale, when $\text{var}^\star[\sqrt{m}G_m^{b}(\hat\theta_n)]$ consistently estimates the meat and $\Psi_n$ approximates the bread $[H_n(\hat\theta_n)]^{-1}$ (up to scale). This is the case for \rnr\, and \rqn, but not \rgd. Suppose $\sqrt{m}G_m^{b}(\hat\theta_n) \sim \mathcal{N}(0,\Sigma_n)$ then, for large $b$ we approximately have $\theta_b \sim \mathcal{N}(\hat\theta_n,\gamma^2 \Psi_n \Sigma_n \Psi_n^\prime/n)$ which matches the asymptotic distribution, after re-scaling the variance, if $\Psi_n$ approximates the bread (up to scale). The following Theorem gives a non-asympotic bound for the approximation error $\|\overline{\theta}_{\textsc{re}}-\hat\theta_n\|_2$ and the large-sample consistency result.

\begin{theorem}[Large Sample Estimation]
      \label{th:average_estimation} Suppose that $\Sigma_n = \text{var}^\star [\sqrt{m}G_m^{(b)}(\hat\theta_n)]$ is finite and bounded. If  the conditions of Proposition \ref{prop:coupling} hold for $p=2$, then $\overline{\theta}_{\textsc{re}}$ satisfies:
      \begin{itemize} \setlength\itemsep{0em}
\item[i.] $ \mathbb{E}^\star\left( \| \overline{\theta}_{\textsc{re}} - \hat\theta_n \|_2 \right) \leq C_7 \left( \frac{1}{m} + \frac{d_{0,n} + d_{0,n}^2}{B} + \frac{1}{\sqrt{mB}} \right)$ where $C_7$ depends on $C_{6}$ and $\text{trace}(\Sigma_n)$. 
\item[ii.] If $\sqrt{n}/\min(B,m) \to 0$:
      $\sqrt{n} \left( \overline{\theta}_{\textsc{re}} - \theta^\dagger \right) =  \sqrt{n} \left( \hat\theta_n - \theta^\dagger \right) + o_{p^\star}(1).$
 \end{itemize}
\end{theorem}
First-order asymptotic equivalence of $\overline{\theta}_{\textsc{re}}$ and $\hat\theta_n$ requires both $B$ and $m$ to be large relative to $\sqrt{n}$. Using fixed $m$ and $\gamma$ does not produce consistent estimates. This is relevant for \sgd\, with $m\ll n$ and $\gamma_k=\gamma$. For inference, consider the standard bootstrap as a benchmark. Let $\overset{d^\star}{\to}$ denote convergence in distribution conditional on $(z_1,\dots,z_n)$. Under regularity conditions, an $n$ out of $n$ bootstrap estimate $\hat \theta_n^{(b)}  =\text{argmin}_{\theta \in \Theta} Q_n^{(b)}(\theta)$ is such that
\begin{align*}
\label{eq:stdboot}
\sqrt{n}(\hat \theta_n^{(b)} - \hat\theta_n) =- [H_n(\hat\theta_n)]^{-1} \sqrt{n}G_n^{(b)}(\hat\theta_n)+o^\star_p(1).
\end{align*}
If $\sqrt{n}G_n^{(b)}(\hat\theta_n)$ satisfies a conditional CLT then the bootstrap distribution targets the correct asymptotic distribution. The bootstrap is often used to compute standard errors. Consistency of moments further requires uniform integrability.\footnote{See \citet[Th16.14]{billingsley2013}. Specifically, to consistently estimate the variance $V_n$, the condition requires $\sup_{n} \mathbb{E}^\star(\|\sqrt{n}(\hat \theta_n^{(b)} - \hat\theta_n)\|^{2+\varepsilon})$ finite for some $\varepsilon >0$. See \citet{gonccalveswhite2005}, \citet{Kato:11}, \citet{cheng2015}, for applications to moment convergence of the bootstrap in OLS regressions and M-estimation.} Here, convergence in distribution of $\theta_b$ comes from Proposition \ref{prop:coupling} and a conditional CLT for $\theta_b^\star$, i.e. the growing sum $\sum_{j=0}^b (1-\gamma)^j [H_n(\hat\theta_n)]^{-1}\sqrt{m}G_m^{(b)}(\hat\theta_n)$. Assumption \ref{ass:A4} is used to ensure that distributional convergence for $\sqrt{m}G_m^{(b)}(\hat\theta_n)$ extends to $\theta_b^\star$. Primitive conditions are given in Appendix \ref{apx:primA4} for $m$ out of $n$ resampling, and reweighting with Gaussian weights.

\begin{assumption} \label{ass:A4} Let $V_n = [H_n(\hat\theta_n)]^{-1} \Sigma_n [H_n(\hat\theta_n)]^{-1}$, where $\Sigma_n \ m \times \text{var}^\star [G_m^{(b)}(\hat\theta_n)]$. There exists constants $\beta \in (0,1/2]$, $\kappa > 0$, and a remainder $\|r_m(\tau)\| \leq C_\psi \|\tau\|^\kappa$  such that for $\mathbf{i}^2=-1$,
      $\mathbb{E}^\star \left(  \exp \left[\sqrt{m} \mathbf{i}\tau^\prime (V_n^{-1/2} [H_n(\hat\theta_n)]^{-1}G_m^{(b)}(\hat\theta_n) \right] \right) = \exp \left( - \frac{\|\tau\|_2^2}{2}\right) \times \left( 1+ \frac{r_m(\tau)}{m^{\beta}} \right).$
\end{assumption}

\begin{assumption} \label{ass:A5} Suppose the following holds:
      \vspace{-\topsep}
\begin{itemize} \setlength\itemsep{0em}
      \item[i.] $V_n \overset{p}{\to} V = [H(\theta^\dagger)]^{-1} \Sigma [H(\theta^\dagger)]^{-1}$, $\sqrt{n}(\hat\theta_n-\theta^\dagger) \overset{d}{\to} \mathcal{N}(0,V)$,
      \item[ii.] $h: \Theta \to \mathbb{R}$ is continuously differentiable, $\nabla h$ is Lipschitz-continous, $\nabla h(\theta^\dagger) \neq 0$.
\end{itemize}
\end{assumption}

Assumption \ref{ass:A5} is standard. It requires asymptotic normality of $\hat\theta_n$, consistency of the sandwich estimator, and validity of the delta method for inference on a function $h$ of $\theta$. Theorem \ref{th:asym_norm} presents the large sample inference. First, the asympotic validity quantile-based confidence intervals. Second, consistency of standard errors computed in Algorithm \ref{algo:REI}. Both results require a choice of $P_b$ satisfying $\overline{P}_m = [H_n(\hat\theta_n)]^{-1}$.

\begin{theorem}[Large Sample Inference]
      \label{th:asym_norm} 
Suppose that the conditions of Proposition \ref{prop:coupling} hold with $\overline{P}_m = [H_n(\hat\theta_n)]^{-1}$ and $p=2$, that $\Sigma_n$ is bounded and non-singular, and Assumptions \ref{ass:A4}-\ref{ass:A5} hold, then as $m,n,$ and $b\rightarrow\infty$, with $\log(m)/b \to 0$, for any $\alpha \in (0,1)$:
\begin{itemize} \setlength\itemsep{0em}
      \item[i.] $\lim_{n\to\infty} \mathbb{P}_n[ c_{h,b}(\alpha/2) \leq h(\theta^\dagger) - h(\hat\theta_n) \leq c_{h,b}(1-\alpha/2) ] = 1-\alpha$,
\end{itemize}
where $c_{h,b}(\alpha)$ is the $\alpha$-th quantile of the adjusted draws $h(\tilde \theta_b)$, $\tilde \theta_b = \hat\theta_n + \sqrt{\frac{m}{n\phi(\gamma)}}(\thetab - \hat\theta_n)$. If, in addition, the assumptions hold with $p=4$:
\begin{itemize} \setlength\itemsep{0em}
\item[ii.] $\frac{m}{\phi(\gamma)} \text{var}^\star [h(\theta_b)] = \nabla h(\hat\theta_n) V_n \nabla h(\hat\theta_n)^\prime + o(1).$
\end{itemize}
\end{theorem}

\subsection{Implications for \rnr\, and \rqn}  \label{sec:corro}
Having the general estimation and inference results for Algorithm \ref{algo:REI}, the following verifies the implementation-specific Assumptions \ref{ass:A2} iv, \ref{ass:A3} to show that Theorems \ref{th:average_estimation}-\ref{th:asym_norm} apply to \rnr\, and \rqn. For \rqn, it is assumed, without loss of generality, that $\underline{\lambda} = \underline{\lambda}_H/2$.\footnote{The proof only relies on $\underline{\lambda} < \lambda{\min}(H_n(\theta))$ around $\hat\theta_n$.}
\begin{corollary}[Estimation and Inference with \rnr\, and \rqn] \label{coro:rnrqn} Consider an implementation of Algorithm \ref{algo:REI} using either \rnr\, or \rqn, i.e. the choice of $P_b$ described in Algorithms \ref{algo:rnr} and \ref{algo:rqn}. 
      Suppose Assumptions \ref{ass:A1} i-ii and \ref{ass:A2} i-iii hold with $p \geq 2$, then Assumption \ref{ass:A2} iv holds for the same $p$. If the assumptions for Lemma \ref{lem:cv_stochastic} also hold, then Assumption \ref{ass:A3} holds with the same $p$ and $\overline{P}_m = [H_n(\hat\theta_n)]^{-1}$. Then, under the additional assumptions required for each set of results, Theorems \ref{th:average_estimation} and \ref{th:asym_norm} apply to \rnr\, and \rqn.
\end{corollary}

Consistency of standard errors implies that they can be used to compute t-statistics, confidence intervals of the form $h(\overline{\theta}_{\textsc{re}}) \pm 1.96 \sqrt{m/[n\phi(\gamma)] \widehat{\text{var}}[h(\theta_b)]}$, where $\widehat{\text{var}}[h(\theta_b)]$ is the sample variance of $h(\theta_b)$, after discarging the first $\textsc{burn} \gg \log(m)$ iterations. A simple rule-of-thumb to gauge the order of magnitude for $\textsc{burn}$ is that $(1-\overline{\gamma}) \geq (1-\gamma)$ and the initial value bias is of order $(1-\overline{\gamma})^b d_{0,n}$. Suppose, we target a bias of order $d_{0,n}/100$, then we need $b > \log(1/100)/\log(1-\overline{\gamma}) \geq \log(1/100)/\log(1-\gamma) \simeq 43$ for $\gamma=0.1$. This implies that the burn-in should include at least $40-50$ draws. A burn-in of $225$ draws corresponds to setting $\overline{\gamma}=0.02$ in the same calculations.

The quasi-Newton update proposed in Algorithm \ref{algo:rqn} differs from \textsc{bfgs} and other methods in the literature to ensure the following three properties are satisfied: i) $P_b$ is symmetric, and ii) positive definite with strictly positive eigenvalues, iii) Assumption \ref{ass:A3} holds when $\|\theta_b - \hat\theta_n\|$ is small. Conditions i-ii) are required to ensure optimization is stable and apply Lemma \ref{lem:cv_stochastic}. Condition iii) is needed for valid inference. Other methods do not satisfy i), ii), and iii) simultaneously. Amongst the most widely used methods, \textsc{bfgs} enforces i) and ii) plus the so-called secant equation: $P_b y_{k} = s_k$ so that $P_b H_n(\theta_k)s_{k} = s_k$ always holds but is only guaranteed in this direction so that iii) need not hold for non-quadratic $Q_n$.\footnote{See \citet{ren1983} for counter-examples. The BFGS matrix $P_k$ can have negative eigenvalues when $Q_n$ is non-convex. This is problematic for convergence as illustrated in Section \ref{sec:saddle}.} The \textsc{sr1} update, used with trust-region algorithms, satisfies i) and iii) but not ii) which makes it unstable in estimation using (\ref{eq:gradient-rule}).\footnote{See \citet{fiacco1990}, \citet{Conn1991} for derivations of iii) with SR1.} Both methods update $P_k$ from $P_{k-1}$ using a rank-one update based on the secant equation above. Here $P_b$ is fully updated using OLS, this is related to the multi-secant update in \citet{schnabel:83} which can satisfy iii) but typically not i-ii), and requires $L \leq d_{\theta}$.

In the proof of the Corollary, the OLS estimate $\hat H_b$ is shown to satisfy iii) but not necessarily i)-ii). The transformation  $\hat H_b^\prime \hat H_b$ enforces i), and the regularization $\tau_b$ enforces ii).\footnote{These two transformations can also be useful for \nr\, and \rnr\, as illustrated in Section \ref{sec:discussion}.} Algorithm \ref{algo:rqn} requires three inputs. At least $L \geq d_\theta$ secant updates are needed to compute $\hat H_b$ by OLS; $L \geq \max(25,1.5 \times d_\theta)$ works well in the examples. The cutoff $\underline{\lambda}_S$ ensures the least-squares problem is well-conditionned. With $\underline{\lambda}_S=10^{-6}$, the while loop never runs in the applications, but it should be included as a safeguard. The requirement for $\underline{\lambda}$ is to be small enough to satisfy $0<\underline{\lambda} < \lambda_{\min}(H_n)$ around $\hat\theta_n$.

\section{Three Examples} \label{sec:examples}
The following provides one Monte-Carlo and two empirical examples which illustrate the properties of \rnr\, and \rqn\, for estimation and inference.

\subsection{Example 1: Dynamic Discrete Choice with Unobserved Heterogeneity} \label{sec:DDCh}

The first example considers likelihood estimation and inference in a single-agent dynamic discrete choice model with unobserved heterogeneity.\footnote{See e.g. \citet{aguirregabiria2010} and \citet{arcidiacono-ellickson:11} for reviews on the estimation of dynamic discrete choice models.} It is common in this class of models to use a parametric mixture distribution to model heterogeneity since it is parsimonious relative to a more computationally demanding non-linear fixed effect estimation. There are two challenges in this example. First, the likelihood is non-convex because of the unobserved mixture component. Second, the non-linear transformation used to integrate out the unobserved heterogeneity results in finite sample bias and sizable distortion for inference. Analytical bias correction can be challenging to implement, because the terms are not easily tractable. Standard bootstrap bias correction would be computationally intensive. Using a large-T linearization of the bias, it is shown that the split panel jackknife of \citet{dhaene2015} can reduce this bias, and straightforward to use with \rqn.\footnote{See Appendix \ref{apx:add_ex} for the derivations under simplifying assumptions.} A comparison of \rnr\, and \rqn\, with MLE and the standard bootstrap for a simpler homogeneous agents model are given in Appendix \ref{apx:add_ex1}, Table \ref{tab:homDDC}.

The setting is similar to the seminar model of \citet{rust:87}, each agent $i$ solves a dynamic programming problem at each period $t$ with a flow utility for the set of actions $a_{it}=0$ or $1$:
\[ U_{i0}(x_{it},a_{it-1}) = \varepsilon_{it}(0), \quad U_{i1}(x_{it},a_{it-1}) = \beta_{0i} + \beta_{1i}x_{it} -(1-a_{it-1})\delta_1 + \varepsilon_{it}(1), \]
where $\delta_1$ is the entry cost, $\varepsilon_{it}(0)$ and $\varepsilon_{it}(1)$ follow an extreme-value type I distribution, $x_{it}$ is an exogenous regressor which has finite support and follows a first-order Markov process with transition probability $\Pi$. Both $a_{it}$ and $x_{it}$ are observed for $i=1,\dots,n$ and $t=1,\dots,T$. Here the heterogeneity is modeled on the intercept and slope parameters of the utility function. The intercept $\beta_{0i}$ equals $\mu_{0}^1$ with probability $\omega \in [0,1]$, and $\mu_{0}^2$ with probability $1-\omega$. The slope is continuously distributed $\beta_{1i} \sim \omega \mathcal{N}(\mu_1,\sigma_1^2) + (1-\omega)\mathcal{N}(\mu_2,\sigma_2^2)$. The parameters of interest are $\theta=(\mu_0^1,\mu_1^1,\mu_0^2,\mu_1^2,\sigma_1,\sigma_2,\omega,\delta_1)$, the discount factor $\rho=0.98$ is fixed. The matrix $\Pi$ is estimated by sample average and taken as input in the log-likelihood below:
\[ \ell_{nT}(\theta,\Pi) = \frac{1}{nT}\sum_{i=1}^n \log \left( \int \exp[\sum_{t=1}^T \ell_{it}(\theta,\Pi,\beta_i)]f(\beta_i|\theta)d\beta_i \right), \]
where $\ell_{it}(\theta,\Pi,\beta_i) = \log[\mathbb{P}(a_{it}|x_{it},a_{it-1},\theta,\Pi,\beta_i)]$ is the conditional probability of choice $a_{it} \in \{0,1\}$ computed by solving the dynamic programming problem by fixed point iterations, and $f(\cdot|\theta)$ is the mixture distribution for $\beta_i = (\beta_{0i},\beta_{1i})$.

The split panel jacknife is implemented using a full sample implementation of \rqn\, $\theta_{b,nT}$ and two half panel implementations $\theta_{b,nT/2}^1,\theta_{b,nT/2}^2$ based on the first and last $T/2$ time-observations each. All three use the same $m=n/2$ out of $n$ resampled observations at each $b$. Bias-corrected draws are given by $\tilde{\theta}_{b,nT}=2\theta_{b,nT}-[\theta_{b,nT/2}^1+\theta_{b,nT/2}^2]/2$. Table \ref{tab:DDC_het} summarizes the results from $400$ Monte Carlo replications. Full panel \rqn\, estimates (denoted as \rqn) are close to $\theta^\dagger$. Some coefficients have a bias comparable in magnitude to sampling uncertainty. Inference for these coefficient is far from the nominal $5\%$ size using either quantiles or standard errors. The split panel correction (\rqn-bc) reduces bias by an order of magnitude and bias-corrected inference is much closer to nominal size. Table \ref{tab:DDC_het_50} provides additional results with $T=50$ and a comparison with R's \textsc{bfgs} optimizer. In that configuration, \textsc{bfgs} estimates are very inaccurate, and far less reliable than \rqn.  
\begin{table}[ht] \caption{Dynamic Discrete Choice Model with Heterogeneity} \label{tab:DDC_het}
      \centering
      \begin{tabular}{l|cccccccc}
        \hline \hline
       & $\mu_0^1$ & $\mu_1^1$ & $\mu_0^2$ & $\mu_1^2$ & $100 \sigma_1$ & $100 \sigma_2$ & $\omega$ & $\delta_1$ \\ 
        \hline
        $\theta^\dagger$ & -2.000 & 0.300 & -1.000 & 0.900 & 1.000 & 1.000 & 0.300 & 1.000 \\ \hline
        & \multicolumn{8}{c}{Average Estimates}
        \\\hline
        \rqn & -2.033 & 0.307 & -0.982 & 0.903 & 1.015 & 1.154 & 0.302 & 0.976 \\ 
        \rqn-bc & -2.019 & 0.307 & -0.991 & 0.896 & 1.006 & 1.046 & 0.302 & 1.001 \\ 
        \hline
        & \multicolumn{8}{c}{Standard Deviation}
        \\\hline 
        \rqn & 0.044 & 0.013 & 0.020 & 0.009 & 0.127 & 0.181 & 0.012 & 0.018 \\ 
        \rqn-bc & 0.043 & 0.013 & 0.021 & 0.011 & 0.146 & 0.155 & 0.013 & 0.019 \\ \hline
        & \multicolumn{8}{c}{Rejection Rate}
        \\ \hline
        \rqn & 0.315 & 0.025 & 0.030 & 0.041 & 0.036 & 0.330 & 0.028 & 0.122 \\ 
        \rqn-bc & 0.025 & 0.030 & 0.028 & 0.030 & 0.030 & 0.033 & 0.028 & 0.046 \\ 
        \rqn-bc$_{se}$ & 0.023 & 0.033 & 0.030 & 0.033 & 0.028 & 0.033 & 0.030 & 0.041 \\ 
         \hline \hline
      \end{tabular}\\
      {\footnotesize Legend: \rqn-bc = split panel bias-corrected \rqn; \rqn, \rqn-bc = quantile-based CIs, \rqn-bc$_{se}$  = std-error based CIs. $\gamma=0.1,n=1000,T=100,m=n/2, B=2000, \textsc{burn}=250,$ nominal level $=5\%.$}
\end{table}

\subsection{Example 2: A Probit with Many Regressors}
It is common to use MCMC methods for large-scale estimation problems. The second example compares the properties of \rnr\, and \rqn\, with a Metropolis-Hastings algorithm called MALA (Metropolis-adjusted Langevin Algorithm). The idea is to compare the convergence and mixing properties of the draws when there are hundreds of parameters to estimate.\footnote{A description of and implementation details for MALA are given in the Appendix.}

The model is a simple Probit using the data from \citet{helpman2008}. There are 10 regressors plus 324 exporter and importer fixed effects for a total of 334 parameters to be estimated using 248,060 country/year observations. \rnr\, and \rqn\, are implemented with $\gamma = 0.1$ by re-weighting with Gaussian multiplier weights. Implementation details for the different algorithms are given in the Appendix. While \rqn\, accepts all draws by design, MALA only accepts 57\% of them which makes the Markov-chain more persistent as illustrated in Figure \ref{fig:BigProbit} below. The top panel shows the faster convergence \rnr\, and \rqn. After discarding the initial burn-in, the autocorrelation coefficients are $0.889$ for \rnr\, and \rqn; $0.995$ for MALA, which indicates better mixing.\footnote{Note that $0.995^{25} \simeq 0.889$ this is why $25$ times more draws are generated for MCMC than \rqn.} The bottom panel shows that \rnr\, and \rqn\, draws are more localized around $\hat\theta_{n}$ than MCMC. Indeed, $\gamma=0.1$ implies $\phi(\gamma)=19$: the variance of $\theta^{\rqn}_b$ is 19 times smaller than $\theta^{\text{\textsc{mcmc}}}_b$.
 
Table \ref{tab:BigProbit} in the Appendix further reports estimates and standard errors for the 10 regressors and the intercept. MLE and its standard errors are computed using R's \textit{glm} function and \textit{vcovHC} for robust standard errors. Results for \rqn\, are close to MLE for all 11 coefficients. Bayesian posterior standard deviations tend to underestimate. Clustered re-weighting with exponential weights are also reported for comparison, it only requires modifying the line of code generating the weights. In terms of computation, the likelihood, gradient, and hessian are computed in C++ using Rcpp for performance. R's built-in \textit{glm} takes about 2min for estimation using iteratively reweighted least-squares, \rqn\, 4min, \rnr\, 10min, and MCMC 1h30min. As for most applications of gradient-based optimization on large datasets with many parameters, a good implementation of the code computing the gradient and Hessian is crutial for fast computation: unoptimized code runs is about 2 hours for \rqn, and 1.5 days for MCMC. Results for \sgd\, and an infeasible implementation of \snr, which relies on the full sample $H_n(\hat\theta_n)$, are reported in the Appendix. Even after 5 million iterations ($>2$ hours), \sgd\, is still far from the solution $\hat\theta_n$. 

\begin{figure}[h] \caption{Probit - Comparison of \rnr, \rqn\, and \textsc{mcmc}} \label{fig:BigProbit}
\includegraphics[scale=0.55]{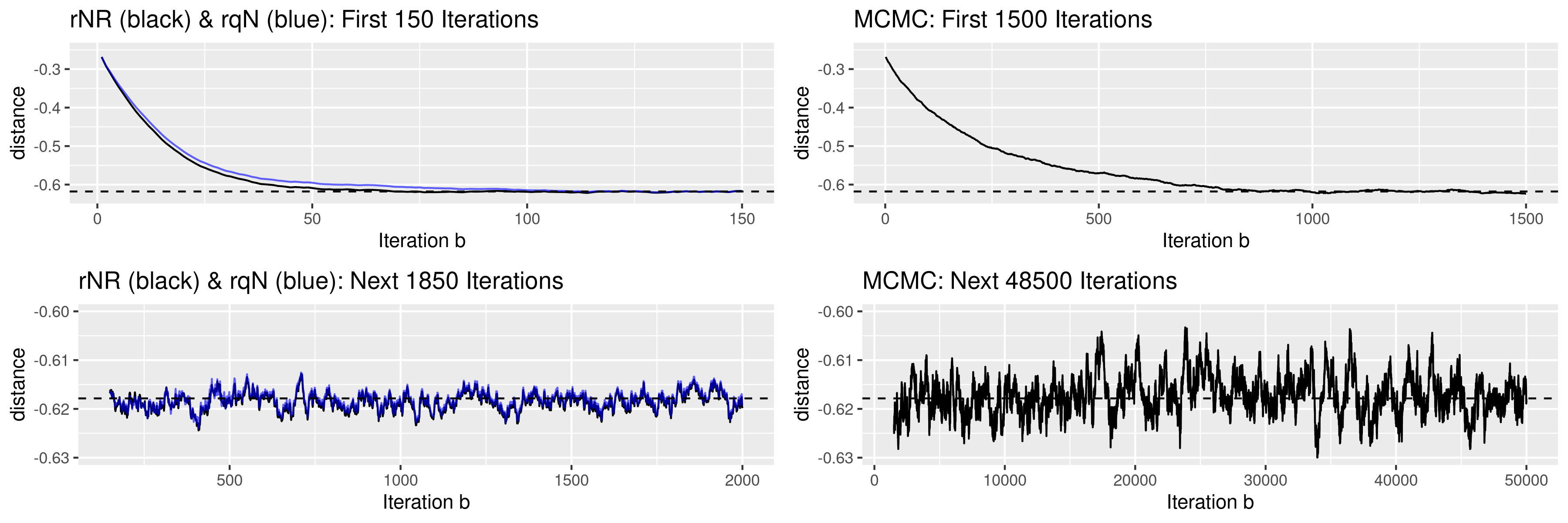}\\
{\footnotesize Legend: Solid line: \rnr (black), \rqn (blue), and \textsc{mcmc} (black) draws. Dashed line: full sample MLE estimate.}
\end{figure}

\subsection{Example 3: NLS Estimation of Transportation Costs}

The third and final example replicates an estimation in \citet{donaldson2018} of the transportation costs of salt in India between 1861-1930. The model is a simple NLS regression using the log-price of salt of type $o$ at destination $d$ in year $t$:
\[ \log(\text{price}_{dt}) = \beta_0 + \delta \log( \text{LCRED}(R_t,\alpha)_{odt} ) + \text{controls} + \varepsilon_{odt}. \]
The main regressor $\text{LCRED}(R_t,\alpha)_{odt}$, measures the lowest-cost route effective distance between origin $o$ and destination $d$ in year $t$ given the transportation network $R_t$ and the transportation costs $\alpha = (\alpha_{\text{rail}},\alpha_{\text{road}},\alpha_{\text{river}},\alpha_{\text{coast}})$ for each mode of transportation for the salt products.\footnote{The normalization $\alpha_{\text{rail}}=1$ is used so that all costs are relative to rail.} The regressor $\text{LCRED}(R_t,\alpha)_{odt}$ is not observed directly. For a given value $\alpha$, it is constructed using Dijkstra’s shortest-path algorithm applied to each triplet $o,d,t$ to find the most cost-effective route along the transportation network $R_t$. The parameters of interest are transportation costs $\alpha$, inferred from the regression and the price elasticity $\delta$.

This particular example is interesting for two reasons. First, the model is fairly difficult and time-consuming to estimate. \citet{donaldson2018} uses a grid search for both estimation and bootstrap inference in a 100 core cluster environment.\footnote{The estimation relied on a coarse grid $(\alpha_{\text{road}},\alpha_{\text{river}},\alpha_{\text{coast}}) \in [1,10]^3$ with step size $0.125$ for a total of 373248 grid points, and bootstrap inference was conducted using 200 replications on a coarser grid with step size 0.5 totalling in 5832 grid points.} \rnr\, runs in less than 6 hours on a desktop computer.  Second, the model is potentially unidentified for several values. This is the case if $\delta=0$, or for $\delta \neq 0$ when a transportation mode systematically dominates the others, e.g. if $\alpha_{\text{road}}$ is small and $\alpha_{\text{river}},\alpha_{\text{coast}}$ are large. The specific values for which this occurs depend on the structure of the transportation network $R_t$. A quadratic penalty is added to the objective: $\text{pen}(\theta) = \lambda ( \|\alpha-\bar{\alpha}\|_2^2 + \delta^2)$, with $\lambda = 0.1/n$ and $\bar{\alpha} = (\bar{\alpha}_{\text{road}},\bar{\alpha}_{\text{river}},\bar{\alpha}_{\text{coast}}) = (4.5,3.0,2.25)$ are observed historical relative
freight rate estimates reported in \citet[p916]{donaldson2018}. Notice that if e.g. $\alpha_{\text{river}}$ is unidentified, then $\partial_{\alpha_{\text{river}}}Q_m^{(b)}(\theta)=0$ around $\hat\theta_n$ so that only the penalty determines the path of $\alpha_{b,\text{river}}$, and, after convergence, the draws are degenerate $\alpha_{b,\text{river}}=\bar{\alpha}_{\text{river}}$. With a penalty, degeneracy is indicative that the coefficient is unindentified. This is not the case for a Bayesian posterior.

\begin{figure}[ht] \caption{Transportations Costs: Estimation and Inference using \rnr} \label{fig:Rails}
      \includegraphics[scale=0.55]{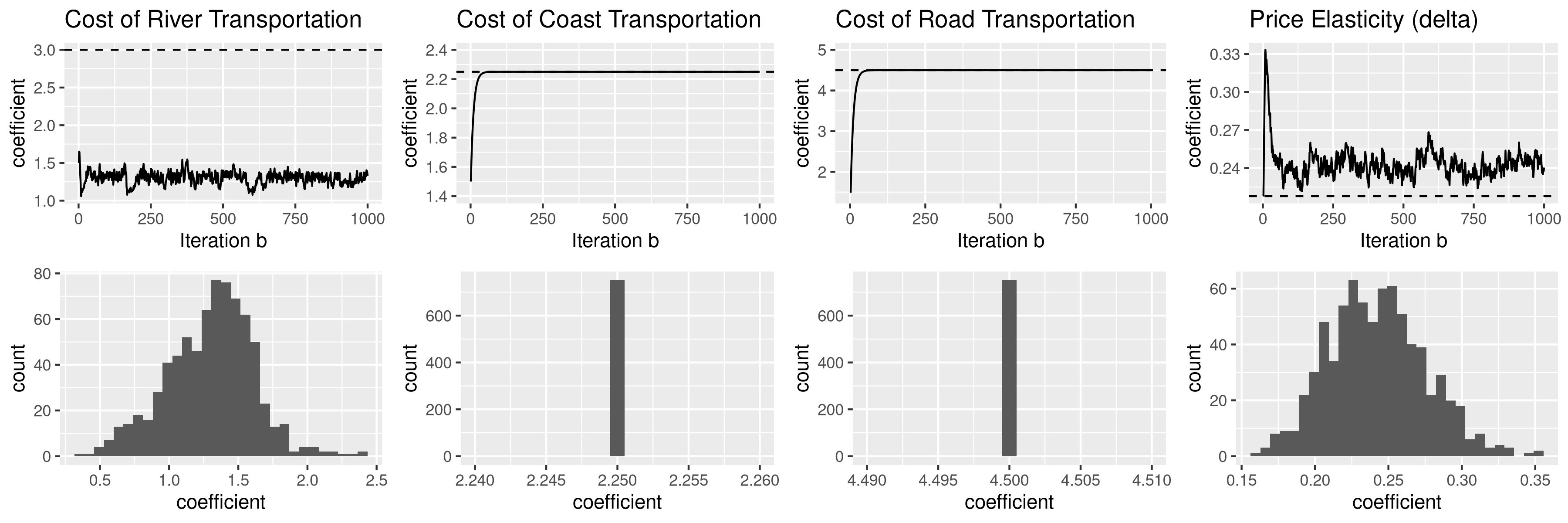}\\
      {\footnotesize Legend: Top panel: solid line \rnr\, draws for transportation costs $\alpha$ and price elasticity $\delta$. Dashed line: average historical costs $\bar{\alpha}_j$ and \textsc{ols} estimate $\hat\delta_{\textsc{ols}}(\bar{\alpha})$. Bottom panel: distribution after discarding the first $b=75$ iterations, and adjusting for $\phi(\gamma)$. }
\end{figure}
Figure \ref{fig:Rails} shows the draws and the adjusted \rnr\, distribution for $(\alpha,\delta)$. Baseline OLS results evaluated for $\alpha= \bar{\alpha}$ are reported with bootstrap standard errors clustered at the district level -- indexed by $d$. There are also $142$ origin/destination fixed effects, a linear time trend and the intercept, totalling in $147$ parameters to be estimated using 7345 observations.\footnote{\citet{donaldson2018} includes origin/time fixed effects instead of a linear time trend but the design matrix is singular with these controls in the baseline OLS regression. Hence, they are replaced with the linear trend.}  \rnr\, converges quickly in about 75 iterations. 1000 draws are produced in 5h40m on a desktop computer with 7000 calls to Dijkstra’s algorithm, the main computational bottleneck.\footnote{For reference, a grid search using a small set of 19 equally spaced points for each $\alpha_j$ requires $19^3=6859$ calls to the algorithm.} The bottom panel shows the distribution of the draws after discarding the first 75 iterations, and adjusting for $\phi(\gamma)$. The distributions for coast and road are degenerate, suggesting they are not identified. The results are not sensitive to the choice of $\theta_0$ and $\lambda$. Table \ref{tab:Rails} reports estimates and standard errors with and without fixed effects. The price elasticity is very similar to the one estimated by OLS when fixing $\alpha=\bar{\alpha}$. The estimated coefficient $\alpha_{\text{river}}$ is less than half the historical average $\bar{\alpha}_{\text{river}}$. The null that river transportation is more cost effective than rail, $\alpha_{\text{river}}<1$, cannot be rejected at the 5\% significance level.

\begin{table}[H]  \caption{Transportations Costs: Estimates and Standard Errors} \label{tab:Rails}
      \begin{center}
      \begin{tabular}{l|c|cccc||c|cccc}
      \hline \hline
      & \multicolumn{1}{c|}{\textsc{ols}} & \multicolumn{4}{c||}{\rnr}& \multicolumn{1}{c|}{\textsc{ols}} & \multicolumn{4}{c}{\rnr}\\
      \hline
      & $\delta$ & $\alpha_{\text{river}}$ & $\alpha_{\text{coast}}$ & $\alpha_{\text{road}}$ & $\delta$ & $\delta$ & $\alpha_{\text{river}}$ & $\alpha_{\text{coast}}$ & $\alpha_{\text{road}}$ & $\delta$ \\ 
      \hline
      Estimates &  0.22 & 1.29 & 2.25 & 4.50 & 0.24 & 0.22 & 1.86 & 2.25 & 4.50 & 0.22 \\ 
      Standard Errors &  0.02 & 0.34 & 0.00 & 0.00 & 0.03 & 0.04 & 0.61 & 0.00 & 0.00 & 0.04\\
      \hline
      Controls? & \multicolumn{5}{c||}{Yes} & \multicolumn{5}{c}{No}\\
      \hline \hline
      \end{tabular}\end{center}
      {\footnotesize Legend: \textsc{ols}: estimated using $\alpha = \bar{\alpha}$, the observed historical relative freight rate estimates. \rnr: $\gamma=0.1$. Both: standard errors computed with Gaussian multiplier reweighting, clustered at the district level. Controls: origin/destination fixed effects.}
\end{table}

\section{Conclusion}
This paper proposes Algorithms that produce estimates, bootstrap confidence intervals and standard errors in a single run. While the theory relies on convexity of the sample and resample objectives, Appendix \ref{sec:discussion} discusses implementation and convergence in some non-convex settings. In Examples 1 and 3, \rnr\, and \rqn\, perform well without convexity. Extending the Algorithms to handle general non-convex problems would be a useful avenue for future research. \citet{Forneron2022} shows how to modify a local gradient-based algorithm to design a fast globally convergent one without using multiple starts. This could also be very useful for sampling, as considered here. 

\newpage
\baselineskip=12.0pt
\bibliographystyle{ecta}
\bibliography{refs}

\begin{appendices}
      \renewcommand\thetable{\thesection\arabic{table}}
      \renewcommand\thefigure{\thesection\arabic{figure}}
      \renewcommand{\theequation}{\thesection.\arabic{equation}}
      \renewcommand\thelemma{\thesection\arabic{lemma}}
      \renewcommand\thetheorem{\thesection\arabic{theorem}}
      \renewcommand\thedefinition{\thesection\arabic{definition}}
        \renewcommand\theassumption{\thesection\arabic{assumption}}
      \renewcommand\theproposition{\thesection\arabic{proposition}}
        \renewcommand\theremark{\thesection\arabic{remark}}
        \renewcommand\thecorollary{\thesection\arabic{corollary}}
\setcounter{equation}{0}
\setcounter{lemma}{0}
\clearpage \baselineskip=18.0pt
\appendix
\section{Preliminary Results} \label{apx:inter}

\begin{lemma}[Taylor Approximation in $L_p$-norm] \label{lem:lips_res} Suppose Assumptions \ref{ass:A1}-\ref{ass:A2} hold for $p \geq 2$ then for any $\theta \in \Theta$:
\[ \left[ \mathbb{E}^\star \left( \| G_m^{(b)}(\theta) - G_m^{(b)}(\hat\theta_n) - H_n(\hat\theta_n)(\theta-\hat\theta_n) \|_2^p \right) \right]^{1/p} \leq \frac{C_3}{\sqrt{m}}\|\theta-\hat\theta_n\|_2 + C_1 \|\theta-\hat\theta_n\|_2^2. \]
\end{lemma}

\begin{lemma}[Asymptotic Normality] \label{lem:asym_normal} Suppose the Assumptions for Theorem \ref{th:asym_norm} i. hold then  as $m,n \to \infty$ with $\log(m)/b \to 0$: $\sqrt{\frac{m}{\phi(\gamma)}} V_n^{-1/2} (\theta_b - \hat\theta_n) \overset{d^\star}{\to} \mathcal{N}(0,I_d).$
\end{lemma}

\section{Proofs for the Main Results} \label{apx:proofs}
\paragraph{Proof of Lemma \ref{lem:cv_stochastic}:} For any $\theta_b$, we have $\thetabf-\hat\theta_n = \thetab-\hat\theta_n-\gamma \pb G_n(\thetab) + \gamma \pb[G_n(\thetab)-G_m^{(b)}(\thetab)]$. The first term is the classical update covered in Lemma \ref{lem:cv_non_stochastic}, the last term is uniformly bounded in $L_p$-norm using Assumption \ref{ass:A2} i. and \ref{ass:A2}' iv.
Taking the $\|\cdot\|_2$ norm on both sides and applying the triangular inequality that:
\begin{align*}
      \|\thetabf-\hat\theta_n\|_2 &\leq \|\thetab-\hat\theta_n-\gamma \pb G_n(\thetab)\|_2 +\gamma\overline{\lambda}_P \left( \sup_{\theta\in\Theta} \|G_n(\thetab)-G_m^{(b)}(\thetab)\|_2 \right)\\
      &\leq (1-\overline{\gamma})\| \thetab - \hat\theta_n \|_2 +\gamma\overline{\lambda}_P \left( \sup_{\theta\in\Theta} \|G_n(\thetab)-G_m^{(b)}(\thetab)\|_2 \right),
\end{align*}
where the last inequality comes from Lemma \ref{lem:cv_non_stochastic}. Taking expectations on both sides:
\begin{align*}
      \left[ \mathbb{E}^\star \left( \|\thetabf-\hat\theta_n\|_2^p\right)\right]^{1/p}
      &\leq (1-\overline{\gamma}) \left[ \mathbb{E}^\star \left( \| \thetab - \hat\theta_n \|_2^p \right) \right]^{1/p} +\frac{\gamma\overline{\lambda}_P C_3}{\sqrt{m}}\\
      &\leq (1-\overline{\gamma})^b \left[ \mathbb{E}^\star \left( \| \theta_0 - \hat\theta_n \|_2^p \right) \right]^{1/p} +\frac{\gamma\overline{\lambda}_P C_3}{\overline{\gamma}\sqrt{m}},
\end{align*}
where the last inequality follows from iterating on the contraction and the identity $\sum_{j=0}^{+\infty} (1-\overline{\gamma})^j = 1/\overline{\gamma}$. The desired result is then obtained using $C_4 = \gamma\overline{\lambda}_P C_3$.\qed

\paragraph{Proof of Proposition \ref{prop:coupling}:}
From the definition of $\thetab$ and $\thetab^\star$, their difference is:
\begin{align*}
      \thetabf-\thetabf^\star &= \left(\thetab - \gamma \pb G_m^{(b+1)}(\thetab) \right) - \left( \hat\theta_n + \Psi_n(\thetab^\star - \hat\theta_n) - \gamma \overline{P}_m G_m^{(b+1)}(\hat\theta_n) \right) \nonumber\\
      &= \Psi_n(\thetab-\thetab^\star)  + (I_d-\Psi_n)(\thetab-\hat\theta_n)  -\gamma \left( \pb G_m^{(b+1)}(\thetab) - \overline{P}_m G_m^{(b+1)}(\hat\theta_n) \right)\\
&= \Psi_n(\thetab-\thetab^\star)+\gamma\overline{P}_m H_n(\hat\theta_n)[\thetab-\hat\theta_n]-
\gamma \left( \pb G_m^{(b+1)}(\thetab) - \overline{P}_m G_m^{(b+1)}(\hat\theta_n) \right)
\end{align*}
where the third equality follows from the fact that $I_d-\Psi_n = \gamma\overline{P}_m H_n(\hat\theta_n)$.  
Now note that by construction,
\begin{align}
      \gamma &\left( \pb G_m^{(b+1)}(\thetab) - \overline{P}_m G_m^{(b+1)}(\hat\theta_n) \right)  -\gamma \overline{P}_m H_n(\hat\theta_n)[\thetab - \hat\theta_n] \nonumber \\ &\quad = \gamma\overline{P}_m \left( G_m^{(b+1)}(\thetab)-G_m^{(b+1)}(\hat\theta_n)- H_n(\hat\theta_n)[\thetab - \hat\theta_n]\right) \label{eq:d3}
      \\ &\quad +\gamma \left( \pb - \overline{P}_m\right) \left( G_m^{(b+1)}(\thetab)-G_m^{(b+1)}(\hat\theta_n) \right). \label{eq:d4}
\end{align}
Assumption \ref{ass:A3} i. implies that $\overline{P}_m = \gamma^{-1}(I_d-\Psi_n)[H_n(\hat\theta_n)]^{-1}$ has bounded eigenvalues, without loss of generality assume $\underline{\lambda}_P \leq \lambda_{\min}(\overline{P}_m) \leq \lambda_{\max}(\overline{P}_m) \leq \overline{\lambda}_P$\footnote{This is satisfied for \rnr\, and \rqn\, with $(\underline{\lambda}_P,\overline{\lambda}_P)=(\underline{\lambda}_H,\overline{\lambda}_H)$.} Together with Lemma \ref{lem:lips_res} this implies that conditional on $P_b,\theta_b$:
\begin{align*}
\left[ \mathbb{E}^\star(\| (\ref{eq:d3}) \|^{p/2}_2 | \theta_b) \right]^{2/p} &\leq \gamma \overline{\lambda}_P\left[\frac{C_3}{\sqrt{m}}\|\theta_b-\hat\theta_n\|_2 + C_1\|\theta_b-\hat\theta_n\|_2\right],
\end{align*}
the conditioning is possible because resampling and reweighting are done independently between iterations $b$ and $b+1$. Now using the law of iterated expecations and Lemma \ref{lem:cv_stochastic}:
\begin{align*}
      \left[ \mathbb{E}^\star(\| (\ref{eq:d3}) \|^{p/2}_2 ) \right]^{2/p} &\leq \gamma \overline{\lambda}_P \left[\frac{C_3}{\sqrt{m}}\left[ \mathbb{E}^\star(\|\theta_b-\hat\theta_n\|_2^{p/2}\right]^{2/p} + C_1\left[ \mathbb{E}^\star(\|\theta_b-\hat\theta_n\|^p_2)\right]^{2/p}\right]\\
      &\leq \gamma \overline{\lambda}_P\left[\frac{C_3}{\sqrt{m}}\left( (1-\overline{\gamma})^b d_{0,n} + \frac{C_4}{\overline{\gamma}\sqrt{m}} \right) + C_1\left( (1-\overline{\gamma})^b d_{0,n} + \frac{C_4}{\overline{\gamma}\sqrt{m}} \right)^2\right].
      \end{align*}
By Assumptions \ref{ass:A2} iii., \ref{ass:A2}', \ref{ass:A3} ii., Lemma \ref{lem:cv_stochastic},  mean-value theorem, and  Cauchy-Schwarz inequality, 
\begin{align*}
      \left[ \mathbb{E}^\star \left( \| (\ref{eq:d4}) \|^{p/2}_2\right) \right]^{2/p}& \leq \gamma\left[ \mathbb{E}^\star \left( \|\pb-\overline{P_m}\|_2^p \right) \right]^{1/p} \left[ \mathbb{E}^\star \left( \| H_m^{(b+1)}(\tilde{\theta}_b)(\thetab - \hat\theta_n)\|_2^p \right) \right]^{1/p}\\
      &\leq \gamma\overline{\lambda}_HC_5 \left( \overline{\rho}^b d_{0,n} + \frac{1}{\sqrt{m}} \right)  \left( (1-\overline{\gamma})^{b+1} d_{0,n} + \frac{C_5}{\overline{\gamma}\sqrt{m}} \right),
\end{align*}
where $\tilde \theta_b$ is some intermediate value between $\thetab$ and $\hat \theta_n$, and an upper bound defined in terms of $\bar\rho = \max[1-\overline{\gamma},\rho,\lambda_{\max}(\Psi_n)]<1$ to simplify notation.

The two bounds lead to the following recursion on the coupling distance:
\begin{align*}
      \left[ (\mathbb{E}^\star \left( \|\thetabf - \thetabf^\star \|_2^{p/2}\right)  \right]^{2/p} &\leq \overline{\rho} \left[\mathbb{E}^\star ( \|\thetab - \thetab^\star \|_2^{p/2} ) \right]^{2/p} + \left[ (\mathbb{E}^\star (\| (\ref{eq:d3}) \|_2^{p/2}) \right]^{2/p} + \left[\mathbb{E}^\star (\| (\ref{eq:d4}) \|_2^{p/2} ) \right]^{2/p}\\
      &\leq \overline{\rho} \left[\mathbb{E}^\star ( \|\thetab - \thetab^\star \|_2^{p/2} ) \right]^{2/p} + C_{B_{1,2}}\left(\overline{\rho}^b [d_{0,n}+d_{0,n}^2] + \frac{1}{m}\right)\\
      &\leq \frac{C_{B_{1,2}}}{1-\overline{\rho}} \left( \overline{\rho}^b [d_{0,n}+d_{0,n}^2] + \frac{1}{m}\right),
\end{align*}
where $C_{B_{1,2}}$ is a constant which depends on the terms used to bound (\ref{eq:d3}) and (\ref{eq:d4}).   Recall that $\theta_0 = \theta_0^\star$ so that the coupling distance is zero for $b=0$. Using  $C_6 = C_{B_{1,2}}/(1-\overline{\rho})$ yields the desired result.\qed

\paragraph{Proof of Theorem \ref{th:average_estimation}:}
To bound the distance $\mathbb{E}^\star\left( \| \overline{\theta}_{\FL}^\star - \hat\theta_n \|_2 \right)$, we use the recursive representation (\ref{eq:lin}) given below:
\begin{align}
      \thetabf^\star - \hat\theta_n = \Psi_n (\thetab^\star - \hat\theta_n) - \gamma \overline{P}_m G_m^{(b+1)}(\hat\theta_n), \quad \theta_0^\star = \theta_0.\tag{\ref{eq:lin}}
\end{align}
It can be re-written as:
\begin{align}
      \thetabf^\star - \hat\theta_n = \Psi_n^{b} (\theta_0 - \hat\theta_n) - \gamma \sum_{j=0}^b \Psi_n^{j} \overline{P}_m G_m^{(b+1-j)}(\hat\theta_n),\tag{\ref{eq:lin}'}
\end{align}
and take the average $\overline{\theta}_{\FL}^\star = 1/B \sum_{b=1^B}\theta^\star_b$:
\begin{align*}
      \overline{\theta}_{\FL}^\star- \hat\theta_n &=  \frac{1}{B} \sum_{b=1}^B\Psi_n^{b-1}(\theta_0-\hat\theta_n) -\gamma \frac{1}{B} \sum_{b=1}^B \sum_{j=0}^{b-1} \Psi(\hat\theta_n)^j\overline{P}_m G_m^{(b-j)}(\hat\theta_n).
\end{align*}
Assumption \ref{ass:A3} i. implies that $\|\Psi(\hat\theta_n)^{b}(\theta_0-\hat\theta_n)\|_2 \leq \overline{\rho}^{b}\|\theta_0-\hat\theta_n\|_2$, so the first term is less than $\frac{d_{0,n}}{(1-\overline{\rho})B}$ in expectation. Consider now the second term. Using the relationship between $\|\cdot\|_2$ and the Frobenius $\|\cdot\|_F$ norms, \citep[][p7]{bhatia2013}, the zero-mean condition (Assumption \ref{ass:A2} i.), and the conditional iid assumption, we have:
\begin{align*}
      \mathbb{E}^\star \left( \|\frac{1}{B} \sum_{b=1}^B \sum_{j=0}^{b-1} \Psi_n^j\overline{P}_m G_m^{(b-j)}(\hat\theta_n)\|_2^2 \right) &\leq \mathbb{E}^\star \left( \|\frac{1}{B} \sum_{b=1}^B \sum_{j=0}^{b-1} \Psi_n^j\overline{P}_m G_m^{(b-j)}(\hat\theta_n)\|_F^2 \right)\\
      & = \text{trace}\left( \text{var}^\star \left[ \frac{1}{B} \sum_{b=1}^B \sum_{j=0}^{b-1} \Psi_n^j\overline{P}_m G_m^{(b-j)}(\hat\theta_n) \right] \right)\\
      &= \frac{1}{B^2} \sum_{b=1}^B  \text{trace}\left( \text{var}^\star \left[ \sum_{j=0}^{B-b} \Psi_n^{j} \overline{P}_m G_m^{(b)}  \right] \right)\\
      &= \frac{1}{B^2 m} \sum_{b=1}^B  \text{trace}\left( \left[\sum_{j=0}^{B-b}\Psi_n^{j} \right] \overline{P}_m \Sigma_n  \overline{P}_m^\prime \left[ \sum_{j=0}^{B-b}\Psi_n^{j} \right]^\prime \right)\\
      &\leq \frac{\overline{\lambda}_P^2 \text{trace}(\Sigma_n)}{(1-\overline{\rho})^2} \frac{1}{mB}, \end{align*}
where $\Sigma_n = m \times \text{var}^\star[G_m^{(b)}(\hat\theta_n)]$ is bounded by assumption.
Putting the two inequalities together, we have the first result:
\[ \mathbb{E}^\star \left( \| \overline{\theta}_{\FL} - \hat\theta_n \|_2 \right) \leq \frac{d_{0,n}}{(1-\overline{\rho})B} + \frac{\overline{\lambda}_P}{1-\overline{\rho}} \frac{\sqrt{\text{trace}(\Sigma_n)}}{\sqrt{mB}}. \]
The second result is then derived using this together with Proposition \ref{prop:coupling}:
\begin{align*}
      \mathbb{E}^\star \left( \| \overline{\theta}_{\FL} - \hat\theta_n \|_2 \right) 
      &\leq \mathbb{E}^\star \left( \| \overline{\theta}^\star_{\FL} - \hat\theta_n \|_2 \right) + \mathbb{E}^\star \left( \| \overline{\theta}_{\FL} - \overline{\theta}_{\FL}^\star \|_2 \right)\\
      &\leq \mathbb{E}^\star \left( \| \overline{\theta}^\star_{\FL} - \hat\theta_n \|_2 \right) + \frac{1}{B} \sum_{b=1}^B  \mathbb{E}^\star \left( \| \thetab - \thetab^\star \|_2 \right)\\
      &\leq \frac{\overline{\lambda}_P \sqrt{\text{trace}(\Sigma_n)}}{1-\overline{\rho}} \frac{1}{\sqrt{mB}} + C_6 \left( \frac{1}{m} + \frac{d_{0,n} + d_{0,n}^2}{(1-\overline{\rho})B} \right),
\end{align*}
which is a $o(\frac{1}{\sqrt{n}})$ when $\frac{\sqrt{n}}{\min(m,B)} \to 0$ and $d_{0,n} = O(1)$. This implies the second result:
\[ \sqrt{n}( \overline{\theta}_{\FL} - \theta^\dagger ) = \sqrt{n}( \hat\theta_n - \theta^\dagger ) + o_p^\star(1). \]
\qed
\paragraph{Proof of Theorem \ref{th:asym_norm}:} First, note that Lemma \ref{lem:asym_normal} implies that $\sqrt{m} V_n^{-1/2} (\tilde{\theta}_b - \hat\theta_n) \overset{d^\star}{\to} \mathcal{N}(0,I_d)$. Under Assumption \ref{ass:A5} and Lemma \ref{lem:cv_stochastic} we have:
\begin{align*}
      \frac{h(\tilde{\theta}_b) - h(\hat\theta_n)}{\sqrt{\nabla h(\hat\theta_n) V_n \nabla h(\hat\theta_n)^\prime/n}} =\frac{\nabla h(\hat\theta_n)(\tilde{\theta}_b - \hat\theta_n)}{\sqrt{\nabla h(\hat\theta_n) V_n \nabla h(\hat\theta_n)^\prime/n}} + o_p^\star(1) \overset{d^\star}{\to} \mathcal{N}\left(0,1\right).
\end{align*}
Convergence in distribution implies convergence of the CDF \citep[][p323]{billingsley2013}, which in turn implies convergence of quantiles at every point of continuity of the limiting CDF \citep[][Lem21.2]{vaart_1998}. By continuity everywhere of the normal CDF this implies that for any $\alpha \in (0,1)$:
\[ \frac{c_{h,b}(\alpha)}{\sqrt{ \nabla h(\hat\theta_n) V_n \nabla h(\hat\theta_n)^\prime/n }} -  q(\alpha) \overset{p}{\to} 0,\]
where $q(\alpha)$, $c_{h,b}(\alpha)$ are the $\alpha$-th quantile of a standard Gaussian and $\tilde\theta_b - \hat\theta_n$, respectively. Apply this to $\alpha/2$ and $1-\alpha/2$ to get the first result:
\begin{align*}
      &\mathbb{P}_n \bigg(  c_{h,b}(\alpha/2) \leq h(\theta^\dagger) - h(\hat\theta_n) \leq c_{h,b}(1-\alpha/2) \bigg)\\ &= \mathbb{P}_n \bigg(  q(\alpha/2) + o_p(1) \leq \frac{\sqrt{n}[h(\theta^\dagger)-h(\hat\theta_n)]}{\sqrt{ \nabla h(\hat\theta_n) V_n \nabla h(\hat\theta_n)^\prime}} \leq q(1-\alpha/2) + o_p(1) \bigg) \to 1-\alpha,
\end{align*}
using Assumption \ref{ass:A5} i. to get the desired limit on the right-hand-side.\\

\noindent For the second result, note that we have:\footnote{The equality comes from $\text{var}(X) = \mathbb{E}(X^2)-\mathbb{E}(X)^2$ for a random variable $X$, and the identity $a^2-b^2 = (a-b)(a+b)$ for any numbers $a,b$. The inequality follows from Cauchy-Schwarz and the triangular inequality.}
\begin{align}
      &m \times \|\text{var}^\star[h( \theta_b)] -  \text{var}^\star[h( \theta_b^\star)]\| \label{eq:subst_vars}\\ &= m \times \| \mathbb{E}^\star \left( [h(\theta_b)-h(\theta_b^\star)][h(\theta_b)-h(\hat\theta_n)+h(\theta_b^\star)-h(\hat\theta_n)] \right)  \notag\\&\quad- \mathbb{E}^\star \left( h(\theta_b)-h(\theta_b^\star)\right) \mathbb{E}^\star \left( h(\theta_b)-h(\hat\theta_n)+h(\theta_b^\star)-h(\hat\theta_n) \right) \| \notag\\
      &\leq 2m \times \left[ \mathbb{E}^\star \left( \|h(\theta_b)-h(\theta_b^\star)\|_2^2 \right) \right]^{1/2} \left( \left[ \mathbb{E}^\star \left( \|h(\theta_b)-h(\hat\theta_n)\|_2^2 \right) \right]^{1/2} + \left[ \mathbb{E}^\star \left( \|h(\theta_b^\star) - \hat\theta_n\|_2^2 \right) \right]^{1/2} \right). \notag
\end{align}    
Then using Assumption \ref{ass:A5} ii. with Proposition \ref{prop:coupling} for $p=4$:
\[  \left[ \mathbb{E}^\star \left( \|h({\theta}_b) - h({{\theta}_b}^\star)\|_2^{2} \right) \right]^{1/2} \leq \|\nabla h\|_\infty C_6 \left( \frac{1}{m} + \overline{\rho}^b [d_{0,n} + d_{0,n}^2] \right), \]
where $\|\nabla h\|_\infty = \sup_{\theta \in \Theta} \max_{j=1,\dots,d_\theta} |\partial_{\theta_j} h(\theta)|$ is the largest entry-wise element, which is finite by continuity and compactness. Also, using Assumption \ref{ass:A5} ii., Lemma \ref{lem:cv_stochastic}, we have:
\[ \left[ \mathbb{E}^\star \left( \|h(\theta_b)-h(\hat\theta_n)\|_2^2 \right) \right]^{1/2} \leq  \|\nabla h\|_\infty \left( (1-\overline{\gamma})^b d_{0,n} + \frac{C_4}{\overline{\gamma}\sqrt{m}}\right).\]
Using using Assumption \ref{ass:A5} ii., equation \ref{eq:lin} with $\Psi_n = (1-\gamma)I_d$, and Assumption \ref{ass:A2} i., we have:
\[ \left[ \mathbb{E}^\star \left( \|h(\theta_b^\star)-h(\hat\theta_n)\|_2^2 \right) \right]^{1/2} \leq  \|\nabla h\|_\infty \left( (1-\gamma)^b d_{0,n} + \frac{C_2}{\sqrt{m}}\right).\]
Putting the three inequalities back into (\ref{eq:subst_vars}), we get:
\[ (\ref{eq:subst_vars}) \leq C_{(\ref{eq:subst_vars})} \left( m^{-1/2} + m \overline{\rho}^b [d_{0,n} + d_{0,n}^2] \right), \]
where $C_{(\ref{eq:subst_vars})}$ involves $C_2,C_4,C_6,\overline{\gamma}$, and $\|\nabla h\|_\infty$. Since $m^{-1/2}$ and $m \overline{\rho}^b \to 0$, the right-hand-side is negligible. Hence, we can focus on the quantity $\text{var}^\star[h(\theta_b^\star)]$. \\

\noindent Now repeat the steps above with:
\begin{align}
      &m \times \|\text{var}^\star[h( \theta_b^\star)] -  \text{var}^\star[\nabla h( \hat\theta_n)(\theta_b^\star - \hat\theta_n)]\|. \label{eq:subst_vars2}
\end{align}  
Using the mean-value theorem, the Cauchy-Schwarz inequality, and Assumption \ref{ass:A5} ii.:
\begin{align*}
      \left[ \mathbb{E}^\star \left( \|h({\theta}_b^\star) - h(\hat\theta_n) - \nabla h(\hat\theta_n)(\theta_b^\star - \hat\theta_n)\|_2^{2} \right) \right]^{1/2} &\leq C_h \left[ \mathbb{E}^\star \left( \|\theta_b^\star - \hat\theta_n)\|_2^{4} \right) \right]^{1/2}\\
      &\leq C_h \left( (1-\gamma)^b d_{0,n} + \frac{C_2}{\sqrt{m}}\right)^2,
\end{align*}
where $C_h$ is the Lipschitz constant of the gradient $\nabla h$. The bound for $[ \mathbb{E}^\star ( \|h({\theta}_b^\star) - h(\hat\theta_n) \|_2^{2} ) ]^{1/2}$ is given above and the same upper-bound applies to $[ \mathbb{E}^\star ( \|\nabla h(\hat\theta_n)[{\theta}_b^\star-\hat\theta_n] \|_2^{2} ) ]^{1/2}$.  Putting everything together:
\begin{align*}
      (\ref{eq:subst_vars2}) \leq C_{(\ref{eq:subst_vars2})} \left( m^{-1/2} + m (1-\gamma)^b d_{0,n} \right),
\end{align*}
which is also negligible; $C_{(\ref{eq:subst_vars2})}$ depends on $C_h,\|\nabla h\|_\infty,$ and $C_2$.\\

\noindent From equation (\ref{eq:lin}) and $\text{var}^\star( \sqrt{m}[H_n(\hat\theta_n)]^{-1} G_m^{(b)}(\hat\theta_n)  ) = V_n$, is it immediate that:
\begin{align*}
      m \times \text{var}^\star \left( \nabla h(\hat\theta_n)( \theta_b^\star - \hat\theta_n - (1-\gamma)^b(\theta_0-\hat\theta_n) ) \right) &= \gamma^2 \frac{1-(1-\gamma)^{b+1}}{1-(1-\gamma)^2} \nabla h(\hat\theta_n)V_n \nabla h(\hat\theta_n)^\prime\\ &= [\phi(\gamma)+ o(1)] \nabla h(\hat\theta_n)V_n \nabla h(\hat\theta_n)^\prime .
\end{align*}
Now note that: 
\begin{align*} &\Big|m \times \text{var}^\star \left( \nabla h(\hat\theta_n)( \theta_b^\star - \hat\theta_n - (1-\gamma)^b(\theta_0-\hat\theta_n) ) \right)- m \times \text{var}^\star \left( \nabla h(\hat\theta_n)( \theta_b^\star - \hat\theta_n ) \right)\Big|  \\ &\leq 2m (1-\gamma)^{2b} d_{0,n}^2 \times   \text{var}^\star \left( \nabla h(\hat\theta_n)( \theta_b^\star - \hat\theta_n ) \right) + (1-\gamma)^{2b} d_{0,n}^2 \to 0,
\end{align*}
since $\text{var}^\star ( \nabla h(\hat\theta_n)( \theta_b^\star - \hat\theta_n ) )=O(1).$ Putting everything together, we get the desired result:
\[ \frac{m}{\phi(\gamma)} \text{var}^\star \left( h(\theta_b) \right) = \nabla h(\hat\theta_n)V_n \nabla h(\hat\theta_n)^\prime + o(1). \]
\qed

\paragraph{Proof of Corollary \ref{coro:rnrqn} for \rnr:}
Assumption \ref{ass:A2} ii. implies Assumption \ref{ass:A2}' holds with $\underline{\lambda}_P = \overline{\lambda}_H^{-1}$ and $\overline{\lambda}_P = \underline{\lambda}_H^{-1}$.\\

\noindent Let $\overline{P}_m = [H_n(\hat\theta_n)]^{-1}$, $\Psi_n = I_d - \gamma \overline{P}_m H_n(\hat\theta_n) = (1-\gamma)I_d$ so Assumption \ref{ass:A3} i. holds. Since for \rnr\, we use $P_b = [H_m^{(b+1)}(\theta_b)]^{-1}$, Assumption \ref{ass:A2}' implies:
\begin{align*}
\left[ \mathbb{E}^\star \left( \|I_d - P_b [\overline{P}_m]^{-1}\|^{p/2} \right) \right]^{2/p} &\leq \overline{\lambda}_P \left[ \mathbb{E}^\star \left( \|H_m^{(b+1)}(\theta_b) - H_n(\hat\theta_n)\|^{p/2} \right) \right]^{2/p}\\
&\leq \frac{\overline{\lambda}_PC_3}{\sqrt{m}} + \overline{\lambda}_PC_1 \left[ \mathbb{E}^\star \left( \|\theta_b - \hat\theta_n\|^{p/2} \right) \right]^{2/p},
\end{align*}
using the triangular inequality, Assumptions \ref{ass:A2} ii., and \ref{ass:A1} ii. Applying Lemma \ref{lem:cv_stochastic} into the last inequality to find:
\begin{align*}
      \left[ \mathbb{E}^\star \left( \|I_d - P_b [\overline{P}_m]^{-1}\|^{p/2} \right) \right]^{2/p} &\leq \frac{\overline{\lambda}_PC_3}{\sqrt{m}} + \overline{\lambda}_PC_1 \left( (1-\overline{\gamma})^b d_{0,n} + \frac{C_4}{\overline{\gamma}\sqrt{m}} \right).
\end{align*}
Set $\rho = (1-\overline{\gamma})$ and group terms to verify Assumption \ref{ass:A3} ii. Now all the assumptions required for Theorems \ref{th:average_estimation} and \ref{th:asym_norm} are satisfied which proves Corollary \ref{coro:rnrqn} i-ii. 
\qed

\paragraph{Proof of Corollary \ref{coro:rnrqn} for \rqn:} The proof follows the same outline as the proof of Corollary \ref{coro:rnrqn} for \rnr\, above. The main challenge is to show that Assumption \ref{ass:A2}' holds under Assumptions \ref{ass:A1}, \ref{ass:A2}. Then under the conditions of Lemma \ref{lem:cv_stochastic} establish that Assumption \ref{ass:A3} holds with $\overline{P}_m = [H_n(\hat\theta_n)]^{-1}$. Then, all the desired results will follow.\\

\noindent \textbf{1) Verifying Assumption \ref{ass:A2}'.} First, recall the quantities involved in the quasi-Newton least-squares update:
\[ P_b = (\hat H_b^\prime \hat H_b + \tau_b^2 I_d)^{-1/2}, \quad \hat H_b = Y_b^\prime S_b (S_b^\prime S_b)^{-1}, \]
with $S_b = (s_b,\dots,s_{b-L+1})^\prime$, $s_{b-j} = (\theta_{b-j}-\theta_{b-j-1})/\|\theta_{b-j}-\theta_{b-j-1}\|_2$, and $Y_b = (y_b,\dots,y_{b-L+1})^\prime$, $y_{b-j} = H_m^{(b+1-j)}(\theta_{b-j})s_{b-j}$. The modification term is $\tau_b = \max[\underline{\lambda}_H/2 - \sigma_{\min}(\hat H_b),0]$, where $\sigma_{\min}(\hat H_b)$ is the smallest singular value of $\hat H_b$.

By positive semi-definite transform, $\lambda_{\min}(\hat H_b^\prime \hat H_b) \geq 0$ and $\lambda_{\min}([\hat H_b^\prime \hat H_b + \tau_b^2 I_d]^{1/2}) \geq \underline{\lambda}_H/2$. This implies the first inequality: $\lambda_{\max}(P_b) \leq 2 \underline{\lambda}_H^{-1} < \infty$. For the lower bound, note that $\hat H_b^\prime\hat H_b = (S_b^\prime S_b)^{-1} S_b^\prime  Y_b Y_b^\prime S_b (S_b^\prime S_b)^{-1}$ from which we can derive:
\[ \lambda_{\max}(\hat H_b^\prime\hat H_b) \leq \frac{\lambda_{\max}(S_b^\prime S_b/L)}{[\lambda_{\min}(S_b^\prime S_b/L)]^2} \lambda_{\max}(Y_bY_b^\prime/L), \]
where $\lambda_{\min}(S_b^\prime S_b/L) \geq \underline{\lambda}_S >0$ is enforced in the Algorithm. By normalization of the $s_{b-j}$, $\lambda_{\max}(S_b^\prime S_b/L) \leq \text{trace}(S_b^\prime S_b/L)=1$. Similarly:
\[  \lambda_{\max}(Y_bY_b^\prime/L) \leq \text{trace}(Y_bY_b^\prime/L) = \frac{1}{L}\sum_{j=0}^{L-1} s^\prime_{b-j} H_m^{(b-j)}(\theta_{b-j}) H_m^{(b-j) \prime}(\theta_{b-j}) s_{b-j} \leq \overline{\lambda}_H^2, \]
using the upper bound in Assumption \ref{ass:A2} ii. Hence $\lambda_{\max}(\hat H_b^\prime\hat H_b) \leq [\overline{\lambda}_H/\underline{\lambda}_S]^2$ and $\lambda_{\min}(P_b) \geq ([\overline{\lambda}_H/\underline{\lambda}_S]^2 + \underline{\lambda}_H^2/4)^{-1/2}$. Altogether, Assumption \ref{ass:A2}' holds with:
\[ 0 < \underline{\lambda}_P = ([\overline{\lambda}_H/\underline{\lambda}_S]^2 + \underline{\lambda}_H^2/4)^{-1/2} \leq \lambda_{\min}(P_b) \leq \lambda_{\max}(P_b) \leq 2 \underline{\lambda}_H^{-1} = \overline{\lambda}_P < \infty. \]

\noindent \textbf{2) Verifying Assumption \ref{ass:A3}.} Given that Assumptions \ref{ass:A1}, \ref{ass:A2}, and \ref{ass:A2}' hold, Lemma \ref{lem:cv_stochastic} is satisfied for an appropriate choice of $\gamma \in (0,1]$. Let $\overline{P}_m = [H_n(\hat\theta_n)]^{-1}$, the goal is to prove that:
\[ \left[ \mathbb{E}^\star \left( \| I_d - P_bH_n(\hat\theta_n) \|_2^p \right) \right]^{1/p} \leq C_5 \left(  (1-\overline{\gamma})^b d_{0,n} + m^{-1/2} \right). \]
Using Assumption \ref{ass:A2}', the identity $H_n(\hat\theta_n) = [H_n(\hat\theta_n)^\prime H_n(\hat\theta_n)]^{1/2}$, and the Ando-Hemmen inequality,\footnote{See \citet{AndoHemmen1980} and \citet[Th6.2, p135]{higham2008}, for any two positive definite matrices $A,B$, $\|A^{1/2} - B^{1/2}\| \leq [\sqrt{\lambda_{\min}(A)}+\sqrt{\lambda_{\min}(B)}]^{-1}\|A-B\|$ where $\|\cdot\|$ is a unitarily invariant norm. Recall that $\|\cdot\|_F$ is unitarily invariant.} we have under the Frobenius norm $\|\cdot\|_F$:
\begin{align*}
      \| I_d - P_bH_n(\hat\theta_n) \|_F &\leq \overline{\lambda}_P \| [\hat H_b^\prime \hat H_b + \tau_b^2 I_d]^{1/2} - [H_n(\hat\theta_n)^\prime H_n(\hat\theta_n)]^{1/2} \|_F\\
      &\leq \overline{\lambda}_P \frac{\|\hat H_b^\prime \hat H_b - H_n(\hat\theta_n)^\prime H_n(\hat\theta_n) + \tau_b^2 I_d  \|_F}{\overline{\lambda}^{-1}_P + \underline{\lambda}_H}\\
      &\leq \overline{\lambda}_P \frac{ \sigma_{\max}(\hat H_b) + \overline{\lambda}_H }{ \overline{\lambda}^{-1}_P + \underline{\lambda}_H }\| \hat H_b - H_n(\hat\theta_n) \|_F + \overline{\lambda}_P \frac{ \sqrt{d_\theta} }{\overline{\lambda}^{-1}_P + \underline{\lambda}_H}\tau_b^2.
\end{align*}
Using $\| I_d - P_bH_n(\hat\theta_n) \|_2 \leq \| I_d - P_bH_n(\hat\theta_n) \|_F$, the desired result follows from bounding $\| \hat H_b - H_n(\hat\theta_n) \|_F$ and $\tau_b$ in $L_p$-norm.\\

\noindent First, consider $\tau_b  = \max[ \underline{\lambda}_H/2 - \sigma_{\min}(\hat H_b),0 ]$. By Lipschitz-continuity of the soft-max operator $(x)^+ = \max(x,0)$, and the identity $[\underline{\lambda}_H/2 - \lambda_{\min}(H_n(\hat\theta_n))]^+=0$, we have:
\begin{align*}
      [\underline{\lambda}_H/2 - \sigma_{\min}(\hat H_b)]^+ &= | [\underline{\lambda}_H/2 - \sigma_{\min}(\hat H_b)]^+ - [\underline{\lambda}_H/2 - \lambda_{\min}(H_n(\hat\theta_n))]^+ |\\
      &\leq |\sigma_{\min}(\hat H_b) - \lambda_{\min}(H_n(\hat\theta_n))|\\
      &\leq \frac{|\lambda_{\min}(\hat H_b^\prime \hat H_b) -\lambda_{\min}(H_n(\hat\theta_n)^\prime H_n(\hat\theta_n))|}{\sigma_{\min}(\hat H_b) + \underline{\lambda}_H}\\
      &\leq \frac{ \|\hat H_b^\prime \hat H_b - H_n(\hat\theta_n)^\prime H_n(\hat\theta_n)\|_F }{\sigma_{\min}(\hat H_b) + \underline{\lambda}_H}\\
      &\leq \frac{ \sigma_{\max}(\hat H_b) + \overline{\lambda}_H }{\sigma_{\min}(\hat H_b) + \underline{\lambda}_H}\|\hat H_b - H_n(\hat\theta_n)\|_F.
\end{align*}
The second equality follows from \citet{hoffman1953}, see also \citet[p153]{bhatia2013}. Using $\tau_b^2 \leq \underline{\lambda}_H/2 $, the bound will follow from the derivations for $\|\hat H_b - H_n(\hat\theta_n)\|_F$ given below.\\

\noindent Next, recall that $\hat H_b = Y_b^\prime S_b (S_b^\prime S_b)^{-1}$, where $y_{b-j} = H_m^{(b+1-j)}(\theta_{b-j})s_{b-j} = H_n(\hat\theta_n)s_{b-j} + r_{b-j}$, where $r_{b-j} = [H_m^{(b+1-j)}(\theta_{b-j})-H_n(\hat\theta_n)]s_{b-j}$. Let $R_b = (r_b,\dots,r_{b-L+1})^\prime$, using standard OLS calculations we have:
\[ \hat H_b - H_n(\hat\theta_n) = (R_b^\prime S_b/L)(S_b^\prime S_b/L)^{-1}.  \]
Now recall that $\lambda_{\min}(S_b^\prime S_b/L) \geq \underline{\lambda}_S >0$ and $\|s_{b-j}\|_2=1$ are enforced algorithmically, also $\|\cdot\|_F \leq \sqrt{d_\theta} \|\cdot\|_2$ for matrices. Thus, we have:
\begin{align*}
      \|\hat H_b - H_n(\hat\theta_n)\|_F &\leq \underline{\lambda}_S^{-1} \frac{1}{L} \sum_{j=0}^{L-1} |r_{b-j}^\prime s_{b-j}|\\
      &\leq  \underline{\lambda}_S^{-2} \frac{1}{L} \sum_{j=0}^{L-1} \| H_m^{b-j}(\theta_{b-j}) - H_n(\hat\theta_n) \|_F\\
      &\leq \underline{\lambda}_S^{-2} \frac{\sqrt{d_\theta}}{L} \left( \sum_{j=0}^{L-1} \sup_{\theta \in \Theta} \|H_m^{(b+1-j)}(\theta)-H_n(\theta)\|_2 + C_1 \|\theta_{b-j}-\hat\theta_n\|_2 \right).
\end{align*}
Taking the $L_p$-norm on both sides, we have:
\[ \left[ \mathbb{E}^\star \left( \|\hat H_b - H_n(\hat\theta_n)\|_F^p \right) \right]^{1/p} \leq \underline{\lambda}_S^{-2}C_2\sqrt{ \frac{d_\theta}{m} } + \underline{\lambda}_S^{-2} \sqrt{d_\theta}C_1 \left( \frac{d_{0,n}}{L} \frac{(1-\overline{\gamma})^{b-L+1}}{\overline{\gamma}} + \frac{C_4}{\sqrt{m}} \right).  \]
Finally, putting everything together we get:
\begin{align*}
      &\left[ \mathbb{E}^\star \left( \|I_d - P_b H_n(\hat\theta_n)\|_2^p \right) \right]^{1/p}\\ &\leq \frac{\overline{\lambda}_P}{\overline{\lambda}_P^{-1} + \underline{\lambda}_H}(\sigma_{\max}(\hat H_b) + \overline{\lambda}_H) \left( 1 + \frac{\sqrt{d_\theta}\underline{\lambda}_H/2}{\sigma_{\min}(\hat H_b) + \underline{\lambda}_H } \right) \left[ \mathbb{E}^\star \left( \|\hat H_b - H_n(\hat\theta_n)\|_F^p \right) \right]^{1/p},
\end{align*}
which yields the desired result and concludes the proof.
\qed
\newpage

\section{Derivations for the Least-Squares Example} \label{apx:OLS}
In this example, $y_n = X_n \hat \theta_n + \hat e_n$, $Q_n(\theta) = \frac{1}{2n} (y_n  -X_n\theta)^\prime(y_n  -X_n\theta)$,  $H_n=X_n'X_n/n, G_n=-X_n^\prime \hat e_n/n$, and $Q_m^{(b)}(\theta) = \frac{1}{2m} (y_m^{(b)}  -X_m^{(b)}\theta)^\prime(y_m^{(b)}  -X_m^{(b)}\theta)$, $H_b=H^{(b+1)}_m(\theta_b)=X^{(b+1)^\prime}_mX^{(b+1)}_m/m$. $G_b(\theta) = -X_m^{(b+1)^\prime}[y_m^{(b+1)} - X_m^{(b+1)}\theta]/m$. 
Let $\hat \theta_m^{(b+1)} = (X_m^{(b+1) \prime} X_m^{(b+1)})^{-1} X_m^{(b+1) \prime} y_{m}^{(b+1)}$ be the $m$ out of $n$ bootstrap estimate. Orthogonality of least squares  residuals 
will be used repeatedly.
\paragraph{Gradient Descent}
$ \theta_{k+1} = \theta_k - \gamma \left[ -X_n^\prime (y_n-X_n\theta_k)/n \right]$.
Subtract $\hat \theta_n$ on both sides and note that  $y_n = X_n \hat \theta_n + \hat e_n$ (full sample estimates), then:
\begin{align*}
 \theta_{k+1} - \hat \theta_n &= \theta_k - \hat \theta_n - \gamma \left[ -X_n^\prime (X_n\hat\theta_n + \hat e_n -X_n\theta_k)/n \right]\\
 &= \theta_k - \hat \theta_n - ( \gamma H_n) (\theta_k - \hat \theta_n)  + \gamma X_n^\prime \hat e_n/n= (I- \gamma H_n) (\thetab - \hat \theta_n) \quad \text{since } X_n^\prime\hat e_n=0.
\end{align*}

\paragraph{Newton-Raphson} 
$ \theta_{k+1} = \theta_k - \gamma \left[ H_n \right]^{-1} \left[ -X_n^\prime (y_n-X_n\theta_k)/n \right].$
Subtract $\hat \theta_n$ on both sides:
\begin{align*}
 \theta_{k+1} - \hat \theta_n &= \theta_k - \hat \theta_n - \gamma H_n^{-1}\left[ -[X_n^\prime X_n/n][\hat \theta_n-\theta_k] +  X_n^\prime\hat e_n/n  \right]
 = (1-\gamma)(\theta_k - \hat \theta_n) \quad \text{since } X_n^\prime \hat e_n=0.
\end{align*}

\paragraph{Stochastic Gradient Descent} 
$ \thetabf = \thetab - \gamma_b \left[ - X_m^{(b) \prime} (y_m^{(b)}-X_m^{(b)}\thetab )/m \right]$. Thus
\begin{align*}
      \thetabf - \hat \theta_n &= \thetab - \hat \theta_n - \gamma_b \left[ -X_m^{ (b+1) \prime} (y_m^{(b+1)}- X_m^{(b+1)}\hat\theta_n  - X_m^{(b+1)}[\thetab-\hat\theta_n])/m \right]\\
      &= (I- \gamma_b H_b) (\thetab - \hat \theta_n)  + \gamma_b X_m^{(b+1) \prime }(y_m^{(b+1)}- X_m^{(b+1)}\hat\theta_n )/m\\
      &= (I- \gamma_b H_b) (\thetab - \hat \theta_n)  - \gamma_b G_b(\hat\theta_n)\quad \text{since } X_m^{(b+1)^\prime} \hat e_m^{(b+1)}=0.
     \end{align*}

\paragraph{Resampled Gradient Descent} 
$ \thetabf = \thetab - \gamma \left[ - X_m^{(b+1) \prime} (y_m^{(b+1)}-X_m^{(b+1)}\thetab )/m \right]$.
Subtract $\hat \theta_n$ on both sides and note that  $y_m^{(b+1)} = X_m^{(b+1)} \hat \theta_m^{(b+1)} + \hat e_m^{(b+1)}$ (bootstrap estimates). Then
\begin{align*}
      \thetabf - \hat \theta_n &= \thetab - \hat \theta_n - \gamma \left[ -X_m^{ (b+1) \prime} (X_m^{(b+1) }[\hat\theta_m^{(b+1)}-\hat\theta_n]  + \hat e_m^{(b+1)} -X_m^{(b+1) }[\thetab-\hat\theta_n])/m \right]\\
      &= \thetab - \hat \theta_n - ( \gamma H_b) (\thetab - \hat \theta_n)  + \gamma H_b (\hat\theta_m^{(b)} -\hat\theta_n)\\ &= (I- \gamma H_b) (\thetab - \hat \theta_n) + \gamma H_b (\hat\theta_m^{(b+1)} -\hat\theta_n) \quad \text{since } X_m^{(b+1)^\prime} \hat e_m^{(b+1)}=0.
     \end{align*}
\paragraph{Resampled Newton-Raphson} 
$ \thetabf = \thetab - \gamma [H_b]^{-1} \left[ - X_m^{(b+1) \prime} (y_m^{(b+1)}-X_m^{(b+1)}\thetab )/m \right]$. Then
     \begin{align*}
           \thetabf - \hat \theta_n &= \thetab - \hat \theta_n - \gamma [H_b]^{-1} \left[ -X_m^{ (b+1) \prime} (X_m^{(b+1) }[\hat\theta_m^{(b+1)}-\hat\theta_n]  + \hat e_m^{(b+1)} -X_m^{(b+1) }[\thetab-\hat\theta_n])/m \right]\\
           &= (1- \gamma) (\thetab - \hat \theta_n) + \gamma (\hat\theta_m^{(b+1)} -\hat\theta_n) \quad \text{since } X_m^{(b+1) \prime} \hat e_m^{(b+1)}=0.
          \end{align*}

\begin{titlingpage} 
      \emptythanks
      \title{ {Supplement to\\ \lQ {\bf Estimation and Inference by Stochastic Optimization}''}}
      \author{Jean-Jacques Forneron\thanks{Department of Economics, Boston University, 270 Bay State Rd, MA 02215 Email: jjmf@bu.edu}}
      \setcounter{footnote}{0}
      \setcounter{page}{0}

      \clearpage 
      \maketitle 
      \thispagestyle{empty} 
      \begin{center}
      This Supplemental Material consists of Appendices \ref{apx:prf_prelim}, \ref{apx:Rcode}, \ref{apx:primA2}, \ref{apx:primA4}, \ref{apx:add_ex}, and \ref{sec:discussion} to the main text.
      \end{center}
\end{titlingpage}

\setcounter{page}{1}

\section{Proofs for the Preliminary Results} \label{apx:prf_prelim}

\paragraph{Proof of Lemma \ref{lem:lips_res}:}
Add and substract $G_n(\theta) - G_n(\hat\theta_n)$ to the desired quantity, apply a mean-value expansion to get $G_n(\theta) - G_n(\hat\theta_n) = H_n(\tilde \theta_n)(\theta-\hat\theta_n)$ and note that $\|H_n(\tilde \theta_n)-H_n(\hat \theta_n)\| \leq C_1\|\theta-\hat\theta_n\|$ by Lipchitz continuity of the Hessian and intermediate value. Now using the triangular inequality:
\begin{align*}
      &\left[ \mathbb{E}^\star \left( \| G_m^{(b)}(\theta) - G_m^{(b)}(\hat\theta_n) - H_n(\hat\theta_n)(\theta-\hat\theta_n) \|_2^p \right) \right]^{1/p}\\ 
      &\leq \left[ \mathbb{E}^\star \left( \| [G_m^{(b)}(\theta) - G_m^{(b)}(\hat\theta_n)] - [G_n(\theta) - G_n(\hat\theta_n)] \|_2^p \right) \right]^{1/p} + C_1\|\theta-\hat\theta_n\|^2,
\end{align*}
using the intermediate value Theorem again but to the difference $G_m^{(b)}-G_n$ this time, we have: $[G_m^{(b)}(\theta) - G_m^{(b)}(\hat\theta_n)] - [G_n(\theta) - G_n(\hat\theta_n)] = [H_m^{(b)}(\tilde \theta_n)-H_n(\tilde\theta_n)](\theta-\hat\theta_n)$. Plugging this back into the inequality above, we have:
\begin{align*}
      &\left[ \mathbb{E}^\star \left( \| G_m^{(b)}(\theta) - G_m^{(b)}(\hat\theta_n) - H_n(\hat\theta_n)(\theta-\hat\theta_n) \|_2^p \right) \right]^{1/p}\\
      &\leq \left[ \mathbb{E}^\star \left( \sup_{\theta \in \Theta}\| H_m^{(b)}(\theta)-H_n(\theta) \|_2^p \right) \right]^{1/p}\|\theta-\hat\theta_n\|_2 + C_1\|\theta-\hat\theta_n\|^2
      \leq \frac{C_3}{\sqrt{m}}\|\theta-\hat\theta_n\|_2 + C_1 \|\theta-\hat\theta_n\|_2^2.
\end{align*}
\qed
\paragraph{Proof of Lemma \ref{lem:asym_normal}:}
The property that $\overline{P}_m = [H_n(\hat\theta_n)]^{-1}$ is crucial for what follows. To prove the Lemma,  first substitute $\thetab$ for the linear process $\thetab^\star$ using Proposition \ref{prop:coupling}:
\begin{align*}
      \frac{\sqrt{m}}{\sqrt{\phi(\gamma)}}V_n^{-1/2}(\thetab-\hat\theta_n) &= \frac{\sqrt{m}}{\sqrt{\phi(\gamma)}}V_n^{-1/2}(\thetab^\star-\hat\theta_n) + \frac{\sqrt{m}}{\sqrt{\phi(\gamma)}}V_n^{-1/2}(\thetab-\thetab^\star)\\
      &= \frac{\sqrt{m}}{\sqrt{\phi(\gamma)}}V_n^{-1/2}(\thetab^\star-\hat\theta_n) +o_p^\star(1),
\end{align*}
when $\log(m)/b \to 0$ since it implies $\sqrt{m}\overline{\rho}^b = \exp( b[ \log(m)/(2b)+\log(\overline{\rho})]) \to 0$, using $\log(m)/(2b)+\log(\overline{\rho}) \to \log(\overline{\rho}) < 0$. For $\overline{P}_m = [H_n(\hat\theta_n)]^{-1}$, $\Psi_n = (1-\gamma)I_d$. Using the recursion (\ref{eq:lin}), we have:
\begin{align*}
      \frac{\sqrt{m}}{\sqrt{\phi(\gamma)}}V_n^{-1/2}(\thetab^\star-\hat\theta_n) &= \frac{\sqrt{m}}{\sqrt{\phi(\gamma)}}V_n^{-1/2}(1-\gamma)^b(\theta_0-\hat\theta_n)\\ &\quad - \gamma \sum_{j=0}^{b-1} (1-\gamma)^j \frac{\sqrt{m}}{\sqrt{\phi(\gamma)}}V_n^{-1/2}[H_n(\hat\theta_n)]^{-1} G_m^{(b-j)}(\hat\theta_n).
\end{align*}
The first term is a $o_p^\star(1)$ as long as $\log(m)/b \to 0$. For the second term, since the $[H_n(\hat\theta_n)]^{-1} G_m^{(b-j)}(\hat\theta_n)$ are conditionally independent and identically distributed, we have by a convolution argument:
\begin{align*}
      &\mathbb{E}^\star \left( \exp( \mathbf{i}\tau^\prime \sqrt{m}\gamma \sum_{j=0}^{b-1} (1-\gamma)^j \frac{\sqrt{m}}{\sqrt{\phi(\gamma)}}V_n^{-1/2}[H_n(\hat\theta_n)]^{-1} G_m^{(b-j)}(\hat\theta_n) ) \right)\\ &= \prod_{j=0}^{b-1} \mathbb{E}^\star \left( \exp( \mathbf{i}\tau^\prime \sqrt{m}\gamma (1-\gamma)^j \frac{\sqrt{m}}{\sqrt{\phi(\gamma)}}V_n^{-1/2}[H_n(\hat\theta_n)]^{-1} G_m^{(b-j)}(\hat\theta_n) ) \right)\\
      &= \prod_{j=0}^{b-1} \left[ \exp \left( -\frac{\|\tau\|_2^2}{2} \frac{\gamma^2 (1-\gamma)^{2j}}{\phi(\gamma)} \right) \left( 1+ \frac{r_m( \gamma(1-\gamma)^j \tau/\phi(\gamma) )}{m^\beta} \right) \right]\\
      &= \underbrace{ \vphantom{ \prod_{j=0}^{b-1} \left[  \left( 1+ \frac{r_m( \gamma(1-\gamma)^j \tau/\phi(\gamma) )}{m^\beta} \right) \right] } \exp \left( -\frac{\|\tau\|_2^2}{2} \frac{\gamma^2 [1-(1-\gamma)^{2b}]}{[1-(1-\gamma)^{2}]\phi(\gamma)} \right)}_{=\exp(-\|\tau\|_2^2/2 \times  [1+o(1)])} \underbrace{\prod_{j=0}^{b-1} \left[  \left( 1+ \frac{r_m( \gamma(1-\gamma)^j \tau/\phi(\gamma) )}{m^\beta} \right) \right]}_{(\text{error})}.
\end{align*}
 To show that the last product is convergent under the stated assumptions, take logs and use the inequality $\frac{x}{1+x} \leq \log(1+x) \leq x$ for $x > -1$. Then
\begin{align*}
      \log \left(\|(\text{error})\| \right) &= \sum_{j=0}^{b-1} \log \left( 1+ \frac{|r_m( \gamma(1-\gamma)^j \tau/\phi(\gamma) )|}{m^\beta} \right)
      \leq \sum_{j=0}^{b-1} \frac{|r_m( \gamma(1-\gamma)^j \tau/\phi(\gamma) )|}{m^\beta}\\
      &\leq \sum_{j=0}^{b-1} \frac{ \|\gamma \tau/\phi(\gamma)\|^\kappa (1-\gamma)^{\kappa j}}{m^\beta}
      \leq \frac{ \|\gamma \tau/\phi(\gamma)\|^\kappa }{[1-(1-\gamma)^{\kappa}] m^\beta} \to 0.
\end{align*}
Note that $\frac{\gamma}{\phi(\gamma)} = 2-\gamma \geq 1$ for $\gamma \in (0,1]$.
Putting everything together we have:
\begin{align*}
      &\mathbb{E}^\star \left( \exp( \mathbf{i}\tau^\prime \frac{\sqrt{m}}{\sqrt{\phi(\gamma)}}V_n^{-1/2}(\thetab-\hat\theta_n) \right) = \exp\left(-\frac{\|\tau\|_2^2}{2} \right)\left( 1+O\left(\frac{ \|\tau\|^\kappa }{m^\beta} \frac{(2-\gamma)^{\kappa}}{[1-(1-\gamma)^{\kappa}]} \right)\right).
\end{align*}
Now since convergence of the characteristic function implies weak convergence \citep[][Th26.3, p349]{billingsley2013}, which implies the desired convergence in distribution.\qed
\section{Implementing \rnr\, in R} \label{apx:Rcode}

To illustrate how the \rnr\; is implemented in a real data setting, the following provides detailed commented R code for a probit model using the \citet{mroz:87} data.

\paragraph{Illustration using Labor Force Participation}
The table below presents the estimates and standard errors for all methods and coefficients in the \citet{mroz:87} application.

\begin{table}[H]
      \centering  \caption{Labor Force Participation: Estimates and Standard Errors} \label{tbl:table-mroz_apx} \setlength\tabcolsep{4.5pt}
      { \renewcommand{\arraystretch}{0.935} 
      \begin{tabular}{l|bbbbaaaaaa}
            \hline \hline 
            & \multicolumn{10}{c}{Estimates}\\
       & \textsc{mle}  &  &  &  &  \mc{1}{\rnr$_n$} & \mc{1}{\rnr$_{200}$} & \mc{1}{\rnr$_{100}$} & \mc{1}{r\textsc{qn}$_n$} & \mc{1}{r\textsc{qn}$_{200}$} & \multicolumn{1}{c}{r\textsc{qn}$_{100}$}  \\ 
        \hline
        nwifeinc & -0.012 & - & - & - & -0.012 & -0.013 & -0.014 & -0.012 & -0.011 & -0.012 \\ 
  educ & 0.131 & - & - & - & 0.132 & 0.138 & 0.143 & 0.131 & 0.129 & 0.129 \\ 
  exper & 0.123 & - & - & - & 0.123 & 0.124 & 0.123 & 0.123 & 0.124 & 0.125 \\ 
  exper2 & -0.002 & - & - & - & -0.002 & -0.002 & -0.002 & -0.002 & -0.002 & -0.002 \\ 
  age & -0.053 & - & - & - & -0.053 & -0.053 & -0.055 & -0.052 & -0.052 & -0.052 \\ 
  kidslt6 & -0.868 & - & - & - & -0.874 & -0.892 & -0.902 & -0.864 & -0.855 & -0.844 \\ 
  kidsge6 & 0.036 & - & - & - & 0.037 & 0.038 & 0.041 & 0.036 & 0.035 & 0.032 \\ 
  const. & 0.270 & - & - & - & 0.271 & 0.216 & 0.234 & 0.248 & 0.256 & 0.249 \\
         \hline
         & \multicolumn{10}{c}{Standard Errors} \\
         & \textsc{ase} & \textsc{boot} & \dmk & \ks & \mc{1}{\rnr$_n$} & \mc{1}{\rnr$_{200}$} & \mc{1}{\rnr$_{100}$} & \mc{1}{r\textsc{qn}$_n$} & \mc{1}{r\textsc{qn}$_{200}$} & \multicolumn{1}{c}{r\textsc{qn}$_{100}$}\\  \hline
         nwifeinc & 0.005 & 0.005 & 0.005 & 0.005 & 0.005 & 0.006 & 0.005 & 0.005 & 0.005 & 0.005 \\ 
  educ & 0.025 & 0.026 & 0.026 & 0.025 & 0.025 & 0.027 & 0.028 & 0.027 & 0.025 & 0.025 \\ 
  exper & 0.019 & 0.020 & 0.019 & 0.019 & 0.019 & 0.020 & 0.021 & 0.019 & 0.018 & 0.017 \\ 
  exper2 & 0.001 & 0.001 & 0.001 & 0.001 & 0.001 & 0.001 & 0.001 & 0.001 & 0.001 & 0.001 \\ 
  age & 0.008 & 0.009 & 0.008 & 0.008 & 0.009 & 0.008 & 0.009 & 0.009 & 0.008 & 0.008 \\ 
  kidslt6 & 0.119 & 0.120 & 0.118 & 0.118 & 0.120 & 0.119 & 0.129 & 0.117 & 0.113 & 0.117 \\ 
  kidsge6 & 0.043 & 0.046 & 0.045 & 0.045 & 0.045 & 0.048 & 0.047 & 0.044 & 0.042 & 0.045 \\ 
  const. & 0.509 & 0.512 & 0.507 & 0.505 & 0.494 & 0.535 & 0.544 & 0.544 & 0.494 & 0.506 \\ \hline  \hline
      \end{tabular} }
\end{table}


\paragraph{Sample R code to implement the Mroz Example.}
\begin{lstlisting}[language=R]
      
set.seed(123)     # set the seed
library(numDeriv) # compute numerical derivaties using finite differences, alternative: library(pracma) is usually faster
library(foreign)  # to load the data set in Stata dta format


data = read.dta('mroz.dta') # read the mroz data

y = data$inlf # outcome variable
X = cbind(data$nwifeinc,data$educ,data$exper, # regressors
            data$exper^2,data$age,data$kidslt6,data$kidsge6,1) 

colnames(X) = c('nwifeinc','educ','exper','exper2', # labels
                  'age','kidslt6','kidsge6','constant')

n = 753      # sample size
index0 = 1:n # indices for the sample data

loglik <- function(coef,index=index0) {
      # compute the log-likelihood for the Probit model on the observations indexed by index (default 1:n, the original sample) at theta = coef
      
      score = X[index,]%*%coef # compute the z-scores
      ll    = y[index]*log( pnorm(score) ) +
              (1-y[index])*log( 1-pnorm(score) ) 
      return( sum( ll ) )
}

d_loglik <- function(coef,index=index0) {
      # compute the gradient of the log-likelihood for the Probit model on the observations indexed by index (default 1:n, the original sample) at theta = coef
      # In this example, the gradient is analytically tractable, it could be evaluated by finite differences by using the following:
      # d_loglik <- function(coef,index=index0) { return(jacobian(loglik,coef,index=index)) }

      yy = y[index] # keep observations indexed by index
      XX = X[index,] # keep observations indexed by index
      score = XX%*%coef # compute the z-score
      dll   = 0 # initialize the gradient
    
      for (i in 1:length(index)) {
            dll = dll +
         (yy[i]*XX[i,]*dnorm(score[i])/pnorm(score[i]) - 
         (1-yy[i])*XX[i,]*dnorm(score[i])/(1-pnorm(score[i])))
      }
      return(dll)
}


rNR <- function(coef0, learn = 0.1, iter = 500, m = n) {
      # generate 'B = iter' rNR draws with learning rate 'gamma = learn' with m out of n resampling
      
      coefs     = matrix(NA,iter,length(coef0)) # matrix where draws will be stored
      coefs[1,] = coef0 # initialize the first-draw
    
      for (i in 2:iter) {
            index = sample(1:n,m,replace=TRUE) # sample m out of n observations with replacement
            
            G = d_loglik(coefs[i-1,],index=index) # compute the resampled gradient G using analytical derivatives. Alternative using finite differences: 
            # G = jacobian(loglik,coefs[i-1,],index=index)
            H = hessian(loglik,coefs[i-1,],index=index) # compute the resampled hessian H using finite differences; we could also compute the jacobian of the gradient d_loglik
        
            coefs[i,] = coefs[i-1,] - learn*solve(H,G) # update
      }
      colnames(coefs) = colnames(X) # label the coefficients
      return( list(coefs = coefs) ) # return draws
}
# estimates and standard errors (source: Introductory Econometrics, A Modern Approach 2nd Edition, Wooldridge)
coef = c(-0.012,0.131,0.123,-0.0019,-0.053,-0.868,0.036,0.270)
ses  = c( 0.005,0.025,0.019, 0.0006, 0.008, 0.119,0.043,0.509)

iter_rNR = 2e3         # number of rNR draws
learn    = 0.3         # learning rate
coef0    = coef*3.25   # starting value

m1 = 753 # m = n
m2 = 200 # m = 200
m3 = 100 # m = 100 

# adjustments to get valid standard errors
adj_rnr1 = sqrt(m1/n)*sqrt( (1-(1-learn)^2)/learn^2 )
adj_rnr2 = sqrt(m2/n)*sqrt( (1-(1-learn)^2)/learn^2 )
adj_rnr3 = sqrt(m3/n)*sqrt( (1-(1-learn)^2)/learn^2 )

b1 = 1 + round(log(0.01)/log(1-learn)) # burn-in sample size

# generate rNR draws
out_rNR1   = rNR(coef0,learn,b1 + iter_rNR, m1)
out_rNR2   = rNR(coef0,learn,b1 + iter_rNR, m2)
out_rNR3   = rNR(coef0,learn,b1 + iter_rNR, m3)

# format output
estimates = 
      rbind( coef,
        apply(out_rNR1$coef[b1:(iter_rNR+b1),],2,mean),
        apply(out_rNR2$coef[b1:(iter_rNR+b1),],2,mean),
        apply(out_rNR3$coef[b1:(iter_rNR+b1),],2,mean))

std_errs  = 
      rbind( ses,
        apply(out_rNR1$coef[b1:(iter_rNR+b1),],2,sd)*adj_rnr1,
        apply(out_rNR2$coef[b1:(iter_rNR+b1),],2,sd)*adj_rnr2,
        apply(out_rNR3$coef[b1:(iter_rNR+b1),],2,sd)*adj_rnr3)


estimates = as.data.frame(estimates)
colnames(estimates) = colnames(X)
rownames(estimates) = c('MLE','rNRn','rNR200','rNR100')

std_errs = as.data.frame(std_errs)
colnames(std_errs) = colnames(X)
rownames(std_errs) = c('ase','rNRn','rNR200','rNR100')

# print results
print(round( cbind( t(estimates), t(std_errs) ), digits = 3 ))

# output printed below:
#            MLE   rNRn rNR200 rNR100   ase  rNRn rNR200 rNR100
#nwifeinc -0.012 -0.012 -0.013 -0.014 0.005 0.005  0.005  0.005
#educ      0.131  0.132  0.136  0.140 0.025 0.026  0.026  0.028
#exper     0.123  0.123  0.123  0.125 0.019 0.019  0.020  0.021
#exper2   -0.002 -0.002 -0.002 -0.002 0.001 0.001  0.001  0.001
#age      -0.053 -0.053 -0.054 -0.055 0.008 0.008  0.009  0.009
#kidslt6  -0.868 -0.872 -0.895 -0.917 0.119 0.121  0.121  0.126
#kidsge6   0.036  0.038  0.040  0.038 0.043 0.045  0.047  0.049
#constant  0.270  0.272  0.282  0.276 0.509 0.506  0.505  0.535
\end{lstlisting}


\newpage
\section{Primitive Conditions for Assumption \ref{ass:A2}} \label{apx:primA2}
The following provides primitive conditions for Assumption \ref{ass:A2} i-ii. in the main text. Although the dimension of $\theta$ is fixed in the main results, the derivation make explicit the dependence of $C_2,C_3$ on the dimension $d_\theta$.
\paragraph{Setup and Notation.} Let $z_i = (y_i,x_i)$ be iid, the parameter space $\Theta$ is a compact, convex subset of $\mathbb{R}^{d_\theta}$, the sample objective, gradient and hessian are given by:
\[ Q_n(\theta) = \frac{1}{n} \sum_{i=1}^n q(z_i,\theta), \quad G_n(\theta) = \frac{1}{n} \sum_{i=1}^n \nabla q(z_i,\theta),  \quad H_n(\theta) = \frac{1}{n} \sum_{i=1}^n \partial^2_{\theta,\theta^\prime} q(z_i,\theta).\]
The re-sampled objective relies on $m$ draws $z_i^{(b)}$ taken with equal probability and with replacement from $(z_1,\dots,z_n)$, the re-sampled objective, gradient and hessian are given by:
\[ Q_m^{(b)}(\theta) = \frac{1}{m} \sum_{i=1}^m q(z_i^{(b)},\theta), \quad G_m^{(b)}(\theta) = \frac{1}{m} \sum_{i=1}^m \nabla q(z_i^{(b)},\theta),  \quad H_m^{(b)}(\theta) = \frac{1}{m} \sum_{i=1}^m \nabla^2 q(z_i^{(b)},\theta).\]
Alternatively, the re-weighted objective relies on $m=n$ iid random weights $w_i^{(b)}$ with mean and variance equal to one, the re-sampled objective, gradient and hessian are given by:
\[ Q_m^{(b)}(\theta) = \frac{1}{m} \sum_{i=1}^m w_i^{(b)} q(z_i,\theta), \quad G_m^{(b)}(\theta) = \frac{1}{m} \sum_{i=1}^m w_i^{(b)} \nabla q(z_i,\theta),  \quad H_m^{(b)}(\theta) = \frac{1}{m} \sum_{i=1}^m w_i^{(b)} \nabla^2 q(z_i,\theta),\]
for Gaussian multiplier weights we have $w_i^{(b)} \sim \mathcal{N}(1,1)$.
\begin{assumption}[Lipschitz Derivatives] \label{as:lip}
There exists a measurable function $C_q$ such that for any two $\theta_1,\theta_2 \in \Theta$ and any $z_i$:
\begin{align}
      \|\nabla q(z_i,\theta_1)-\nabla q(z_i,\theta_2)\| &\leq C_q(z_i)\|\theta_1-\theta_2\|,\\
      \|\nabla^2 q(z_i,\theta_1)-\nabla^2 q(z_i,\theta_2)\| &\leq C_q(z_i)\|\theta_1-\theta_2\|,
\end{align}
Let $p \geq 2$, $C_q$ and the derivatives have finite $p$-th moment:
\begin{align}
      &\mathbb{E}( |C_q(z_i)|^p ) <+\infty, \quad \mathbb{E}(\|\nabla q(z_i,\theta^\dagger)\|^p ) <+\infty, \quad \mathbb{E}(\|\nabla^2 q(z_i,\theta^\dagger)\|^p ) <+\infty \label{eq:momas}
\end{align}
\end{assumption}
The following subsections will prove the Lemma below under re-sampling and then under re-weighting with Gaussian multiplier weights.
\begin{lemma} \label{lem:as2rs}
      Suppose Assumption \ref{as:lip} holds, then we have:
      \vspace{-\topsep}
      \begin{itemize} \setlength\itemsep{0em}
            \item[i.] $\mathbb{E}^\star[G_m^{(b)}(\theta)] = G_n(\theta)$,
            \item[ii.]  $(\mathbb{E}^\star[ \sup_{\theta \in \Theta} \|G_m^{(b)}(\theta) - G_n(\theta)\|_2^p])^{1/p} \leq C_{2n} m^{-1/2}$, where $C_{2n} \overset{p}{\to} C_2$ finite,
            \item[ii.]  $(\mathbb{E}^\star[ \sup_{\theta \in \Theta} \|H_m^{(b)}(\theta) - H_n(\theta)\|_2^p])^{1/p} \leq C_{3n} m^{-1/2}$, where $C_{3n} \overset{p}{\to} C_3$ finite. 
      \end{itemize}
\end{lemma}

\subsection{Re-sampled objective}
\begin{proof}[Proof of Lemma \ref{lem:as2rs}] Condition i. is immediate from iid resampling and the additivity over $i$ of the sample and re-sampled gradients. Lemma \ref{lem:as2rs} ii. and iii. are proved the same way so the below will only focus on ii. for brievety.

Let $F_q(z_i,\theta) = C_q(z_i)\|\theta - \theta^\dagger\| + \|\partial_\theta q(z_i,\theta^\dagger)\|$ be the enveloppe function which satisfies: $\|\partial_\theta q(z_i,\theta)\| \leq F_q(z_i,\theta)$ for all $\theta \in \Theta$. We have $F_q(z_i,\theta) \leq F_q(z_i) = C_q(z_i)\text{diam}(\Theta) + \|\partial_\theta q(z_i,\theta^\dagger)\|$, which does not depend on $\theta$. Here $\text{diam}(\Theta) = \sup_{\theta_1,\theta_2 \in \Theta}\|\theta_1-\theta_2\|$, by compactness it is finite. The main idea is to apply Theorem 2.14.5 in \citet{VanderVaart1996}, under the distribution $\mathbb{P}^\star$, for a $p \geq 2$ such that Assumption \ref{as:lip} holds:
\begin{align*}
      &\left[ \mathbb{E}^\star \left( \sup_{\theta \in \Theta} \sqrt{m}\|G_m^{(b)}(\theta) - G_n(\theta)\|^p \right) \right]^{1/p}\\ 
      &\lesssim \mathbb{E}^\star \left( \sup_{\theta \in \Theta} \sqrt{m}\|G_m^{(b)}(\theta) - G_n(\theta)\| \right) + m^{1/p-1/2}\left[ \frac{1}{n} \sum_{i=1}^n |F_q(z_i)|^p \right]^{1/p},
\end{align*}
where $\lesssim$ stands for less or equal than, up to universal constants. By the strong law of large numbers and the finite moment Assumption \ref{as:lip} (\ref{eq:momas}), we have, as $n \to \infty$. $\left[ \frac{1}{n} \sum_{i=1}^n |F_q(z_i)|^p \right]^{1/p} \overset{p}{\to} \left[ \mathbb{E}(|F_q(z_i)|^p) \right]^{1/p}$. Also for $p\geq 2$ and $m \geq 1$, $0 \leq m^{1/p-1/2} \leq 1$. Now using Theorem 2.14.2 in \citet{VanderVaart1996}, again under the distribution $\mathbb{P}^\star$, we have:
\begin{align*}
      &\mathbb{E}^\star \left( \sup_{\theta \in \Theta} \sqrt{m}\|G_m^{(b)}(\theta) - G_n(\theta)\| \right) \lesssim J_{[\,]}(1,\mathcal{F},L_2(\mathbb{P}_n) )\left[ \frac{1}{n} \sum_{i=1}^n |F_q(z_i)|^2 \right]^{1/2},
\end{align*}
where $J_{[\,]}(1,\mathcal{F},L_2(\mathbb{P}_n) ) = \int_{0}^1 \sqrt{ 1+ \log N_{[\, ]}(\varepsilon\|F_q\|,\mathcal{F},L_2(\mathbb{P}_n)) }d\varepsilon$ is the bracketing integral of the functions class $\mathcal{F} = \{ \theta \to \partial_\theta q(z_i,\theta) \}$ and $N_{[\, ]}(\varepsilon\|F_q\|,\mathcal{F},L_2(\mathbb{P}_n))$ its bracketing number. By Theorem 9.23 in \citet{kosorok2007}, we have $N_{[\, ]}(\varepsilon\|F_q\|,\mathcal{F},L_2(\mathbb{P}_n)) \leq N(\varepsilon/2,\Theta,\|\cdot\|) \leq (6/\varepsilon)^{d_\theta} \text{vol}(\Theta)/\text{vol}(B)$ where $\text{vol}(\Theta),\text{vol}(B)$ are the volumes of $\Theta$ and the unit sphere $B$ in $\mathbb{R}^{d_\theta}$, respectively. As a result, we have:
\begin{align*}
      &\mathbb{E}^\star \left( \sup_{\theta \in \Theta} \sqrt{m}\|G_m^{(b)}(\theta) - G_n(\theta)\| \right) \lesssim \sqrt{d_{\theta}} \left[ \frac{1}{n} \sum_{i=1}^n |F_q(z_i)|^2 \right]^{1/2},
\end{align*}
and $\left[ \frac{1}{n} \sum_{i=1}^n |F_q(z_i)|^2 \right]^{1/2} \overset{p}{\to} \left[ \mathbb{E}( |F_q(z_i)|^2) \right]^{1/2}$ finite. Using the inequality: $\left[ \frac{1}{n} \sum_{i=1}^n |F_q(z_i)|^2 \right]^{1/2} \leq \left[ \frac{1}{n} \sum_{i=1}^n |F_q(z_i)|^p \right]^{1/p}$, we can conclude that:
\[ \left[ \mathbb{E}^\star \left( \sup_{\theta \in \Theta} \sqrt{m}\|G_m^{(b)}(\theta) - G_n(\theta)\|^p \right) \right]^{1/p} \lesssim (1+\sqrt{d_\theta})\left[ \frac{1}{n} \sum_{i=1}^n |F_q(z_i)|^p \right]^{1/p}, \]
where the last term converges in probability to $[\mathbb{E}(|F_q(z_i)|^p)]^{1/p}$, finite.
\end{proof} 
\subsection{Re-weighted objective with Gaussian multiplier weights.} The following specializes to the case where $w_i^{(b)} \sim \mathcal{N}(1,1)$, iid. Notice that:
\begin{align*} G_m^{(b)}(\theta) 
      &= G_n(\theta) + \frac{1}{m} \sum_{i=1}^m (w_i^{(b)}-1)\nabla q(z_i,\theta)  \Rightarrow G_m^{(b)}(\theta) \sim \mathcal{N}\left(G_n(\theta), \frac{1}{m} V_n(\theta) \right), \end{align*}
where $V_n(\theta) = 1/n \sum_{i=1}^n \nabla q(z_i,\theta)\nabla q(z_i,\theta)^\prime$ and $w_i^{(b)}-1 \sim \mathcal{N}(0,1)$. Note that the proof below only requires Assumption \ref{as:lip} to hold with $p=2$, even if when the desired result is stated for $p>2$.
\begin{proof}[Proof of Lemma \ref{lem:as2rs}]
Condition i. is immediate from the above. For any $\theta \in \Theta$,
\[ \sqrt{m} \left(G_m^{(b)}(\theta) - G_n(\theta) \right) \sim \mathcal{N}(0,V_n(\theta)),\]
is a vector-valued Gaussian process, conditional on the sample of data $z_1,\dots,z_n$. For each $i=1,\dots,n$ we have:
\[ \|(w_i^{(b)}-1)\nabla q(z_i,\theta)\| \leq |w_i^{(b)}-1|[C_q(z_i)\text{diam}(\Theta) + \|\nabla q(z_i,\theta^\dagger)\|] = |w_i^{(b)}-1| F_q(z_i), \] 
which defines the enveloppe function for the re-weighted objective. For any $a \in S^{d_\theta} = \{ a \in \mathbb{R}^{d_\theta}, \|a\|=1 \}$, the surface of the unit sphere, define $\sqrt{m} a^\prime (G_m^{(b)}(\theta)-G_n(\theta))$, a scalar-value Gaussian process defined on $S^{d_\theta} \times \Theta$ a compact subset of $\mathbb{R}^{d_\theta} \times \mathbb{R}^{d_\theta}$. Let $\sigma_n^2 = \sup_{(a,\theta)\in S^{d_\theta} \times \Theta} \mathbb{E}^\star(m|a^\prime(G_m^{(b)}(\theta)-G_n(\theta)|^2) = \sup_{(a,\theta)\in S^{d_\theta} \times \Theta} a^\prime V_n(\theta) a$ which is finite and converges in probability to $\sup_{\theta \in \Theta} [\lambda_{\max}(\mathbb{E}[\nabla q(z_i,\theta)\nabla q(z_i,\theta)^\prime])]$, also finite. We can now apply results for scalar valued Gaussian processes. In particular, using Theorem 5.8 in \citet{boucheron2013}, see also Proposition 3.19 in \citet{massart2007}, we have for any $u \geq 0$:
\begin{align*}
&\mathbb{P}^\star \left( \sup_{(a,\theta)\in S^{d_\theta} \times \Theta }\sqrt{m}|a^\prime(G_m^{(b)}(\theta)-G_n(\theta)| \geq M_n + u \right) \leq \exp \left( - \frac{u^2}{2\sigma_n^2}\right),
\end{align*}
where $M_n = \mathbb{E}^\star \left[ \sup_{(a,\theta)\in S^{d_\theta} \times \Theta }\sqrt{m}|a^\prime(G_m^{(b)}(\theta)-G_n(\theta)| \right]$.

Let $Z_m^{(b)} = \sup_{(a,\theta)\in S^{d_\theta} \times \Theta }\sqrt{m}|a^\prime(G_m^{(b)}(\theta)-G_n(\theta)|$. The main idea is to notice that $\sup_{(a,\theta)\in S^{d_\theta} }\sqrt{m}|a^\prime(G_m^{(b)}(\theta)-G_n(\theta)|$ is atained at $a = [G_m^{(b)}(\theta)-G_n(\theta)]/\|G_m^{(b)}(\theta)-G_n(\theta)\|$ so that $|Z_m^{(b)}|^p = \sup_{\theta \in \Theta} \|G_m^{(b)}(\theta)-G_n(\theta)\|^p$ which is the quantity we want to bound in expectations. Given that $Z_m^{(b)}\geq 0$ we can write its $p$-th moment as:
\begin{align*}
      \mathbb{E}^\star( |Z_m^{(b)}|^p ) &= p \int_{0}^\infty z^{p-1}\mathbb{P}^\star(Z_m^{(b)} > z)dz\\
      &= p \int_{0}^{M_n} z^{p-1}\mathbb{P}^\star(Z_m^{(b)} > z)dz + p \int_{0}^\infty (M_n+z)^{p-1}\mathbb{P}^\star(Z_m^{(b)} > M_n + z)dz\\
      &\leq M_n^p + p \int_{0}^\infty (M_n+z)^{p-1}\exp\left( -\frac{z^2}{2\sigma_n^2} \right)dz\\
      &= M_n^p + p\sigma_n \int_{0}^\infty (M_n+\sigma_n z)^{p-1}\exp\left( -\frac{z^2}{2} \right)dz.
\end{align*}
We already known that $\sigma_n$ is bounded and converges in probability to a finite limit. It remains to bound the moment $M_n$ before we can conclude.

Let $X_m^{(b)}(\theta) = \sqrt{m}(G_m^{(b)}(\theta)-G_n(\theta)$, by Gaussian re-weighting it is a Gaussian process. We will use an inequality for separable sub-Gaussian processes, specifically Corollary 2.2.8 in \citet{VanderVaart1996}, to bound $M_n$. The following verifies the assumptions required to apply the result.  For any $\theta_1,\theta_2 \in \Theta$, we have:
\begin{align*}
      &X_m^{(b)}(\theta_1)- X_m^{(b)}(\theta_2) = \frac{1}{\sqrt{m}}  \sum_{i=1}^m (w_i^{b}-1) [ \partial_\theta q(z_i,\theta_1)-\partial_\theta q(z_i,\theta_2) ] \sim \mathcal{N}\left(0,V_m(\theta_1,\theta_2)  \right),
\end{align*}
where $V_m(\theta_1,\theta_2) = \frac{1}{m}\sum_{i=1}^m [ \partial_\theta q(z_i,\theta_1)-\partial_\theta q(z_i,\theta_2) ][ \partial_\theta q(z_i,\theta_1)-\partial_\theta q(z_i,\theta_2) ]^\prime$. For any $\theta$, let $X_{m,j}^{(b)}(\theta)$ be the $j$-th row of $X_m^{(b)}(\theta)$. For any $j \in \{ 1,\dots,d_\theta\}$, we have $X_{m,j}^{(b)}(\theta_1)- X_{m,j}^{(b)}(\theta_2) \sim \mathcal{N}(0,V_{m,jj}(\theta_1,\theta_2))$ where $V_{m,jj}(\theta_1,\theta_2)$ is the $j$-th diagonal element of $V_m(\theta_1,\theta_2)$ which satisfies $V_{m,jj}(\theta_1,\theta_2) \leq \text{trace}[V_m(\theta_1,\theta_2)] = \frac{1}{m} \sum_{i=1}^m \|\partial_\theta q(z_i,\theta_1)-\partial_\theta q(z_i,\theta_2)\|^2 \leq \frac{1}{m} \sum_{i=1}^m F_q(z_i)^2\|\theta_1-\theta_2\|^2,$ by Lipschitz continuity. This implies the tail inequality:
\[ \mathbb{P}^\star\left( |X_{m,j}^{(b)}(\theta_1)- X_{m,j}^{(b)}(\theta_2)| > x \right) \leq 2\exp\left( \frac{-x^2}{2 \frac{1}{m} \sum_{i=1}^m F_q(z_i)^2\|\theta_1-\theta_2\|^2} \right), \]
which implies that each $X_{m,j}^{(b)}$ is a separable sub-Gaussian process under the semi-metric $d(\theta_1,\theta_2) = [\frac{1}{m} \sum_{i=1}^m F_q(z_i)^2]^{1/2}\|\theta_1-\theta_2\|$. Using Corollary 2.2.8 in \citet{VanderVaart1996}, we have for some universal constant $K$:
\[ \mathbb{E}^\star\left( \sup_{\theta \in \Theta} |X_{m,j}^{(b)}(\theta)| \right) \leq \mathbb{E}^\star\left( |X_{m,j}^{(b)}(\theta^\dagger)| \right) + K \int_0^\infty \sqrt{\log D(\varepsilon,d)}d\varepsilon, \]
where by Gaussianity of $X_{m,j}^{(b)}$, we have $\mathbb{E}^\star\left( |X_{m,j}^{(b)}(\theta^\dagger)| \right) = 2/\pi V_{n,jj}^{1/2}$, $V_{n,jj}$ is the $j$-th diagonal element of $V_n$ defined earlier and $D(\varepsilon,d) \leq 3[ \frac{1}{n}\sum_{i=1}^n F_q(z_i)^2]^{1/2} \text{diam}(\Theta)/\varepsilon$. Now notice that:
\begin{align*}
      M_n = \mathbb{E}^\star\left( \sup_{\theta \in \Theta} \|X_{m}^{(b)}(\theta)\| \right) &\leq \sum_{j=1}^{d_\theta} \mathbb{E}^\star\left( \sup_{\theta \in \Theta} |X_{m,j}^{(b)}(\theta)| \right)\\ &\leq d_\theta \left[2/\pi (\max_{j} V_{n,jj}^{1/2}) +  K \int_0^\infty \sqrt{\log D(\varepsilon,d)}d\varepsilon \right].
\end{align*}
Without loss of generality, assume that $M_n \geq 1$ then we have:
\begin{align*}
      \mathbb{E}^\star( |Z_m^{(b)}|^p ) &\leq M_n^p \left[ 1 + p\sigma_n \int_{0}^\infty (1+\sigma_n z)^{p-1}\exp\left( -\frac{z^2}{2} \right)dz \right].
\end{align*}
Putting everything together, we have the following inequality for any $p \geq 2$:
\begin{align*}
      &\left[\mathbb{E}^\star \left( \sup_{\theta} \|\sqrt{m}(G_m^{(b)}(\theta)-G_n(\theta))\|^p \right)\right]^{1/p}\\ &\leq d_\theta \left[2/\pi (\max_{j} V_{n,jj}^{1/2}) +  K \int_0^\infty \sqrt{\log D(\varepsilon,d)}d\varepsilon \right] \left[ 1 + p\sigma_n \int_{0}^\infty (1+\sigma_n z)^{p-1}\exp\left( -\frac{z^2}{2} \right)dz \right]^{1/p}.
\end{align*}
where $\int \sqrt{\log D(\varepsilon,d)}d\varepsilon$ and $\sigma_n$ converge in probability to a finite limit if Assumption \ref{as:lip} holds for $p =2$ since the quantities involved only depend on second moments.
\end{proof}

\newpage
\section{Primitive Conditions for Assumption \ref{ass:A4}} \label{apx:primA4}
The derivations below use the same setup and notation as Appendix \ref{apx:primA2} above.
\paragraph{Re-sampled objective.}
Suppose the gradient has finite third moment, i.e. $\mathbb{E}(\|\nabla q(z_i,\theta)\|^3)<\infty$. The calculations below are based on the derivations in \citet[Ch6.2, pp147-148]{lahiri2006}. Note that $V_n^{-1/2}[H_n(\hat\theta_n)]^{-1}G_m^{(b)}(\hat\theta_n) \overset{d^\star}{=} \Sigma_n^{-1/2}G_m^{(b)}(\hat\theta_n)$ so that we can focus on the latter. Also, note that $\mathbb{E}^\star[\Sigma_n^{-1/2}G_m^{(b)}(\hat\theta_n)]=0$, $\text{var}^\star[\Sigma_n^{-1/2}G_m^{(b)}(\hat\theta_n)]=I_d$. Now compute the characteristic function under $\mathbb{E}^\star$:
\[ \varphi^\star_m(\tau) = \mathbb{E}^\star \left( \exp[ \mathbf{i}\tau^\prime V_n^{-1/2} \sqrt{m}G_m^{(b)}(\hat\theta_n)] \right) = \left[ \mathbb{E}^\star \left( \exp\left[ \mathbf{i}\tau^\prime V_n^{-1/2} \frac{\nabla q(z_i^{(b)},\hat\theta_n)}{\sqrt{m}} \right] \right) \right]^m. \]
Using a Taylor expansion with remainder, we have for each $\tau$:
\[ \mathbb{E}^\star \left( \exp\left[ \mathbf{i}\tau^\prime V_n^{-1/2} \frac{\nabla q(z_i^{(b)},\hat\theta_n)}{\sqrt{m}} \right]  \right) = 1 - \frac{\|\tau\|_2^2}{2m} + R_{m}(\tau), \]
where $|R_{m}(\tau)| \leq \mathbb{E}^\star(\|\nabla q(z_i^{(b)},\hat\theta_n)\|_2^3) \|\tau\|_2^3 m^{-3/2}$ because all derivatives of the exponential term have modulus less than one. Apply a Taylor expansion with integral remainder to the logarithm:
\[ \log[\varphi_m^\star(\tau)] =  - \frac{\|\tau\|_2^2}{2} + mR_m(\tau) - m \int_{0}^{- \frac{\|\tau\|_2^2}{2m} + R_m(\tau)} \frac{(-\|\tau\|_2^2/[2m] + R_m(\tau)-t)}{(1+t)^2}dt, \]
where $mR_m(\tau) \leq \mathbb{E}^\star(\|\nabla q(z_i^{(b)},\hat\theta_n)\|_2^3) \|\tau\|_2^3 m^{-1/2}$; $\mathbb{E}^\star(\|\nabla q(z_i^{(b)},\hat\theta_n)\|_2^3) \overset{p}{\to} \mathbb{E}(\|\nabla q(z_i,\theta^\dagger)\|_2^3)<\infty$ using a uniform law of large numbers. For $m$ sufficiently large, $|- \frac{\|\tau\|_2^2}{2m} + R_m(\tau)| < 1/2$ with probability approaching $1$, so we can use the following bound:
\[ m\Big|\int_{0}^{- \frac{\|\tau\|_2^2}{2m} + R_m(\tau)} \frac{(-\|\tau\|_2^2/[2m] + R_m(\tau)-t)}{(1+t)^2}dt\Big| \leq 6m \left[\frac{\|\tau\|_2^4}{4m^2} + |R_m(\tau)|^2 \right],  \]
which satisfies the conditions required for Assumption \ref{ass:A4}.
\paragraph{Re-weighted objective with Gaussian multiplier weights.}
Assumption \ref{ass:A4} automatically holds because as noted in Appendix \ref{apx:primA2}, conditionally on the sample: $G_m^{(b)}(\hat\theta_n) \sim \mathcal{N}(0,\Sigma_n)$ which implies $V_n^{-1/2} [H_n(\hat\theta_n)]^{-1} G_m^{(b)}(\hat\theta_n) \sim \mathcal{N}(0,I_d)$ for which the characteristic function is the same as in Assumption \ref{ass:A4} with $r_m(\tau)=0$. 
\newpage
\section{Additional Results for Section \ref{sec:examples}} \label{apx:add_ex}
\subsection{Example 1: Dynamic Discrete Choice with Unobserved Heterogeneity} \label{apx:add_ex1}
\paragraph{Baseline results with homogeneity.}
Table \ref{tab:homDDC} provides baseline results for homogeneous dynamic discrete choice model where estimation is much faster so that a comparion with \rnr, and the standard bootstrap is feasible. The specification is taken from \citet{abbring-klein}. The first three columns correspond to MLE estimates and rejection rates using the standard $m$ out of $n$ bootstrap. The other 6 columns correspond to \rnr\, and \rqn.
\begin{table}[H] \caption{Baseline estimation and inference results with homogeneous agents} \label{tab:homDDC}
      \centering
      \begin{tabular}{l|ccc|ccc|ccc}
        \hline \hline
        & \multicolumn{3}{c|}{\textsc{mle}/bootstrap} & \multicolumn{3}{c|}{\rnr} & \multicolumn{3}{c}{\rqn}\\ \hline
        m & $\beta_0$ & $\beta_1$ & $\delta_1$ & $\beta_0$ & $\beta_1$ & $\delta_1$ & $\beta_0$ & $\beta_1$ & $\delta_1$ \\ \hline
        & \multicolumn{9}{c}{Average Estimate}\\ \hline
        1000 & -0.500 & 0.200 & 1.000 & -0.500 & 0.200 & 1.000 & -0.500 & 0.200 & 1.000 \\ 
        500 & - & - & - & -0.500 & 0.200 & 1.000 & -0.500 & 0.200 & 1.000 \\ 
        100 & - & - & - & -0.500 & 0.200 & 1.000 & -0.500 & 0.200 & 1.000 \\  \hline
        & \multicolumn{9}{c}{Standard Deviation}\\ \hline
        1000 & 0.014 & 0.004 & 0.013 & 0.014 & 0.004 & 0.013 & 0.014 & 0.004 & 0.013 \\ 
        500 & - & - & - & 0.014 & 0.004 & 0.014 & 0.014 & 0.004 & 0.014 \\ 
        100 & - & - & - & 0.014 & 0.004 & 0.013 & 0.014 & 0.004 & 0.013 \\ \hline
        & \multicolumn{9}{c}{Rejection Rates}\\ \hline
        1000 & 0.053 & 0.051 & 0.055 & 0.054 & 0.050 & 0.051 & 0.051 & 0.049 & 0.052 \\ 
        500 & 0.052 & 0.053 & 0.062 & 0.053 & 0.052 & 0.066 & 0.051 & 0.051 & 0.061 \\ 
        100 & 0.064 & 0.060 & 0.061 & 0.057 & 0.052 & 0.059 & 0.061 & 0.058 & 0.055 \\ 
        \hline \hline
      \end{tabular}\\
      {\footnotesize Legend: $n=1000$, $T=100$, $\gamma=0.1$, nominal size $=5\%$.}
\end{table}
\paragraph{Results with $T=50$, comparison with BFGS}
\newpage

\begin{table}[ht]  \caption{Dynamic Discrete Choice Model with Heterogeneity ($T=50$)} \label{tab:DDC_het_50}
      \centering
      \begin{tabular}{l|cccccccc}
        \hline \hline
       & $\mu_0^1$ & $\mu_1^1$ & $\mu_0^2$ & $\mu_1^2$ & $100 \sigma_1$ & $100 \sigma_2$ & $\omega$ & $\delta_1$  \\ 
       \hline
        $\theta^\dagger$ & -2.000 & 0.300 & -1.000 & 0.900 & 0.010 & 0.010 & 0.300 & 1.000 
        \\ \hline
        & \multicolumn{8}{c}{Average Estimates}
        \\\hline
        \rqn &  -2.032 & 0.302 & -0.983 & 0.906 & 1.132 & 0.907 & 0.294 & 0.954  \\ 
        \rqn-bc & -2.014 & 0.303 & -1.000 & 0.894 & 1.140 & 0.793 & 0.294 & 1.000\\
        \textsc{bfgs} & -1.823 & 0.222 & -0.950 & 0.945 & 4.029 & 3.620 & 0.329 & 0.860 \\
        \textsc{bfgs}-bc & -1.996 & 0.269 & -1.014 & 0.923 & 1.167 & 1.938 & 0.320 & 0.955 \\
        \hline
        & \multicolumn{8}{c}{Standard Deviation}
        \\\hline 
        \rqn & 0.044 & 0.012 & 0.025 & 0.015 & 0.165 & 0.365 & 0.013 & 0.023 \\ 
        \rqn-bc & 0.047 & 0.013 & 0.025 & 0.017 & 0.242 & 0.429 & 0.013 & 0.026 \\ 
        \textsc{bfgs} & 0.049 & 0.018 & 0.031 & 0.015 & 0.306 & 0.180 & 0.011 & 0.028 \\
        \textsc{bfgs}-bc & 0.056 & 0.021 & 0.032 & 0.016 & 0.368 & 0.258 & 0.013 & 0.031 \\
        \hline
        & \multicolumn{8}{c}{Rejection Rate}
        \\ \hline
        \rqn & 0.035 & 0.035 & 0.068 & 0.035 & 0.033 & 0.393 & 0.018 & 0.255\\
        \rqn-bc & 0.022 & 0.033 & 0.013 & 0.015 & 0.030 & 0.370 & 0.018 & 0.043 \\ 
        \rqn-bc$_{se}$ & 0.022 & 0.030 & 0.015 & 0.015 & 0.033 & 0.367 & 0.018 & 0.045 \\ 
        \hline \hline
     \end{tabular}\\
     {\footnotesize Legend: \rqn-bc = split panel bias-corrected \rqn; \rqn, \rqn-bc = quantile-based CIs, \rqn-bc$_{se}$  = std-error based CIs. $\gamma=0.1,n=1000,T=50,m=n/2, B=2000, \textsc{burn}=250,$ nominal level $=5\%.$ \textsc{bfgs}: R's optim with \textsc{bfgs} optimizer and tight convergence criteria, \textsc{bfgs}-bc: split panel biased corrected.}
\end{table}

\paragraph{Implementation details.} A grid of $60$ sobol points is used to integrate out the random coefficients and solve the fixed point problem. During the burn-in of 250 draws, a smaller grid of $30$ points is used. The split panel $\theta_{b,nT/2}^1,\theta_{b,nT/2}^2$ are only computed after $\frac{2}{3} \textsc{burn}$ draws and initialized at $\theta_{b,nT}$ to save time. The quasi-Newton matrix $P_b$ is pooled between split panel draws but computed separately for the full panel ones. The regularization schedule is $\lambda = 20$ for the first $\textsc{burn}/2$ draws and then an exponentially decaying schedule is used: $\lambda_{b+1} = 0.9 \lambda_{b}$.
\paragraph{Derivations for the second-order bias.} The following derives the second-order asymptotic bias of the random coefficient dynamic discrete choice model in Section \ref{sec:DDCh} under large-$T$ asymptotics and simplifying assumptions. In particular, the transition probability $\Pi$ is assumed to be known and a quadratic assumption is used to shorten integral derivations. Let $Q_{nT}(\theta)$ be the negative sample log-likelihood:
\[ Q_{nT}(\theta,\Pi) = -\frac{1}{nT} \sum_{i=1}^n \log[ \int \exp(T\ell_{iT}(\theta,\Pi,\beta))f(\beta|\theta)d\beta ], \]
where $\ell_{iT}(\theta,\Pi,\beta) = \sum_{t=1}^T \ell_{it}(\theta,\Pi,\beta)/T$ is the log-likelihood for individual $i$ using observations $t=1,\dots,T$ for a given set of parameters $(\theta,\Pi,\beta)$, $f(\beta|\theta)$ is the mixture distribution. The gradient $G_{nT}$ of the objective function $Q_{nT}$ is given by:
\[ G_{nT}(\theta,\Pi) = -\frac{1}{nT} \sum_{i=1}^n \frac{ \int (T\partial_\theta \ell_{iT}(\theta,\Pi,\beta) + \partial_\theta \log[f(\beta|\theta)])\exp(T\ell_{iT}(\theta,\Pi,\beta) + \log[f(\beta|\theta)])d\beta }{ \int \exp(T\ell_{iT}(\theta,\Pi,\beta) + \log[f(\beta|\theta)])d\beta }.\]
Under regularity conditions, the estimator $\hat\theta_{nT}$ is $\sqrt{n}$-consistent and asymptotically normal, with $T$ fixed with second order representation:\footnote{ The estimator is $\sqrt{n}$ not $\sqrt{nT}$ asympotically normal because of the dependence over $t$ for each $i$ due to $\beta_i$. See \citet{rsu-96} for general results on second-order bias in nonlinear models.}
\[ \sqrt{n}\left( \hat\theta_{nT} - \theta^\dagger \right) = -[H_{nT}(\theta^\dagger,\Pi^\dagger)]^{-1} G_{nT}(\theta^\dagger,\Pi^\dagger) + \frac{1}{\sqrt{n}}B_{nT}(\theta^\dagger,\Pi^\dagger) + o_p(n^{-1/2}),\]
and $B_{nT}(\theta^\dagger,\Pi^\dagger)$ converges in probability to some limit which depends on $T$. The main idea for the following is to capture some of the terms in $B_{nT}(\theta^\dagger,\Pi^\dagger)$ using a large-T asymptotic framework where $T$ and $n$ grow at the same rate. The main idea here is to approximate $\sqrt{n}G_{nT}(\theta^\dagger,\Pi^\dagger)$ as a asymptotically normal term which does not involve the log/integral/exponential transformation plus a $1/\sqrt{T} \asymp 1/\sqrt{n}$ term which accounts for the transformation. The formula is very similar to the incidental parameter bias but is of second rather than first order.
To further simplify the derivations below, the following assumption will be used.
\begin{assumption} \label{as:quad}
      The functions $\ell_{iT}(\theta,\Pi,\beta)$, $\partial_\theta\ell_{iT}(\theta,\Pi,\beta)$, $\log[f(\beta|\theta)]$ are quadratic in $\beta$ for each $\theta,\Pi$. 
\end{assumption}
Let $\beta_i$ be the true value of $\beta$ for each individual $i$, using the change of variable $\beta = \beta_i + h T^{-1/2}$, re-write:
\begin{align}
      &G_{nT}(\theta,\Pi) = \nonumber\\ 
      &-\frac{1}{n} \sum_{i=1}^n \partial_\theta \ell_{iT}(\theta,\Pi,\beta_i) \label{eq:expr1}\\
      &-\frac{1}{n\sqrt{T}} \sum_{i=1}^n \partial^2_{\theta,\beta} \ell_{iT}(\theta,\Pi,\beta_i) \frac{ \int h \exp(T\ell_{iT}(\theta,\Pi,\beta_i + hT^{-1/2}) + \log[f(\beta_i + hT^{-1/2})])dh}{\int \exp(T\ell_{iT}(\theta,\Pi,\beta_i + hT^{-1/2}) + \log[f(\beta_i + hT^{-1/2})])dh} \label{eq:expr2}\\
      &-\frac{1}{2n T} \sum_{i=1}^n \sum_{j=1}^{d_\theta}\partial^3_{\theta,\beta,\beta_j} \ell_{iT}(\theta,\Pi,\beta_i) \frac{ \int  h h_j \exp(T\ell_{iT}(\theta,\Pi,\beta_i + hT^{-1/2}) + \log[f(\beta_i + hT^{-1/2})])dh}{\int \exp(T\ell_{iT}(\theta,\Pi,\beta_i + hT^{-1/2}) + \log[f(\beta_i + hT^{-1/2})])dh}. \label{eq:expr3}
\end{align}
The first term, (\ref{eq:expr1}), is $\sqrt{n}$-asymptotically normal; it corresponds to a likelihood where all the heterogeneity is fully observed. The second term,  (\ref{eq:expr2}), is the linear expansion term of the transformation and is also $\sqrt{n}$-asymptotically normal; i.e. it contributes to the asymptotic variance by taking into account the estimation of the unobserved heterogeneous distribution. The third term,  (\ref{eq:expr2}), is the quadratic expansion term of the transformation which accounts for the non-linear effect of the transformation. It will contribute to the second-order asymptotic bias. 
Notice that, with the simplifying assumptions, the terms inside the exponential function are quadratic so that the corresponding distribution is Gaussian. To derive the mean and variance-covariance matrix of the Gaussian let $Q_{iT}(\theta,\Pi,\beta_i + h T^{-1/2})=-\ell_{iT}(\theta,\Pi,\beta_i + hT^{-1/2}) - \frac{1}{T}\log[f(\beta_i + hT^{-1/2})]$ and $G_{iT}(\theta,\Pi,\beta_i),H_{iT}(\theta,\Pi,\beta_i)$ be its gradient and hessian. Given the quadratic assumption, we have:
\begin{align*}
      &\exp(-TQ_{iT}(\theta,\Pi,\beta_i + hT^{-1/2})) =\\ &\exp \left( -TQ_{iT}(\theta,\Pi,\beta_i) +\frac{T}{2} G_{iT}(\theta,\Pi,\beta_i)^\prime H_{iT}(\theta,\Pi,\beta_i)G_{iT}(\theta,\Pi,\beta_i) \right)\\
      &\times \exp \left( - \frac{1}{2}\left[ h + \sqrt{T}H_{iT}(\theta,\Pi,\beta_i)^{-1}G_{iT}(\theta,\Pi,\beta_i) \right]^\prime H_{iT}(\theta,\Pi,\beta_i) \left[ h + \sqrt{T}H_{iT}(\theta,\Pi,\beta_i)^{-1}G_{iT}(\theta,\Pi,\beta_i) \right]   \right),
\end{align*}
which is, up to normalizing constants, equal to the density of a multivariate Gaussian with mean $-\sqrt{T}H_{iT}(\theta,\Pi,\beta_i)^{-1}G_{iT}(\theta,\Pi,\beta_i)$ and variance $H_{iT}(\theta,\Pi,\beta_i)^{-1}$. Notice that this is the asymptotic distribution of $\sqrt{T}(\hat\beta_i-\beta_i)$, the individual fixed effect estimator of $\beta$ when the information matrix equality holds. This implies that (\ref{eq:expr2}) simplifies to:
\[ (\ref{eq:expr2}) = \frac{1}{n} \sum_{i=1}^n \partial^2_{\theta,\beta} \ell_{iT}(\theta,\Pi,\beta_i) H_{iT}(\theta,\Pi,\beta_i)^{-1} G_{iT}(\theta,\Pi,\beta_i). \]
Likewise, (\ref{eq:expr3}) simplifies to:
\[ (\ref{eq:expr3}) = -\frac{1}{2nT} \sum_{i=1}^n \sum_{j=1}^{d_\theta}\partial^3_{\theta,\beta,\beta_j} \ell_{iT}(\theta,\Pi,\beta_i) [H_{iT}(\theta,Pi,\beta_i)^{-1}]_j, \]
where $[H_{iT}(\theta,Pi,\beta_i)^{-1}]_j$ is the $j$-th column of $H_{iT}(\theta,Pi,\beta_i)^{-1}$.
\begin{assumption} \label{as:cltlln}
      Suppose that as both $n,T \to \infty$:
      \[ \frac{1}{\sqrt{n}} \sum_{i=1}^n \sum_{t=1}^T  [\partial_\theta \ell_{it}(\theta^\dagger,\Pi^\dagger) + \partial^2_{\theta,\beta} \ell_{iT}(\theta^\dagger,\Pi^\dagger,\beta_i)H_{iT}(\theta^\dagger,\Pi^\dagger,\beta_i)^{-1}\sqrt{T}H_{iT}(\theta^\dagger,\Pi^\dagger,\beta_i)]   \overset{d}{\to} \mathcal{N}(0,\Sigma),  \]
      and that:
      \[ \frac{1}{n} \sum_{i=1}^n \sum_{j=1}^{d_\theta}\partial^3_{\theta^\dagger,\beta,\beta_j} \ell_{iT}(\theta^\dagger,\Pi^\dagger,\beta_i) [H_{iT}(\theta^\dagger,Pi^\dagger,\beta_i)^{-1}]_j \overset{p}{\to} 2\mathcal{B}(\theta^\dagger,\Pi^\dagger), \]
      where $\mathcal{B}(\theta^\dagger,\Pi^\dagger)=\mathbb{E}\left( \sum_{j=1}^{d_\theta}\partial^3_{\theta^\dagger,\beta,\beta_j} \ell_{iT}(\theta^\dagger,\Pi^\dagger,\beta_i) [H_{iT}(\theta^\dagger,Pi^\dagger,\beta_i)^{-1}]_j \right)/2.$ 
\end{assumption}

\begin{proposition}[Asymptotic Bias]
Suppose Assumptions \ref{as:quad}-\ref{as:cltlln} hold and $n/T \to \kappa^2 \in (0,+\infty)$, then:
\begin{align} \sqrt{n} G_{nT}(\theta^\dagger,\Pi^\dagger) \overset{d}{\to} \mathcal{N}(0, \Sigma ), \label{eq:pr0}\\
      \sqrt{nT} \left(G_{nT}(\theta^\dagger,\Pi^\dagger) - (\ref{eq:expr1}) - (\ref{eq:expr2}) \right) \overset{p}{\to} -\kappa \mathcal{B}(\theta^\dagger,\Pi^\dagger). \label{eq:pr1} \end{align}
Suppose $\hat\theta_{nT}$ satisfies the asymptotic expansion:
\begin{align} \sqrt{n}( \hat\theta_{nT} - \theta^\dagger ) = -  \sqrt{n}[H_{\infty}(\theta^\dagger,\Pi^\dagger)]^{-1}G_{nT}(\theta^\dagger,\Pi^\dagger) + \frac{1}{\sqrt{n}}B_{nT}(\theta^\dagger,\Pi^\dagger) + o_p(n^{-1/2}), \label{eq:pr2} \end{align}
where $H_{nT}(\theta^\dagger,\Pi^\dagger) \overset{p}{\to} H_{\infty}(\theta^\dagger,\Pi^\dagger)$ positive definite. The estimator is asymptotically normal:
\begin{align}
      \sqrt{n}\left( \hat\theta_{nT} - \theta^\dagger  \right) \overset{d}{\to} \mathcal{N}(0, V ), \label{eq:pr3}
\end{align}
where $V = [H_{\infty}(\theta^\dagger,\Pi^\dagger)]^{-1} \Sigma [H_{\infty}(\theta^\dagger,\Pi^\dagger)]^{-1}$. The second-order bias can be decomposed in two:
\begin{align*}
      \sqrt{nT}\left( \hat\theta_{nT} - \theta^\dagger - [H_{nT}(\theta^\dagger,\Pi^\dagger)]^{-1}[(\ref{eq:expr1}) + (\ref{eq:expr2})] \right) \overset{p}{\to} &\kappa [H_{\infty}(\theta^\dagger,\Pi^\dagger)]^{-1}\mathcal{B}(\theta^\dagger,\Pi^\dagger)\\ &+ \kappa^{-1}\text{plim}_{n,T \to \infty}B_{nT}(\theta^\dagger,\Pi^\dagger)
\end{align*}
\end{proposition}
The asymptotic expansion condition (\ref{eq:pr2}) can be verified using conditions in \citet{rsu-96}, the convergence in probability of the Hessian can be derived from a law of large numbers. The positive definiteness is a local identification condition.
\begin{proof}
Note that $\sqrt{n}G_{nT}(\theta^\dagger,\Pi^\dagger) = \sqrt{n}[ (\ref{eq:expr1}) + (\ref{eq:expr2}) ] + \sqrt{n}(\ref{eq:expr3})$. Assumption \ref{as:cltlln} implies that the first term converges in distribution to a Gaussian distribution with zero mean and asymptotic variance-covariance matrix $\Sigma$. Using the same Assumption, the second term converges in probability to zero. Using the same identity, $\sqrt{nT} [G_{nT}(\theta^\dagger,\Pi^\dagger)-(\ref{eq:expr1}) - (\ref{eq:expr2}) ] = \sqrt{nT}(\ref{eq:expr3}) \overset{p}{\to} -\kappa \mathcal{B}(\theta^\dagger,\Pi^\dagger)$ which is the first result. Then (\ref{eq:pr1}) comes from Assumption \ref{as:cltlln}.
Pre-multiplying (\ref{eq:pr1}) by $H_{nT}(\theta^\dagger,\Pi^\dagger)^{-1} \overset{p}{\to} H_{\infty}(\theta^\dagger,\Pi^\dagger)^{-1}$ and using the asymptotic expansion (\ref{eq:pr2}) then implies the last result.
\end{proof}
\begin{proposition}[Spit Panel Bias Reduction]
Let $\hat\theta_{nT/2}^1,\hat\theta_{nT/2}^2$ be the estimators computed using the first $\lfloor T/2 \rfloor$ and the last $T-\lfloor T/2 \rfloor$ time-observations, both using all individuals $i=1,\dots,n$.
Suppose the associated objectives $Q_{nT/2}^1,Q_{nT/2}^2$ satisfy Assumptions (\ref{as:quad})-(\ref{as:cltlln}) using $\lfloor T/2 \rfloor$ and $T-\lfloor T/2 \rfloor$ instead of $T$.
Let $\hat\theta^{1,2}_{nT/2}$ be either estimator, suppose they both satisfy the asymptotic expansions:
\begin{align} \sqrt{n}( \hat\theta^{1,2}_{nT/2} - \theta^\dagger ) = -  \sqrt{nT}[H^{1,2}_{\infty}(\theta^\dagger,\Pi^\dagger)]^{-1}G^{1,2}_{nT/2}(\theta^\dagger,\Pi^\dagger) + \frac{1}{\sqrt{n}} B^{1,2}_{nT/2}(\theta^\dagger,\Pi^\dagger) + o_p(n^{-1/2}), \label{eq:pr22} \end{align}
where $H^{1,2}_{nT/2}(\theta^\dagger,\Pi^\dagger) \overset{p}{\to} H_{\infty}(\theta^\dagger,\Pi^\dagger)$ positive definite. Then the bias corrected estimator $\tilde \theta_{nT} = 2\hat \theta_{nT}-(\hat \theta^1_{nT/2}+\hat \theta^2_{nT/2})/2$ satisfies:
\begin{align} \sqrt{n} \left( \tilde\theta_{nT} - \theta^\dagger  \right) \overset{d}{\to} \mathcal{N}\left( 0, V \right), \label{eq:pr23_0}\\
 \sqrt{nT} \left( \tilde\theta_{nT} - \theta^\dagger - [H_{\infty}(\theta^\dagger,\Pi^\dagger)]^{-1}[ (\ref{eq:expr1}) + (\ref{eq:expr2}) ] \right) \overset{p}{\to} \kappa^{-1}\text{plim}_{n,T \to \infty} B_{nT}(\theta^\dagger,\Pi^\dagger). \label{eq:pr23} \end{align}
where $V$ is the same as the asymptotic variance of $\hat\theta_{nT}$. 
\end{proposition}
The main takeaway of the Proposition is (\ref{eq:pr23}) which indicates that the part of the second-order bias which is associated with the integral transformed can be removed using the split panel jackknife.

\begin{proof}
Combining asymptotic expansions (\ref{eq:pr2}), (\ref{eq:pr22}), we have:
\begin{align*}
      \sqrt{n} \left( \tilde\theta_{nT} - \theta^\dagger \right) &= - \sqrt{n}H_{\infty}(\theta^\dagger,\Pi^\dagger)^{-1}\left( 2G_{nT}(\theta^\dagger,\Pi^\dagger) - [G^1_{nT/2}(\theta^\dagger,\Pi^\dagger)+G^2_{nT/2}(\theta^\dagger,\Pi^\dagger)]/2 \right)\\
      &+ \frac{1}{\sqrt{n}} \left( 2B_{nT} - [B_{nT/2}^1+B_{nT/2}^2]/2 \right).
\end{align*}
Notice that:
\begin{align*}
      2G_{nT}(\theta^\dagger,\Pi^\dagger) - [G^1_{nT/2}(\theta^\dagger,\Pi^\dagger)+G^2_{nT/2}(\theta^\dagger,\Pi^\dagger)]/2 = (\ref{eq:expr1}) + (\ref{eq:expr2}) + (\ref{eq:extra}),
\end{align*}
where (\ref{eq:extra}) is defined below. The identity comes from additivity over $t$ of $(\ref{eq:expr1}),(\ref{eq:expr2})$. The (\ref{eq:extra}) term comes from combining the (\ref{eq:expr3}) terms in each one of $G_{nT}$, $G^1_{nT/2}$, and $G^2_{nT/2}$:
\begingroup\leqnos
\begin{flalign} 
      \label{eq:extra}
      =-2\frac{1}{2nT} &\sum_{i=1}^n \sum_{j=1}^{d_\theta}\partial^3_{\theta,\beta,\beta_j} \ell_{iT}(\theta,\Pi,\beta_i) [H_{iT}(\theta,Pi,\beta_i)^{-1}]_j\\
      &-\frac{1}{2}\frac{1}{2nT/2} \sum_{i=1}^n \sum_{j=1}^{d_\theta}\partial^3_{\theta,\beta,\beta_j} \ell^1_{iT/2}(\theta,\Pi,\beta_i) [H^2_{iT/2}(\theta,Pi,\beta_i)^{-1}]_j \notag\\ 
      &-\frac{1}{2}\frac{1}{2nT/2} \sum_{i=1}^n \sum_{j=1}^{d_\theta}\partial^3_{\theta,\beta,\beta_j} \ell^2_{iT/2}(\theta,\Pi,\beta_i) [H^2_{iT/2}(\theta,Pi,\beta_i)^{-1}]_j. \notag
\end{flalign}
\endgroup
Using the assumptions and $\sqrt{n/[T/2]} \to \sqrt{2}\kappa$, $\sqrt{nT}(\ref{eq:extra})$ has the following limit:
\begin{align*} 
      \sqrt{nT}(\ref{eq:extra}) &=
      -2\sqrt{n/T}\frac{1}{2n} \sum_{i=1}^n \sum_{j=1}^{d_\theta}\partial^3_{\theta,\beta,\beta_j} \ell_{iT}(\theta,\Pi,\beta_i) [H_{iT}(\theta,Pi,\beta_i)^{-1}]_j \\
      &-\frac{1}{2}\sqrt{2}\sqrt{n/[T/2]}\frac{1}{2nT/2} \sum_{i=1}^n \sum_{j=1}^{d_\theta}\partial^3_{\theta,\beta,\beta_j} \ell^1_{iT/2}(\theta,\Pi,\beta_i) [H^2_{iT/2}(\theta,Pi,\beta_i)^{-1}]_j \\ 
      &-\frac{1}{2}\sqrt{2}\sqrt{n/[T/2]}\frac{1}{2nT/2} \sum_{i=1}^n \sum_{j=1}^{d_\theta}\partial^3_{\theta,\beta,\beta_j} \ell^2_{iT/2}(\theta,\Pi,\beta_i) [H^2_{iT/2}(\theta,Pi,\beta_i)^{-1}]_j. \\
      &\overset{p}{\to} 2[-2\kappa\mathcal{B}(\theta^\dagger,\Pi^\dagger)] - \sqrt{2}[  -2 \sqrt{2}\kappa \mathcal{B}(\theta^\dagger,\Pi^\dagger) ] = 0.
\end{align*}
Putting all the derivations together, we get the two sets of results:
\begin{align*}
      \sqrt{n}\left( 2G_{nT}(\theta^\dagger,\Pi^\dagger) - [G^1_{nT/2}(\theta^\dagger,\Pi^\dagger)+G^2_{nT/2}(\theta^\dagger,\Pi^\dagger)]/2 \right) = \sqrt{n}[H_{\infty}(\theta^\dagger,\Pi^\dagger)]^{-1}[(\ref{eq:expr1}) + (\ref{eq:expr2})] + o_p(1),
\end{align*}
from which the asymptotic normality result follows. Also,
\begin{align*}
      &\sqrt{nT}\left( 2G_{nT}(\theta^\dagger,\Pi^\dagger) - [G^1_{nT/2}(\theta^\dagger,\Pi^\dagger)+G^2_{nT/2}(\theta^\dagger,\Pi^\dagger)]/2 - [H_{\infty}(\theta^\dagger,\Pi^\dagger)]^{-1}[(\ref{eq:expr1}) + (\ref{eq:expr2})] \right)\\ &= o_p(1) + \kappa^{-1} \text{plim}_{n,T \to \infty} B_{nT}(\theta^\dagger,\Pi^\dagger),
\end{align*}
where the $o_p(1)$ term includes the bias reduction in (\ref{eq:extra}).
\end{proof}
\paragraph{A split panel implementation of \rqn.} The split panel bias corrected estimator is $\tilde \theta_{nT} = 2\hat\theta_{nT} - [\hat\theta^1_{nT/2}+\hat\theta^2_{nT/2}]/2$. The goal is to target $\tilde \theta_{nT}$ using \rqn. Let $\theta_{b,nT},\theta^1_{b,nT/2},\theta^2_{b,nT/2}$ be \rqn\, draws computed using the full and split panels, respectively but the same resampled individuals $i$. Using Proposition \ref{prop:coupling}, we have the three couplings:
\begin{align*}
      \theta^\star_{b+1,nT} &= \hat\theta_{nT} + (1-\gamma)( \theta^\star_{b,nT} - \hat\theta_{nT} ) - \gamma [H_{nT}(\theta^\dagger,\Pi^\dagger)]^{-1} G_{nT}(\theta^\dagger,\Pi^\dagger)\\
      \theta^{\star,1}_{b+1,nT/2} &= \hat\theta^1_{nT/2} + (1-\gamma)( \theta^{\star,1}_{b,nT/2} - \hat\theta^1_{nT/2} ) - \gamma [H^1_{nT/2}(\theta^\dagger,\Pi^\dagger)]^{-1} G^1_{nT/2}(\theta^\dagger,\Pi^\dagger)\\
      \theta^{\star,2}_{b+1,nT/2} &= \hat\theta^2_{nT/2} + (1-\gamma)( \theta^{\star,2}_{b,nT/2} - \hat\theta^2_{nT/2} ) - \gamma [H^2_{nT/2}(\theta^\dagger,\Pi^\dagger)]^{-1} G^2_{nT/2}(\theta^\dagger,\Pi^\dagger).
\end{align*}
Compute $\tilde \theta_{b+1,nT} = 2\theta_{b+1,nT}-[\theta^{1}_{b+1,nT/2}+\theta^{2}_{b+1,nT/2}]/2$, we have the associated coupling:
\begin{align*}
      \tilde \theta^\star_{b+1,nT} &= \tilde\theta_{nT} + (1-\gamma)( \tilde \theta^\star_{b,nT} - \tilde\theta_{nT}) - \gamma \Big( 2[H_{nT}(\theta^\dagger,\Pi^\dagger)]^{-1} G_{nT}(\theta^\dagger,\Pi^\dagger) \\ &- \Big[[H^1_{nT/2}(\theta^\dagger,\Pi^\dagger)]^{-1} G^1_{nT/2}(\theta^\dagger,\Pi^\dagger)+[H^2_{nT/2}(\theta^\dagger,\Pi^\dagger)]^{-1} G^2_{nT/2}(\theta^\dagger,\Pi^\dagger)\Big]/2 \Big),
\end{align*}
which has innovations that match the first-order expansion for $\tilde\theta_{nT}$.
\newpage
\subsection{Example 2: A Probit with Many Regressors} \label{apx:add_ex2}
Both \rqn\, and MCMC are initialized at the same starting value $\theta_0$ which is computed as follows. A linear probability model is estimated by OLS yielding coefficients $\tilde \theta_{n,\textsc{ols}}$. The starting value $\theta_0$ is then computed using the linear approximation: $x_i^\prime\tilde \theta_{n,\textsc{ols}} \simeq  \Phi(x_i^\prime \theta_0 ) \simeq \Phi(0) + \phi(0)x_i^\prime \theta_0$, where $\phi,\Phi$ are the Gaussian pdf and cdf, $\Phi(0)=1/2$, $\phi(0)=[2\pi]^{-1/2}$. For the intercept, this yields $\theta_{0,\text{const}} = [\tilde \theta_{n,\text{const},\textsc{ols}}-\Phi(0)]/\phi(0)$. For all other coefficients this yields $\theta_{0,j} = \tilde \theta_{n,j,\textsc{ols}}/\phi(0)$. 
\paragraph{Implementation of the MALA algorithm.}
For Bayesian inference, let $Q_n$ be the sample negative log-likelihood and $\pi$ the prior distribution. Let $G_n$ be the gradient of the negative log-posterior distribution and $H_n$ its Hessian.\footnote{The negative log-likelihood is minus one times the log-likelihood. The log-posterior distribution is the log-likelihood plus the log-prior distribution.} Given a draw $\theta_b$, the baseline MALA algorithm produces the next draw $\theta_{b+1}$ using the following steps:
\begin{align*}
      &\text{Draw }\tilde \theta = \theta_b - \gamma G_n(\theta_b) + \sqrt{2\gamma} Z_b,\, Z_b \sim \mathcal{N}(0,I), \text{ and } u\sim\mathcal{U}_{[0,1]},\\
      &\text{If } u \leq \frac{\exp[-Q_n(\tilde\theta)]\pi(\tilde\theta)}{\exp[-Q_n(\theta_b)]\pi(\theta_b)} \frac{q(\tilde \theta | \theta_b)}{q(\theta_b|\tilde \theta)}, \text{ set } \theta_{b+1} = \tilde \theta. \text{ Otherwise, set } \theta_{b+1} = \theta_{b},
\end{align*}
where $q(\tilde \theta | \theta_b) = \phi( [\tilde \theta - \theta_b + \gamma G_n(\theta_b)]/\sqrt{2\gamma} )$ and $q(\theta_b|\tilde \theta) = \phi( [ \theta_b - \tilde \theta + \gamma G_n(\tilde \theta)]/\sqrt{2\gamma} )$ are the transition probabilities, $\phi$ is the  pdf of the multivariate Gaussian. Compared to the random-walk Metropolis-Hastings algorithm which uses the proposal $\tilde \theta = \theta_b + \sqrt{2\gamma} Z_b$, MALA adds the gradient descent term, $ - \gamma G_n(\theta_b)$, which directs the proposal towards the maximum of the posterior distribution to get faster convergence and a higher acceptance rate. Here, the proposal is modified for more direct comparison with \rnr, \rqn. It requires computing $H_n(\hat\theta_n)$:
\[ \tilde \theta = \theta_b - \gamma [H_n(\hat\theta_n)]^{-1}G_n(\theta_b) + \sqrt{2\gamma} Z_b, \quad Z_b \sim \mathcal{N}(0,[H_n(\hat\theta_n)]^{-1}), \]
the transition probability $q(\cdot|\cdot)$ is adjusted accordingly.  Following \citet{roberts1998}, $\gamma$ is tuned to get an acceptance rate of 0.57, after the initial convergence phase. This requires $\gamma = O(1/d_\theta)$ as illustrated below. The desired acceptance rate is achieved for $\gamma \simeq 0.0034$ to be compared with $\gamma=0.1$ used for \rnr\, and \rqn.
\paragraph{Choice of learning rate $\gamma$ for the MALA algorithm.} Recall the modified proposal:
\[ \tilde \theta = \theta_b - \gamma [H_n(\hat\theta_n)]^{-1}G_n(\theta_b) + \sqrt{2\gamma} Z_b, \quad Z_b \sim \mathcal{N}(0,[H_n(\hat\theta_n)]^{-1}). \]
To illustrate how the dimension of the parameters $d_\theta$ affects the choice of $\gamma$, consider the case wher $\theta_b = \hat\theta_n$ so that $G_n(\theta_b)=0$. This implies that conditional on $\theta_b=\hat\theta_n$, $\tilde \theta = \hat\theta_n + \sqrt{2\gamma} Z_b \sim \mathcal{N}(\hat\theta_n,2\gamma[H_n(\hat\theta_n)]^{-1})$. The log-posterior is approximately:
\[ Q_n(\tilde\theta) - \log(\pi(\tilde\theta)) \simeq Q_n(\hat\theta_n) - \log(\pi(\hat\theta_n)) + \gamma  Z_b^\prime H_n(\hat\theta_n)Z_b, \]
where $Z_b^\prime H_n(\hat\theta_n)Z_b \sim \chi^2_p$. For $d_\theta$ large, we have $Z_b^\prime H_n(\hat\theta_n)Z_b/d_{\theta} = 1 + o_p(1).$ Plug this into the accept/reject step of the algorithm:
\[ \mathbb{P}\left( u \leq \frac{\exp[-Q_n(\tilde\theta)]\pi(\tilde\theta)}{\exp[-Q_n(\hat\theta_n)]\pi(\hat\theta_n)}\right) \simeq \mathbb{P}\left( u \leq \exp[ -\gamma d_\theta ]\right) = \exp[ -\gamma d_\theta ], \]
which equals $0.57$ for $\gamma = -\log(0.57)/d_{\theta}$. At $\theta_b = \hat\theta_n$, the optimal choice of $\gamma$ is inversely proportional to $d_\theta$. For $\theta_b \neq \hat\theta_n$, the acceptance rate is higher because the gradient term sets $\tilde \theta$ in directions of increasing posterior, on average. This implies that a larger value $\gamma > -\log(0.57)/d_{\theta}$ should be used to get the desired acceptance rate overall. In the empirical application, the learning rate was set to $\gamma = -2\log(0.57)/d_{\theta}$ which resulted in the desired acceptance rate.

\paragraph{Additional Empirical Results and Comparisons.}
\,\,
\begin{table}[H] \caption{Probit - Estimates and Standard Errors} \label{tab:BigProbit}
      \centering
      { \renewcommand{\arraystretch}{0.935} \setlength{\tabcolsep}{1.8pt}
      \begin{tabular}{l|ccccccccccc}
        \hline \hline
       & const & distance & border & island & landlock & legal & language & colonial & currency & FTA & religion \\ 
        \hline
      & \multicolumn{10}{c}{Estimates}\\
      \hline
      \textsc{mle} & 2.537 & -0.618 & -0.380 & 0.355 & 0.220 & 0.072 & 0.275 & 0.288 & 0.530 & 1.854 & 0.249 \\ 
      \rnr & 2.541 & -0.619 & -0.381 & 0.355 & 0.219 & 0.071 & 0.275 & 0.287 & 0.529 & 1.863 & 0.249 \\ 
      \rqn & 2.539 & -0.618 & -0.380 & 0.355 & 0.219 & 0.071 & 0.275 & 0.286 & 0.529 & 1.860 & 0.249 \\ 
      \textsc{mcmc} & 2.533 & -0.617 & -0.380 & 0.355 & 0.215 & 0.072 & 0.274 & 0.289 & 0.534 & 1.854 & 0.250 \\ \hline
      & \multicolumn{10}{c}{Standard Errors (iid)}\\
      \hline
      \textsc{mle} & 0.068 & 0.008 & 0.033 & 0.023 & 0.031 & 0.009 & 0.012 & 0.075 & 0.039 & 0.101 & 0.018 \\ 
      \rnr & 0.066 & 0.008 & 0.035 & 0.022 & 0.030 & 0.009 & 0.011 & 0.077 & 0.038 & 0.118 & 0.017 \\ 
      \rqn & 0.065 & 0.008 & 0.034 & 0.021 & 0.030 & 0.009 & 0.011 & 0.073 & 0.038 & 0.115 & 0.017 \\ 
      \textsc{mcmc} & 0.040 & 0.004 & 0.020 & 0.015 & 0.020 & 0.006 & 0.008 & 0.060 & 0.028 & 0.069 & 0.012 \\ \hline
      & \multicolumn{10}{c}{Standard Errors (clustered)}\\ \hline
      \textsc{mle} & 0.094 & 0.015 & 0.039 & 0.031 & 0.030 & 0.012 & 0.016 & 0.073 & 0.068 & 0.103 & 0.029 \\ 
      \rnr & 0.087 & 0.015 & 0.038 & 0.032 & 0.030 & 0.011 & 0.014 & 0.073 & 0.069 & 0.102 & 0.028 \\ 
      \rqn & 0.105 & 0.015 & 0.040 & 0.034 & 0.036 & 0.011 & 0.016 & 0.105 & 0.071 & 0.134 & 0.030 \\
      \hline \hline
      \end{tabular} }
{\footnotesize Legend: \rnr/\rqn\, computed using $\gamma=0.1$, Gaussian (iid) and exponential (clustered) weights.\\ MCMC computed using $\gamma = 0.0034$, Gaussian prior with mean zero, variance 10 for fixed effects, flat uniform prior on other parameters.}
\end{table}

\begin{figure}[H] \caption{Stochastic Gradient Descent} \centering
      \includegraphics[scale=0.55]{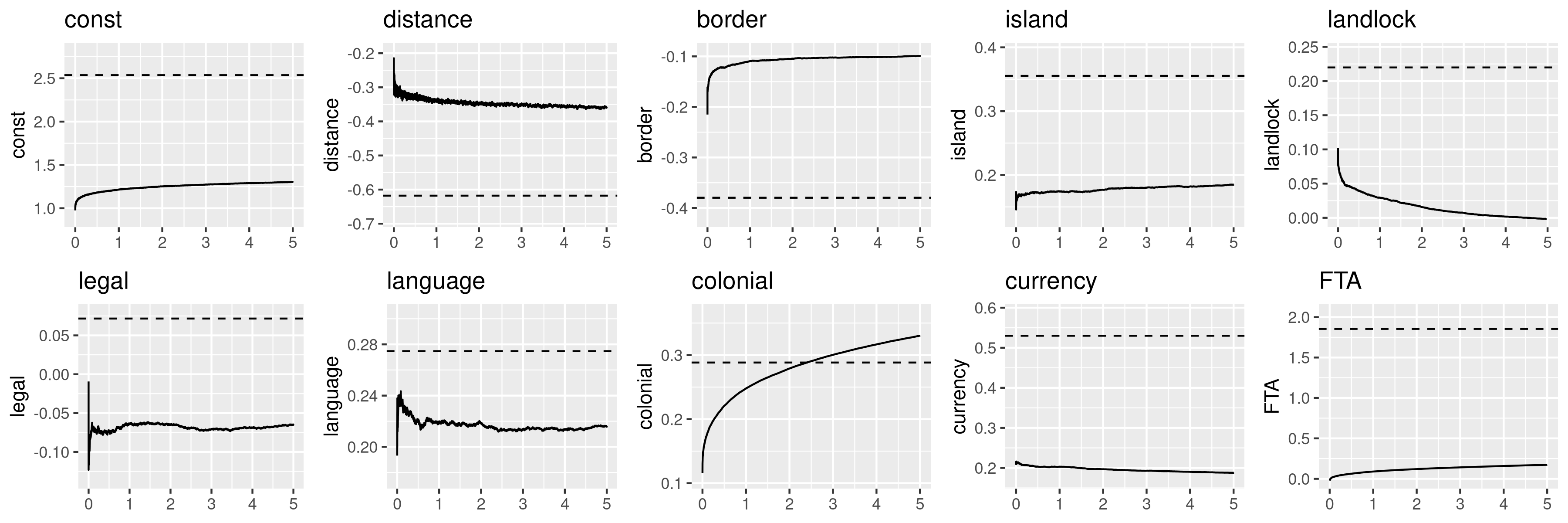}
      {\footnotesize Legend: x-axis $= k$ in millions, $P_b = I_d, m=1,\gamma_k = \gamma_0 k^{-5/8}$, 5 million iterations.}
\end{figure}

\begin{figure}[H] \caption{Infeasible Stochastic Newton-Raphson} \centering
      \includegraphics[scale=0.55]{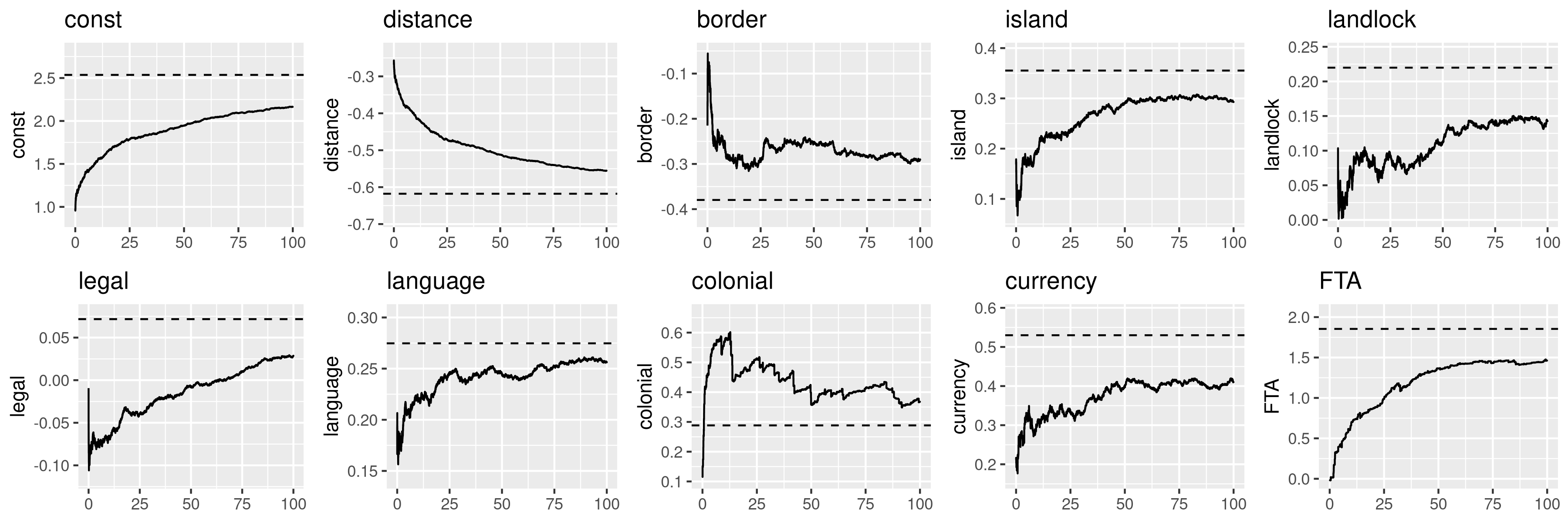}
      {\footnotesize Legend: x-axis $= k$ in thousands, $P_b= [H_n(\hat\theta_n) + \gamma_k I_d]^{-1}$, $m=1,\gamma_k = \gamma_0 k^{-5/8}$, $100$ thousand iterations.}
\end{figure}

\newpage

\subsection{Example 3: NLS Estimation of Transportation Costs} \label{apx:Ex3}
\paragraph{Implementation of \rnr.} The objective function can be written as:
\[ Q_{n}(\theta) =  \frac{1}{n}\sum_{odt} [y_{odt} - f(x_{odt},\theta)]^2, \]
with $\theta = (\alpha,\delta,\beta)$, $\beta$ are the intercept and fixed effect coefficients and $x_{odt}$ consists of $R_t$ and the intercept, fixed effect dummies. $n$ is the total number of observations. The gradient and hessian are:
\begin{align*}
      G_n(\theta) &= -\frac{2}{n}\sum_{odt} \partial_\theta f(x_{odt},\theta)[y_{odt} - f(x_{odt},\theta)],\\ 
      H_n(\theta) &= \frac{2}{n}\sum_{odt} \partial_\theta f(x_{odt},\theta)\partial_\theta f(x_{odt},\theta)^\prime  -\frac{2}{n}\sum_{odt} \partial^2_{\theta,\theta^\prime} f(x_{odt},\theta) [y_{odt} - f(x_{odt},\theta)].
\end{align*}
If $\mathbb{E}\left( y_{odt} - f(x_{odt},\theta^\dagger)| x_{odt} \right)=0$, then $\hat H_n(\hat\theta_n) = \frac{2}{n}\sum_{odt} \partial_\theta f(x_{odt},\hat\theta_n)\partial_\theta f(x_{odt},\hat\theta_n)^\prime$ is a consistent estimator of the population hessian $H(\theta^\dagger) = 2 \mathbb{E}[\partial_\theta f(x_{odt},\theta^\dagger)\partial_\theta f(x_{odt},\theta^\dagger)^\prime]$.

Using this estimator, a cost-effective implementation of \rnr\, with $m=n$ and multiplier weights clustered at the district level $d$ relies on:
\begin{align*}
      G_m^{(b)}(\theta) &= -\frac{2}{n}\sum_{odt} w^{(b)}_d \partial_\theta f(x_{odt},\theta)[y_{odt} - f(x_{odt},\theta)],\\ 
      H_m^{(b)}(\theta) &= \phantom{-}\frac{2}{n}\sum_{odt} w^{(b)}_d \partial_\theta f(x_{odt},\theta)\partial_\theta f(x_{odt},\theta)^\prime,
\end{align*}
where $w^{(b)}_d \sim \mathcal{N}(1,1)$ are iid over districts $d$ and iterations $b$. The appeal of this approach is that $H_m^{(b)}(\theta)$ is symmetric positive semidefinite for all $\theta$, even if there are values for which $H_n(\theta)$ is non-definite. The added penalty ensures that $H_m^{(b)}(\theta) + \lambda \partial^2_{\theta,\theta^\prime}\text{pen}(\theta)/n$ is symmetric positive definite. The derivative $\partial_\theta f(x_{odt},\theta)$ is computed once for each iteration $b$, and then used to evaluate both the gradient and hessian. More specifically $\partial_\alpha f(x_{odt},\theta)$ is computed using finite differences and $\partial_{(\delta,\beta)} f(x_{odt},\theta)$ is simply the vector of linear regressors, i.e. $\log(\text{LCRED}(R_t,\alpha)_{odt})$ and the intercept and fixed effect dummies. The gradient and hessian of the penalty term are then added to the re-weighted gradient and hessian above.  
\newpage
\section{Convergence in Some Non-Convex Settings} \label{sec:discussion}
\subsection{Recovering from a bad start} \label{sec:bad}
Lemmas \ref{lem:cv_non_stochastic} and \ref{lem:cv_stochastic} require the user-chosen learning rate $\gamma$ to satisfy $A(\gamma) < 1$ to get the desired contraction property, which leads to convergence. In practice, feasible values of $\gamma$, depend on both the choice of $P_b$ and the hessian $H_n$. The following discussion will focus on \nr\, in the context of Lemma \ref{lem:cv_non_stochastic} for simplicity. For values of $\gamma$ that are too large, $A(\gamma)<1$ may not hold and convergence may fail. The following illustrates how introducing a certain quadratic penalty in the first few iterations can restore the contraction property. 

\noindent From the starting value $\theta_0$, the firt \nr\, iteration has the form:
$\theta_{1} = \theta_0 - \gamma [H_n(\theta_0)]^{-1}G_n(\theta_0),$ 
which is only well-defined if $H_n(\theta_0)$ is non-singular. Using the mean-value theorem: $\theta_{1} - \hat\theta_n = [I_d -  \gamma [H_n(\theta_0)]^{-1}H_n(\tilde\theta_0)](\theta_0 - \hat\theta_n),$ for some intermediate value $\tilde\theta_0$ between $\theta_0$ and $\hat\theta_n$. A contraction only occurs if $\sigma_{\max}[I_d -  \gamma [H_n(\theta_0)]^{-1}H_n(\tilde\theta_0)] <1$. For $\|\theta_0-\hat\theta_n\|$ small, continuity of the hessian implies $[H_n(\theta_0)]^{-1}H_n(\tilde\theta_0)]$ is close to the identity matrix $I_d$ so that a small $\gamma < 1$ always leads to a contraction closer to $\hat\theta_n$.

Issues can arise for distant starting values $\theta_0$. Suppose $H_n(\theta_0)$ is close to singular, then $[H_n(\theta_0)]^{-1}H_n(\tilde\theta_0)]$ can be very large and $\gamma$ needs to be very small to get a contraction. In that situation, \nr\, iterations are poorly behaved if $\gamma$ is not sufficiently small. As a solution, add a quadratic penalty: $Q_n(\theta) + \frac{\lambda}{2n} \|\theta-\theta_0\|_2^2.$ Penalizing towards $\theta_0$ implies the gradient is unchaged but the hessian becomes $H_n(\theta_0) + \frac{\lambda}{n}I_d$. The first \nr\, iteration is now:
$\theta_{1} = \theta_0 - \gamma [H_n(\theta_0) + \lambda/nI_d]^{-1}G_n(\theta_0),$
which is always well-defined as long as $H_n(\theta_0)$ is positive semi-definite.\footnote{If $H_n(\theta_0)$ has negative eigenvalues it is preferable to use the positive semi-definite $[H_n(\theta_0)^\prime H_n(\theta_0)]^{1/2}$ instead, as illustrated in the next subsection.} For the same intermediate value, we now have:
$\theta_{1} - \hat\theta_n = [I_d -  \gamma [H_n(\theta_0) +\lambda/nI_d]^{-1}H_n(\tilde\theta_0)](\theta_0 - \hat\theta_n),$ which even for $\gamma=1$ is a contraction, provided $\lambda$ is sufficiently large. Note that the update above coincides with an iteration using the trust-region algorithm.\footnote{The trust-region method is an alternative to line-search which can improve on standard \nr\, iterations, see \citet[Ch4.1]{nocedal-wright:06} and \citet[Ch1.4.2]{Bertsekas2016}. } The modification is also similar to the $\tau_b$ used in the quasi-Newton update proposed in the paper. A simple strategy to is then to pick $\gamma \in (0,1)$, and add a large penalty $\lambda$ in the first couple of iterations. Valid inference requires the penalty to be asymptotically negligible, so after the burn-in one can either set $\lambda=0$, use a schedule where $\lambda_b \to 0$ sufficiently fast as $b$ increases, or use a smaller $\lambda_n = o(1)$. Figure \ref{fig:Ex1_lambda} below illustrates the effect of different $\lambda$s during the initial iterations for Example 1 in Section \ref{sec:examples} where $H_n(\theta_0)$ is singular. The baseline, $\lambda=20$, is used in the Monte-Carlo simulations, see Appendix \ref{apx:add_ex1} for implementation details.

\begin{figure} \caption{Illustration: initial iterations for different choice of $\lambda$} \label{fig:Ex1_lambda} \centering
      \includegraphics[scale=0.55]{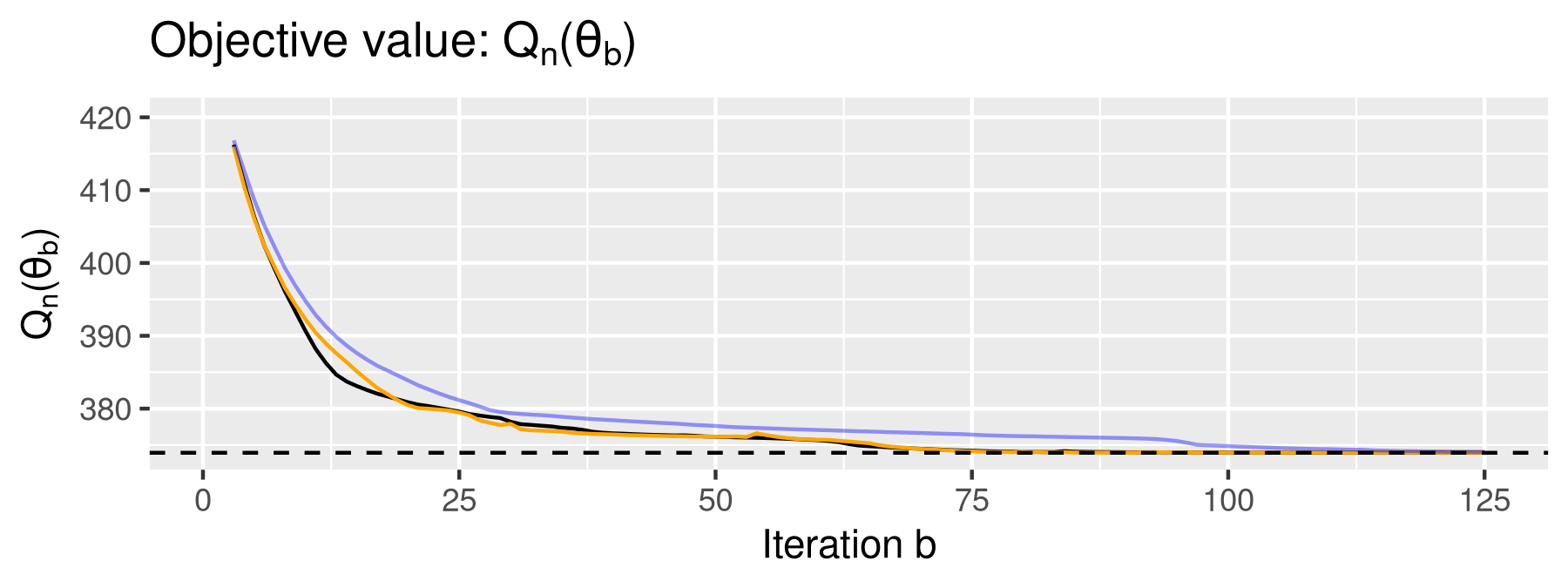}\includegraphics[scale=0.55]{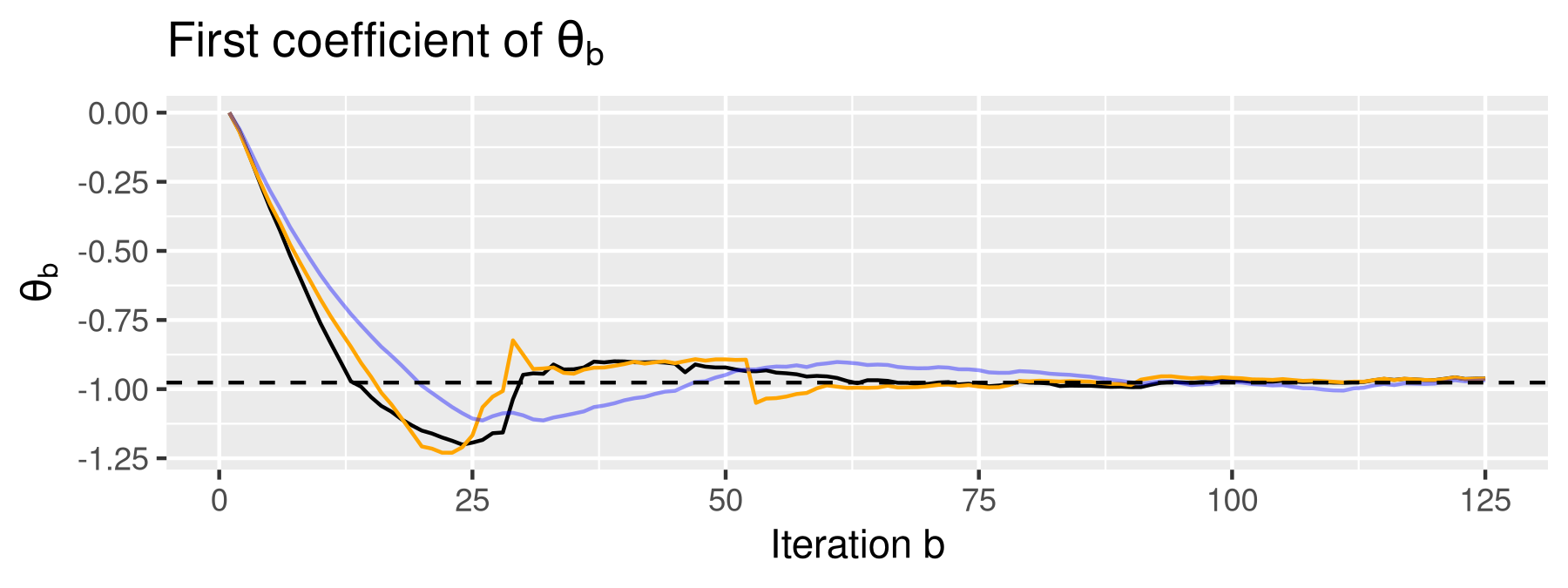}\\
      {\footnotesize Legend: $\gamma=0.1$. Solid lines: black $\lambda=20$, orange $\lambda=0.1$, blue $\lambda=100$. Dashed black line: $Q_n(\hat\theta_n)$.}
\end{figure}

\subsection{Suboptimal solutions} \label{sec:saddle}
An important concern is when a classical optimizer converges to a point that is suboptimal, i.e. is neither a local nor a global optimum. In these cases, the user has to interve and restart an optimization at another value which can be a time-consuming process. The following illustrates, using a stylized  example, the effect of the resampling noise in this scenario.\\
\noindent Consider a point $\theta_{\times}$ for which the Hessian is indefinite, i.e. $H_n(\theta_{\times})$ has both strictly positive and negative eigenvalues.\footnote{See Section 3.4 \citet{nocedal-wright:06} for more details on this issue.} For simplicity, the following assumes that the objective is quadratic: $Q_n(\theta) = \frac{1}{2}(\theta-\theta_{\times})^\prime H_n  (\theta-\theta_{\times})$, with eigendecomposition $H_n = Q \Lambda Q^{\prime}$, $\Lambda = \text{diag}(\lambda_1,\dots,\lambda_{d_\theta})$ where $0 < \lambda_S \leq \lambda_{S-1} \leq \dots \leq \lambda_1$ and $0 > \lambda_{S+1} \geq \dots \geq \lambda_{d_\theta}$. There are $S$ directions with positive curvature, and $d_\theta - S$ directions with negative curvature. Let $q_1,\dots,q_{d_\theta}$ be the corresponding eigenvectors. Without constraints, the global minimum is $Q_n(\theta) = -\infty$ attained at $\theta = +\infty \times q_{s}$, $s \geq S+1$, the optimizer should diverge: $\|\theta_b\| \to \infty$. 

\paragraph{Issues with classical Newton-Raphson:} Consider the Newton-Raphson iteration $\theta_{b+1} = \theta_b - \gamma H_n^{-1}G_n(\theta_b)$, here $G_n(\theta_b) = H_n(\theta_b - \theta_{\times})$ so that: $\theta_{b+1} - \theta_{\times} = (1-\gamma)(\theta_b-\theta_{\times}),$ \nr\, converges exponentially fast to the suboptimal solution $\theta_{\times}$. It is possible to improve the behaviour of \nr\, by using a Hessian modification: $\theta_{b+1} = \theta_b - \gamma (H_n^2)^{-1/2}G_n(\theta_b)$ which enforces $(H_n^2)^{-1/2} = Q (\Lambda^2)^{-1/2} Q^\prime$ positive definite. \nr\, iterations are now such that:
\[ \theta_{b+1} - \theta_{\times} = (I + Q \Gamma Q^\prime)(\theta_b-\theta_{\times}), \quad \Gamma = \text{diag}(\underbrace{-\gamma,\dots,-\gamma}_{S \text{ times}},\underbrace{+\gamma,\dots,+\gamma}_{d_\theta-S \text{ times}}). \]
If $\theta_b-\theta_{\times}$ equals zero or is orthogonal to $q_{S+1},\dots,q_{d_\theta}$, then \nr\, will converge to $\theta_{\times}$ exponentially fast. However, \nr\, will diverge exponentially fast if $\theta_b-\theta_{\times} = q_{S+1}$, for instance, because of the explosive root $1+\gamma >1$ which is outside the unit circle. 

\paragraph{The case of resampled Newton-Raphson:} The main difference between \nr\, and \rnr\, in this example is that \rnr\, diverges exponentially fast with high-probability even when $\theta_0=\theta_{\times}$, which is always problematic for classical optimizers such as \gd, \nr, or \textsc{bfgs}. To simplify notation, assume that $Q=I$ and $G_{m}^{(b)}(\theta) \sim \mathcal{N}\left( G_n(\theta), \Sigma/m \right)$, $\Sigma$ non-singular, and $H_m^{(b)} = H_n$. For \rnr: $\theta_{b+1} = \theta_b - \gamma ([H_m^{(b+1)}]^2)^{-1/2}G_m^{(b+1)}(\theta_b)$, which implies:
\[ \theta_{b+1} - \theta_\times = (I + \Gamma)(\theta_b-\theta_{\times}) + \Gamma Z_m^{(b+1)}, \quad Z_m^{(b+1)} = G_{m}^{(b+1)}(\theta_b)-G_n(\theta_b),  \]
where $\Gamma$ is the same as above. Notice that $(\theta_{b+1} - \theta_\times)$ follows a VAR(1) process with transition matrix $I+\Gamma$ which has $d_\theta-S$ explosive roots outside the unit circle, all equal to $1+\gamma$. Suppose $\theta_0=\theta_{\times}$ which is problematic for \nr. Let $\theta_{j,b+1}$ be the s-th row of $\theta_{b+1}$ and $Z_{m,j}^{(b-j)}$ the $j$-th row of $Z_m^{(b-j)}$, denote its variance by $\sigma^2_{j}>0$. For $j \geq S+1$ we have the AR(1) representation \[ \theta_{j,b+1}-\theta_{j,\times} = \gamma(1+\gamma)^b \sum_{j=0}^b (1+\gamma)^{-j} Z_{m,j}^{(j)} \sim \mathcal{N}\left(0,\gamma^2(1+\gamma)^{2b} \frac{1-(1+\gamma)^{-2b-2}}{1-(1+\gamma)^{-2}}\sigma^2_{j} \right).\]
Using Gaussianity and the Paley-Zygmund inequality we have the probability bound:
\[ \mathbb{P}^\star \left(  \frac{|\theta_{j,b+1}-\theta_{j,\times}|}{\gamma (1+\gamma)^b}\frac{[1-(1+\gamma)^{-2}]^{1/2}}{\sqrt{2/\pi} \sigma_j[1-(1+\gamma)^{-2b-2}]^{1/2}} \geq \varepsilon\right) \geq  \frac{2(1-\varepsilon)^2}{\pi},\]
for any $\varepsilon \in [0,1]$. This implies that, with high probability, $\theta_{b+1}$ diverges exponentially fast.

\paragraph{Numerical Illustration:} To illustrate the above numerically, consider a two-dimensional optimization problem where $\theta_{\times} = (0,0)$. We will set $\theta_0 = q_1 + c q_2$ and vary $c$ to compare the performance of R's \textsc{bfgs} optimizer, which relies on a line-search, with \nr\, and \rnr\, which use a fixed learning rate $\gamma=0.1$ and the hessian modification. For $c=0$, \nr\, will converge to the saddle point and $c\neq 0$ implies that \nr\, will eventually diverge. Table \ref{tab:saddle} reports the results for a range of $c$. Note that \textsc{bfgs} converges to $\theta_{\times}$ even when $c \neq 0$. This is because, the \textsc{bfgs} update does not guarantee positive definiteness of the conditioning matrix so that it behaves similarly to \nr\, without the hessian modification. \nr\, diverges for $c\neq 0$ but slowly for $c<1$. As predicted, \rnr\, diverges quickly for all $c$.

\begin{table}[ht] \caption{Behaviour of \textsc{bfgs}, \nr, and \rnr\, around a saddle point} \label{tab:saddle}
      \centering
      \begin{tabular}{l|ccc|ccc|ccc}
        \hline \hline
      & \multicolumn{3}{c|}{\textsc{bfgs}} & \multicolumn{3}{c|}{\nr} & \multicolumn{3}{c}{\rnr}\\ \hline
       $c$ & $\theta_1$ & $\theta_2$ & $Q_n(\theta)$ & $\theta_1$ & $\theta_2$ & $Q_n(\theta)$ & $\theta_1$ & $\theta_2$ & $Q_n(\theta)$ \\ 
        \hline
      0.0 & 0.00 & 0.00 & 0.00 & 0.01 & -0.01 & 0.00 & -0.01 & -90.42 & -91.35 \\ 
      0.1 & 0.05 & 0.05 & -0.00 & 4.78 & 5.30 & -0.25 & 5.30 & -85.65 & -81.95 \\ 
      0.5 & 0.23 & 0.25 & -0.00 & 23.88 & 26.53 & -6.37 & 26.53 & -66.54 & -49.47 \\ 
      1.0 & 0.45 & 0.50 & -0.00 & 47.76 & 53.06 & -25.48 & 53.06 & -42.66 & -20.33 \\ 
      5.0 & 2.26 & 2.51 & -0.06 & 238.80 & 265.32 & -637.09 & 265.32 & 148.37 & -245.96 \\ 
         \hline \hline
      \end{tabular}\\
      {\footnotesize Legend: \textsc{bfgs}: output of R's optimizer (optim). \nr\, and \rnr: output after $50$ iterations.\\ True minimum is $Q_n = -\infty$. Saddle point is $\theta_{\times} = (0,0)$.}
\end{table}

\end{appendices}
\end{document}